\documentclass[a4paper,twocolumn,10pt,accepted=2023-11-12]{quantumarticle}
\pdfoutput=1

\usepackage[utf8]{inputenc}
\usepackage[english]{babel}
\usepackage[T1]{fontenc}

\usepackage{graphicx}
\usepackage{dcolumn}
\usepackage{braket}
\usepackage{setspace}
\usepackage{bm}
\usepackage{amsmath}
\usepackage{amsthm}
\usepackage{mathtools}
\usepackage{csquotes}
\usepackage{mathrsfs}
\usepackage[svgnames]{xcolor}

\usepackage{hyperref}
\usepackage{cleveref}

\allowdisplaybreaks 

\usepackage{amsfonts}
\usepackage{amssymb}
\usepackage{cancel}

\newtheorem{lemma}{Lemma}

\newtheorem{theorem}{Theorem}
\newtheorem{remark}{Remark}

\usepackage{tikz}
\usetikzlibrary{quantikz}

\usepackage{xcolor}

\usepackage{subcaption} 

\usepackage{caption}
\usepackage{indentfirst} 

\usepackage[numbers,sort&compress]{natbib} 

\newcommand*\mathinhead[2]{\texorpdfstring{$\boldsymbol{#1}$}{#2}} 

\begin{document}

\title{Recovery With Incomplete Knowledge: \newline Fundamental Bounds on Real-Time Quantum Memories}
\author{Arshag Danageozian}
\email{arshag.danageozian@gmail.com}
\affiliation{Hearne  Institute  for  Theoretical  Physics,  Department  of  Physics  and  Astronomy,  Louisiana  State  University,  Baton  Rouge,  Louisiana  70803,  USA}
\orcid{0000-0003-0044-9951}

\begin{abstract}
    The recovery of fragile quantum states from decoherence is the basis of building a quantum memory, with applications ranging from quantum communications to quantum computing. Many recovery techniques, such as quantum error correction, rely on the \textit{apriori} knowledge of the environment noise parameters to achieve their best performance. However, such parameters are likely to drift in time in the context of implementing long-time quantum memories. This necessitates using a ``spectator'' system, which estimates the noise parameter in real-time, then feedforwards the outcome to the recovery protocol as a classical side-information. The memory qubits and the spectator system hence comprise the building blocks for a real-time (i.e. drift-adapting) quantum memory. In this article, I consider spectator-based (incomplete knowledge) recovery protocols as a real-time parameter estimation problem (generally with nuisance parameters present), followed by the application of the ``best-guess'' recovery map to the memory qubits, as informed by the estimation outcome. I present information-theoretic and metrological bounds on the performance of this protocol, quantified by the diamond distance between the ``best-guess'' recovery and optimal recovery outcomes, thereby identifying the cost of adaptation in real-time quantum memories. Finally, I provide fundamental bounds for multi-cycle recovery in the form of recurrence inequalities. The latter suggests that incomplete knowledge of the noise could be an advantage, as errors from various cycles can cohere. These results are illustrated for the approximate [4,1] code of the amplitude-damping channel and relations to various fields are discussed.
\end{abstract}

\maketitle
\tableofcontents

\section{Introduction}
Quantum memories comprise an important component of current and future quantum technologies. Their use ranges from quantum communications and networks \cite{gundougan2021proposal, wallnofer2022simulating, gundougan2023time} to sensing \cite{sidhu2021tight}, and even computation. This wide range of relevance stems from the fact that a quantum memory preserves a quantum system's (often fragile) state from decoherence, which encodes the desired quantum information. 

Depending on the physical implementation of the quantum memory, current coherence times range from milliseconds to minutes \cite{wang2017single}. However, various quantum technologies may require even longer coherence times \cite{stas2022robust}. Two of the most common techniques implemented in a quantum memory are quantum error correction (QEC) \cite{schumacher1996quantum, schumacher1996sending, knill1997theory} and dynamical decoupling \cite{viola1998dynamical, viola1999dynamical}. These techniques generally benefit from the \textit{apriori} knowledge of the noise surrounding the system of interest. For example, channel-adaptation techniques in QEC \cite{fletcher2008channel} have been shown to outperform general QEC codes, as they are given additional knowledge of the environment noise \cite{leung1997approximate, fletcher2007optimum}. Such techniques rely on some physical model of the (noisy) implementation medium of the quantum memory. The noise model is partially built upon physical assumptions (e.g. in the choice of the Hamiltonians) and partially upon phenomenology. Hence, the former gives a physically motivated family of quantum dynamics $\{\mathcal{N}_{\theta}\}_{\theta \in \Theta}$ for the quantum state of the memory qubits \cite{breuer2002theory}, and the latter determines the value of the noise parameter $\theta$ such that the dynamics $\mathcal{N}_{\theta}$ fits the observed data the best.

Although very powerful, a shortcoming of this approach is that the environment noise parameter is generally time-varying. This has been studied most extensively for superconducting qubits \cite{muller2015interacting, klimov2018fluctuations, etxezarreta2021time, dasgupta2020characterizing, dasgupta2021stability}. Hence, real-time techniques to track the change (drift) of the noise are necessary, assuming we want to operate quantum memories for times larger than the characteristic times of the drift. 

Indeed, efforts have been made towards designing ``spectator'' systems that aid in detecting and tracking such changes \cite{bonato2017adaptive, cortez2017rapid, proctor2020detecting, majumder2020real, gupta2020integration, youssry_2023, danageozian2022noisy, song2023optimized, tonekaboni2023greedy, singh2023mid}. Being subject to the same physical environment, the goal of the spectator system is to perform real-time quantum sensing of the noise parameter. The estimate is then used as a classical side-information in various recovery protocols. The physical requirements of spectator systems are two-fold: $(1)$ proximity to the memory qubits, such that the spatial dependence of the noise parameter can be neglected, and $(2)$ exhibiting faster dynamics than the memory qubits. The latter is necessary if the feedback information is to be useful for recovery. We showcase the functionality of the spectator system within a quantum memory using Figs.~\ref{fig:circuit1} and ~\ref{fig:circuit2}. Note that, since the memory and spectator systems are generally different physical systems with different couplings to the same environment, their dynamics will generally be described by different quantum channels with the same noise parameter, i.e. $\mathcal{N}_{\theta}$ and $\mathcal{M}_{\theta}$, respectively. More precisely, a spectator-based recovery protocol is comprised of the following characteristic stages, as shown in Fig.~\ref{fig:protocol}:
\begin{enumerate}
	\item \textit{Individual state preparation} of the quantum memory and the spectator system.
	\item \textit{Free evolution} of the joint memory-spectator system under the action of the shared environment, with generally unknown noise parameters.
	\item \textit{Quantum parameter estimation} of the environment noise parameters, using the spectator system as a real-time quantum sensor (i.e. a probe). 
	\item \textit{Post-processing} of the measurement outcomes to extract the value of the locally unbiased estimator, and use it to construct the best-guess recovery map.
	\item \textit{Recovery} of the original state of the quantum memory by applying the best-guess recovery map.
	\item \textit{Recycling of spectator state}, which prepares it for the next recovery cycle.
\end{enumerate}

There exist systems that satisfy the physical properties of a spectator system. For example, nitrogen-vacancy (NV) centers in diamond, which were used to prove the first loophole-free Bell inequality violation \cite{hensen2015loophole}, provide both a spectator qutrit and a memory qubit. Namely, the nuclear spin degree of freedom of its $^{14}N$ or $^{15}N$ atom comprises the memory qubit, whereas a nearly-closed three-level $\Lambda$ system \cite{togan2010quantum, golter2014optically}, optically selected out of the electronic degrees of freedom of the NV center, comprises the spectator system \cite{danageozian2022noisy, wu2021continuous, turner2022real}. The two-time separation between the pure dephasing times of the memory qubit ($T_{\varphi}^{\operatorname{memo}} \sim 100\mu$s) and the spectator qutrit ($T_{\varphi}^{\operatorname{spec}} \sim 100$ns \cite{lekavicius2017transfer}) is necessary to simultaneously yields (i) a metrologically useful spectator state $\mathcal{M}_{\theta}(\psi)$ for parameter estimation, and (ii) a relatively small noise parameter value of the memory dynamics $\mathcal{N}_{\theta}$ (and hence a generally higher recovery fidelity), for relevant times $t$ of the spectator dynamics. The latter is seen from the implicit dependence of the noise parameter $\theta$ on time \cite{nielsen2002quantum}: $\theta=1-\exp{(-t/T_{\varphi}^{\operatorname{memo}})}=1-[\exp(-t/T_{\varphi}^{\operatorname{spec}})]^{T_{\varphi}^{\operatorname{spec}}/T_{\varphi}^{\operatorname{memo}}}<<1-\exp{(-t/T_{\varphi}^{\operatorname{spec}})}\equiv \theta_{\operatorname{eff}}$, where $\mathcal{M}_{\theta}\equiv \mathcal{N}_{\theta_{\operatorname{eff}}}$. 

Although spectator systems are a promising building block for real-time (i.e. drift-adapting) quantum memories, we expect fundamental limitations to manifest nonetheless. This is based on the following physical intuition: In real-time, the spectator system's goal is to perform a quantum estimation of the environment noise parameter $\theta$. However, due to the quantum Cram\'er-Rao bound (QCRB) \cite{braunstein1994statistical, braunstein1996generalized}, any locally unbiased estimate $\hat{\theta}$ of the noise parameter $\theta$ will have a non-zero variance. Namely, $\operatorname{Var}(\hat{\theta}) \ge 1/\textsf{I}_{\operatorname{QF}}(\mathcal{M}_{\theta}(\psi))$, where $\textsf{I}_{\operatorname{QF}}(\mathcal{M}_{\theta}(\psi))$ is the quantum Fisher information (QFI) of the family of parametric states $\{\mathcal{M}_{\theta}(\psi)\}_{\theta \in \Theta}$ describing the spectator dynamics (see Fig.~\ref{fig:circuit2}). For a given setup, this fundamental uncertainty in the estimate $\hat{\theta}$ will propagate within the overall protocol and manifest itself as a fundamental limitation of the specific recovery technique.

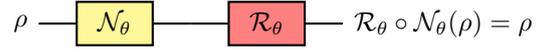
\begin{figure}[t]
    \begin{quantikz}
        \lstick{$\rho$} & \gate[style={fill=yellow!50}, wires=1][1cm]{\mathcal{N}_{\theta}} & \qw &  \gate[style={fill=red!50}, wires=1][1cm]{\mathcal{R}_{\theta}} & \qw & \rstick{\hspace{-0.5cm}$\mathcal{R}_{\theta}\circ \mathcal{N}_{\theta}(\rho) = \rho$} \\
    \end{quantikz}
    \caption{Recovery with perfect knowledge (time flows from left to right). The quantum memory is prepared in the quantum state $\rho$. The recovery channel $\mathcal{R}_{\theta}$ is implemented using perfect knowledge of the noise parameter $\theta \in \Theta$ of the environment noise $\mathcal{N}_{\theta}$.}
    \label{fig:circuit1}
\end{figure}

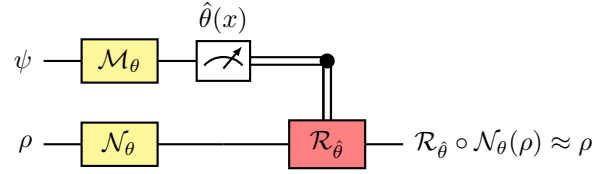
\begin{figure}[t]
    \begin{quantikz}
        & \lstick{$\psi$} & \gate[style={fill=yellow!50},wires=1][1cm]{\mathcal{M}_{\theta}}& \meter{$\hat{\theta}(x)$} & \cwbend{1}\\
        & \lstick{$\rho$} & \gate[style={fill=yellow!50}, wires=1][1cm]{\mathcal{N}_{\theta}} & \qw &  \gate[style={fill=red!50}, wires=1][1cm]{\mathcal{R}_{\hat{\theta}}} \vcw{-1} & \qw & \rstick{\hspace{-0.5cm}$\mathcal{R}_{\hat{\theta}}\circ \mathcal{N}_{\theta}(\rho) \approx \rho$} \\
    \end{quantikz}
    \caption{Recovery with limited knowledge (time flows from left to right). The quantum memory (second register) and spectator (first register) systems are prepared in the quantum states $\rho$ and $\psi$, respectively. The recovery channel $\mathcal{R}_{\hat{\theta}}$ is implemented based on the spectator's best estimate $\hat{\theta}$ of the noise parameter $\theta \in \Theta$ of the environment noise $\mathcal{N}_{\theta}$. The estimate is informed by the measurement outcome $x$ of the spectator observable $X$, following the spectator dynamics $\mathcal{M}_{\theta}$.}
    \label{fig:circuit2}
\end{figure}

\begin{figure*}
	\centering
	\begin{subfigure}[][][t]{0.45\textwidth}
		\includegraphics[width=\textwidth]{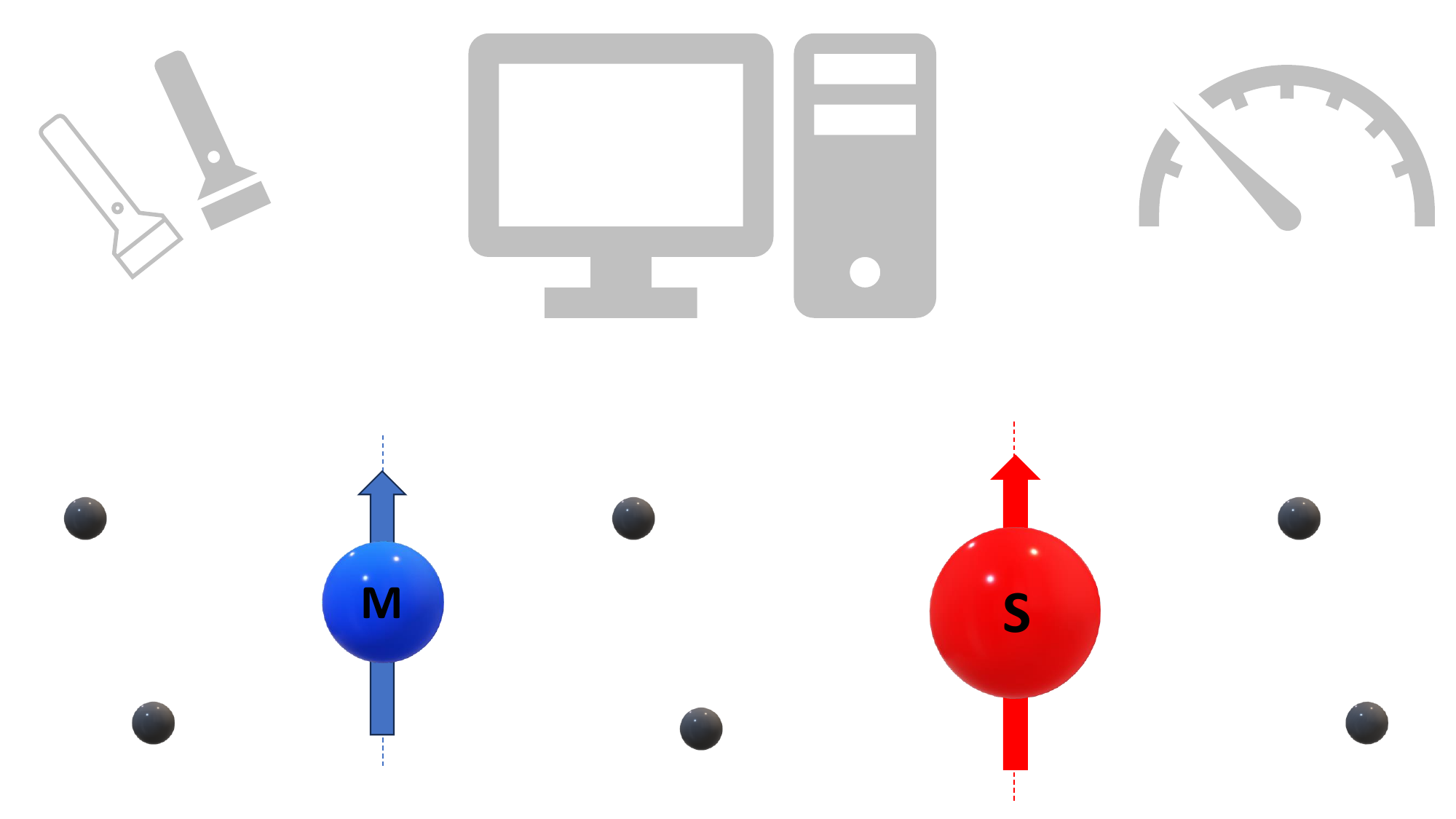}
		\caption{\textit{State Preparation}: A desired state of the quantum memory $M$, and some metrologically useful state of the spectator system $S$, are prepared.}
		\label{fig:first}
	\end{subfigure}
	\hfill
	\begin{subfigure}[][][t]{0.45\textwidth}
		\includegraphics[width=\textwidth]{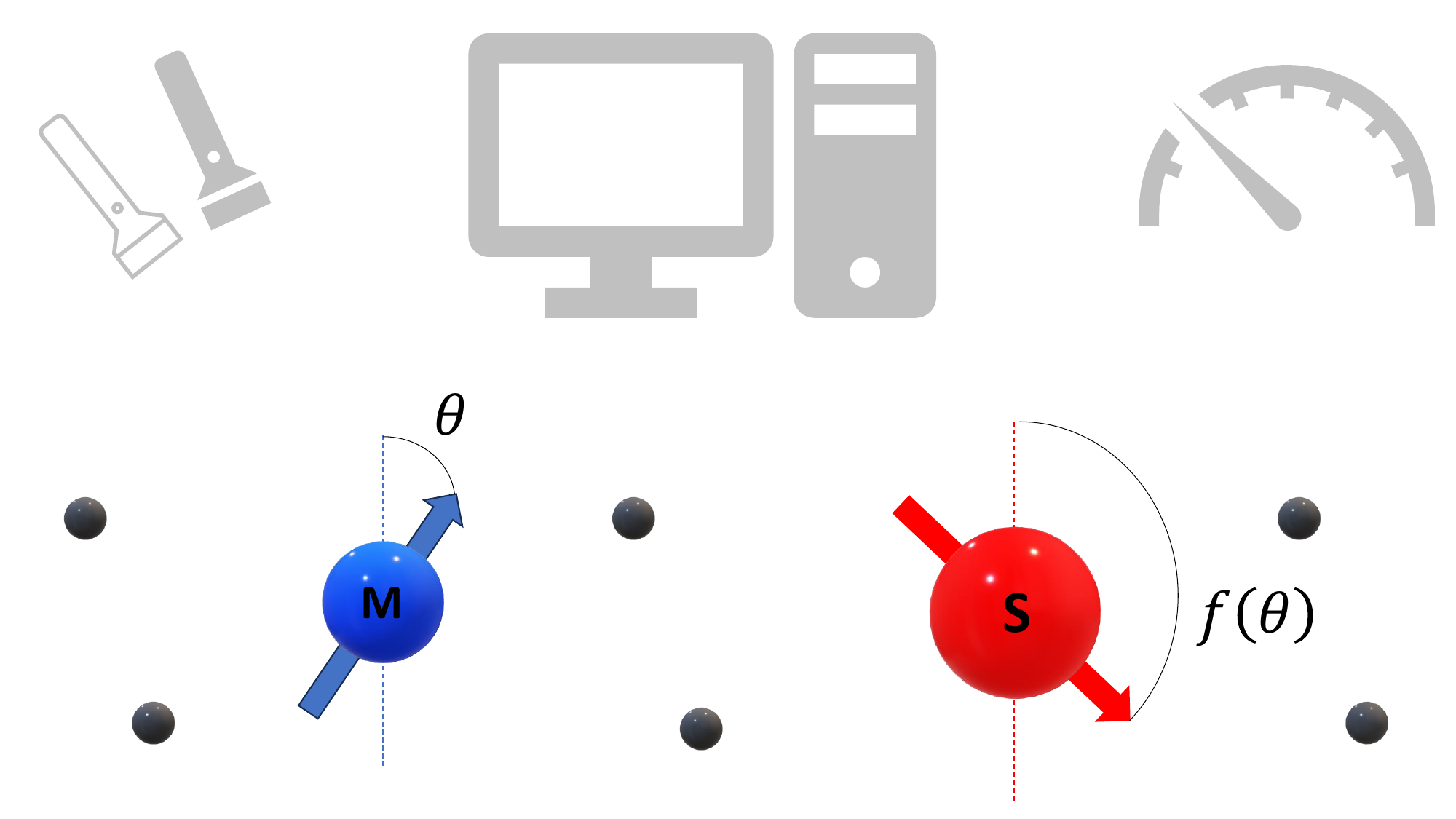}
		\caption{\textit{Free Evolution}: Due to the interaction with their joint environment, the states of the memory and the spectator evolve. The evolution of $M$ is parameterized by $\theta$, and the evolution of $S$ by some function $f(\theta)$ of $\theta$.}
		\label{fig:second}
	\end{subfigure}
	\hfill
	\begin{subfigure}[][][c]{0.45\textwidth}
		\includegraphics[width=\textwidth]{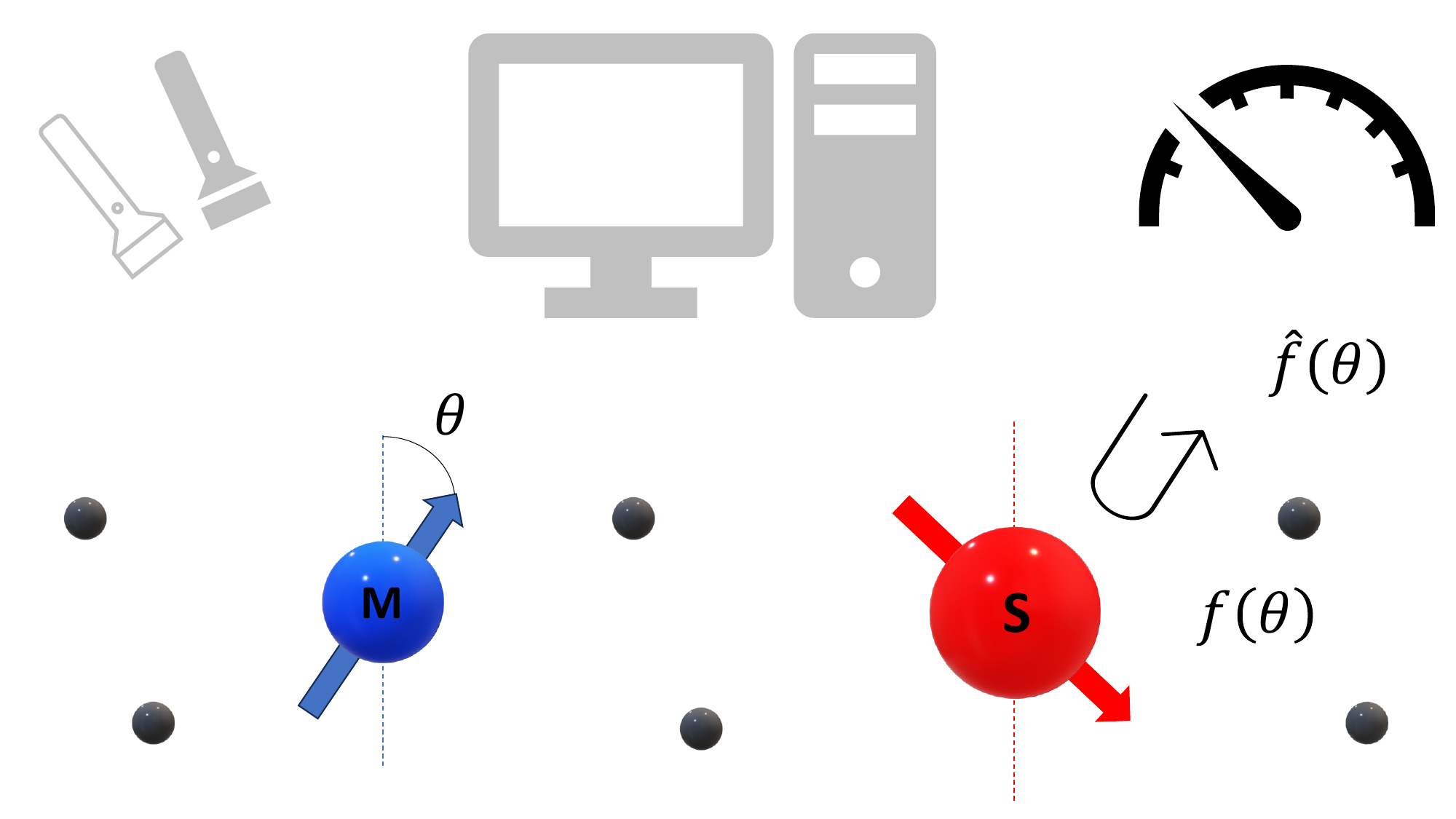}
		\caption{\textit{Quantum Parameter Estimation}: The spectator system is used as a real-time quantum sensor (probe) to find the best estimate $\hat{\theta}$ of the noise parameter $\theta$.}
		\label{fig:third}
	\end{subfigure}
	\hfill
	\begin{subfigure}[][][c]{0.45\textwidth}
		\includegraphics[width=\textwidth]{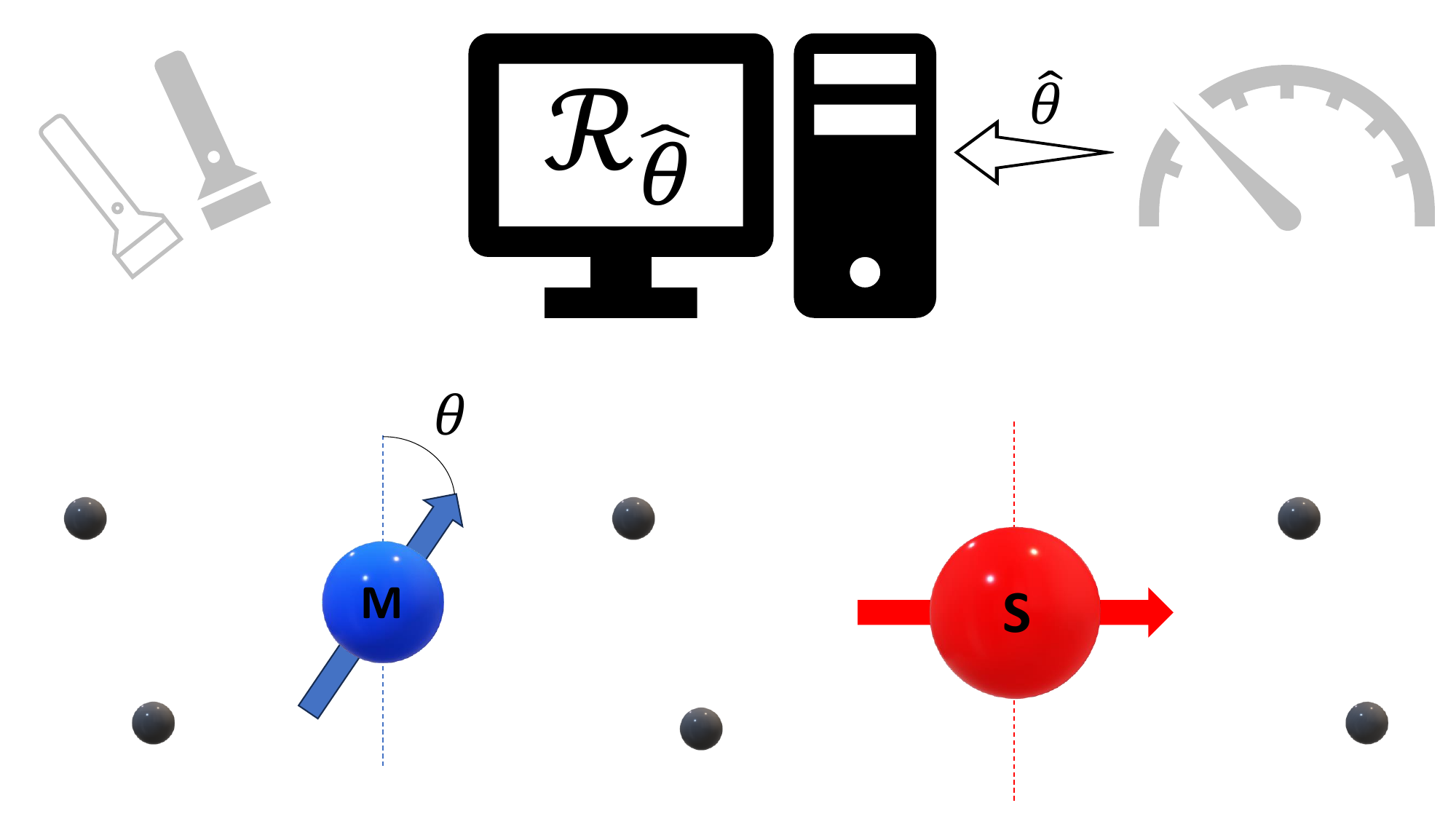}
		\caption{\textit{Post-Processing}: Where the best estimate $\hat{\theta}$ is used to obtain a ``best-guess'' recovery map $\mathcal{R}_{\hat{\theta}}$, which is optimal, given the incomplete knowledge of the true value of the noise parameter $\theta$. This map is generally different from the truly optimal recovery map $\mathcal{R}_{\theta}$, corresponding to the parameterized dynamics of the quantum memory.}
		\label{fig:fourth}
	\end{subfigure}
    \hfill
	\begin{subfigure}[][][b]{0.45\textwidth}
		\includegraphics[width=\textwidth]{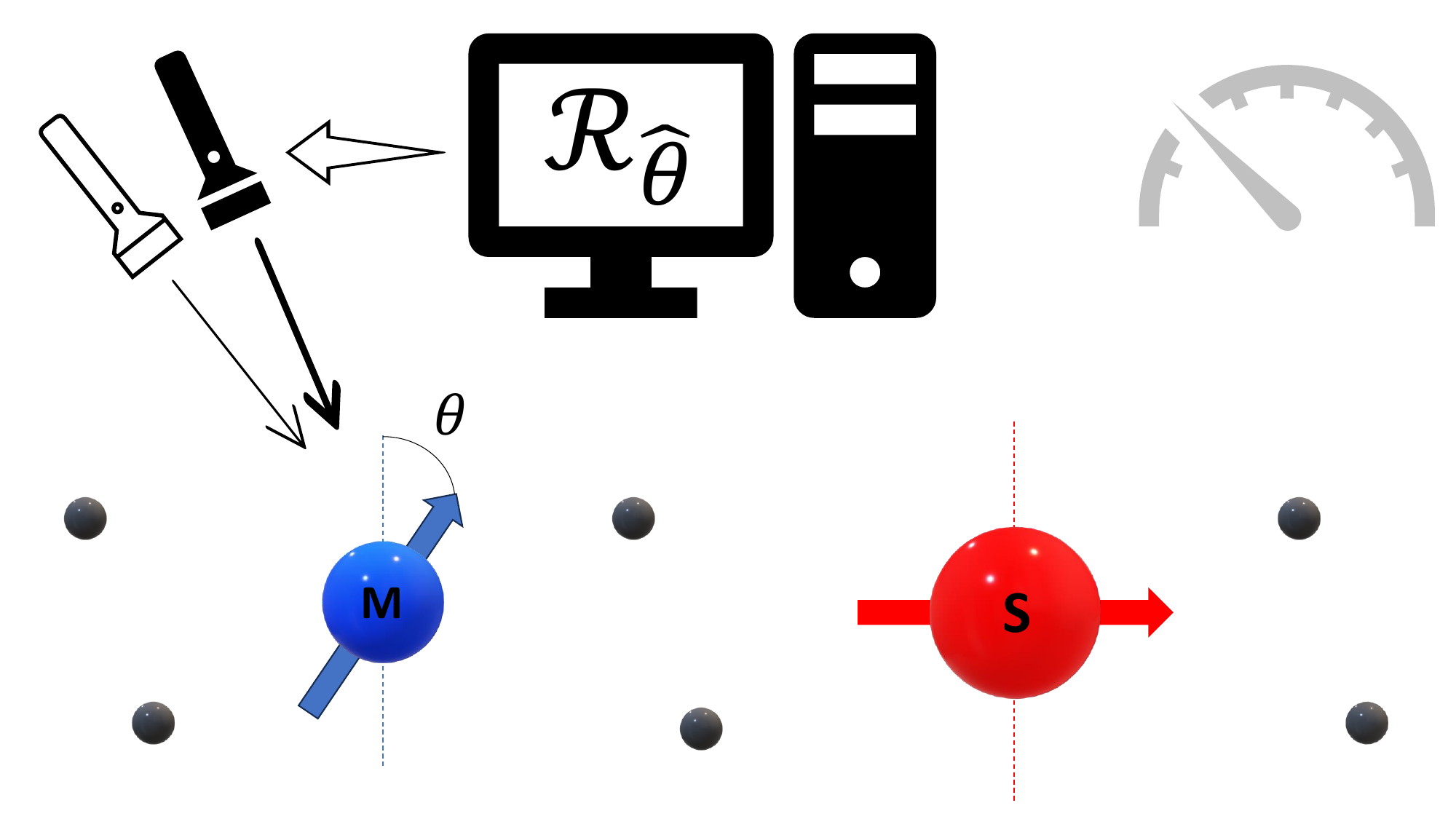}
		\caption{\textit{Best-Guess Recovery}: Where the ``best-guess'' recovery map $\mathcal{R}_{\hat{\theta}}$ is applied to the quantum memory, to recover (perfectly or approximately) the quantum information encoded on its initial state.}
		\label{fig:fifth}
	\end{subfigure}
	\hfill
	\begin{subfigure}[][][b]{0.45\textwidth}
		\includegraphics[width=\textwidth]{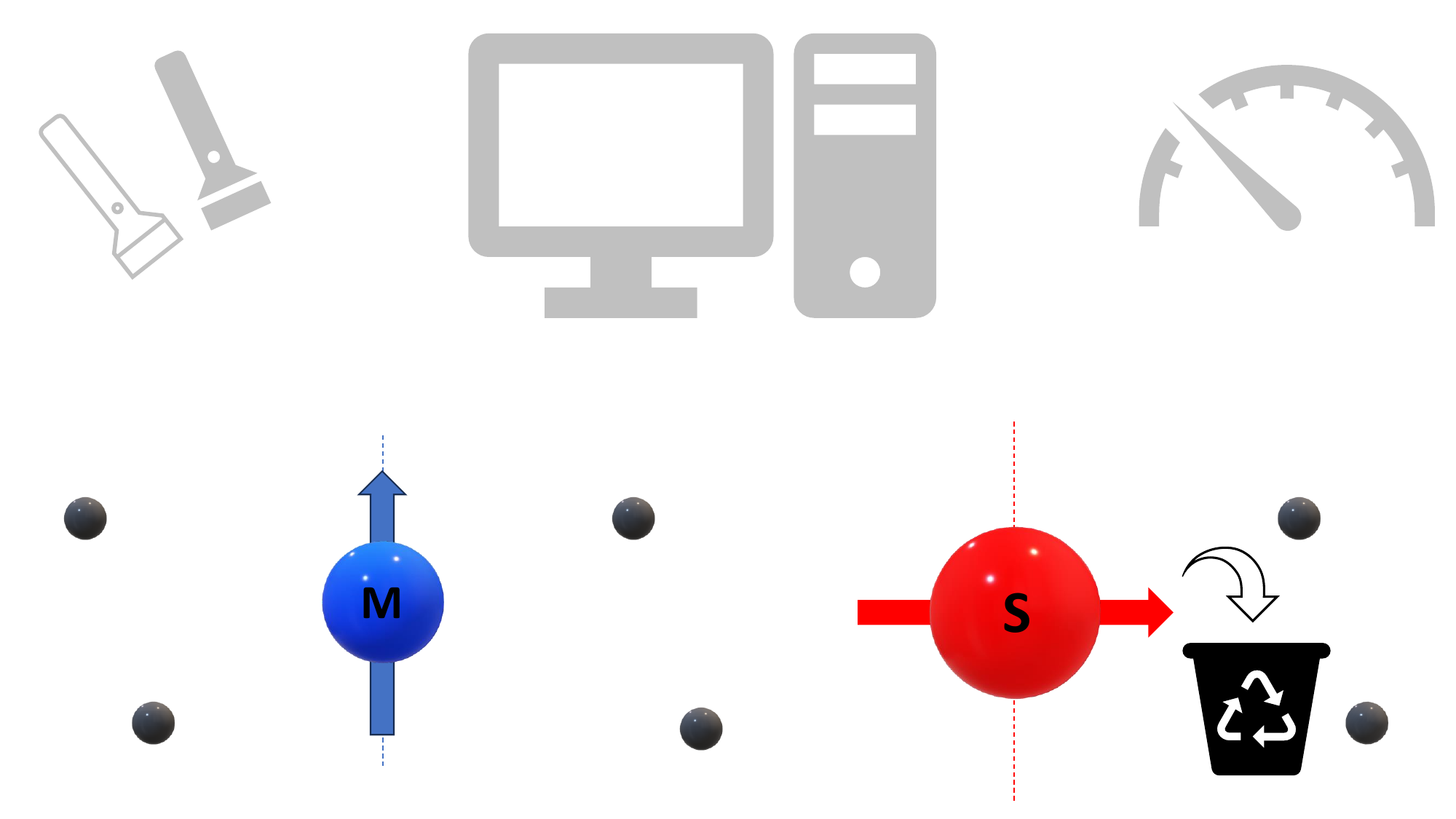}
		\caption{\textit{Spectator System Recycling}: The final step is to recycle the state of the spectator system to prepare it for the next recovery cycle.}
		\label{fig:sixth}
	\end{subfigure}
	\caption{Cartoon description of spectator-based recovery protocols (temporal order corresponds to the alphabetical order of the subfigures). The letters ``$M$'' and ``$S$'' stand for quantum memory and spectator system, respectively. The black dots represent the environment spins that contribute to the noise. This protocol combines two important disciplines of quantum information theory: quantum parameter estimation and recovery of quantum information.}
	\label{fig:protocol}
\end{figure*}

In this article, I formalize the above intuition by proving information-theoretic and metrological bounds on the performance of spectator-based recovery protocols for adaptive quantum memories, considering QEC as an example. The main results of the article are summarized below
\begin{enumerate}
    \item Derivation of a lower bound for the diamond distance between any two quantum channels (Lemma~\ref{co:examples}). This generalizes a lower bound for the diamond distance between a quantum channel and the identity channel in \cite{ouyang2023approximate}, a qudit depolarizing channel and the identity channel in \cite{pirandola2019fundamental}, as well as the analytic formula in \cite{benenti2010computing} for the diamond distance between two-qubit depolarizing channels. A similar lower bound is shown for generalized distinguishability measures, such as entropy or fidelity-based distinguishability measures (e.g. quantum relative entropy or Bures distance, respectively) is found in Theorem~\ref{th:main_th} of Appendix~\ref{apx:gen_dist_meas}.
    \item A general formulation of the spectator-based recovery protocol as two consecutive, but complementary, tasks in a real-time quantum memory (Section~\ref{sec:general_spec_based_recov}): (i) a multi-parameter quantum estimation in the presence of nuisance parameters, and (ii) recovery of quantum information using the ``best-guess'' recovery map, as informed by the estimation outcome.
    \item Derivation of Information-theoretic costs of adaptation in spectator-based recovery protocols. This is shown both for finite (Theorem~\ref{th:finite_var}) and small (Theorem~\ref{th:spec_QFI}) estimation errors of the noise parameters. The latter yields a metrological lower bound, in terms of the quantum Fisher information of the spectator dynamics. The adaptation cost is illustrated for the [4,1] code of the amplitude-damping channel (Figs.~\ref{fig:[4,1]_circuit} and \ref{fig:[4,1]fund}), by comparing the performance of the spectator-based recovery protocol with the corresponding optimal recovery protocol \cite{fletcher2008channel, fletcher2007optimum} when no adaptation is required.
    \item Reformulation of an upper bound for the entanglement fidelity of concatenated quantum channels (building upon a theorem of \cite{carignan2019bounding}) in the form of recurrence inequalities for multi-cycle recovery protocols (Lemma~\ref{th:rec_ent_fid}). These bounds are growing in relevance, as multiple QEC rounds have been demonstrated in practice \cite{cramer2016repeated}. It is shown that spectator-based recovery protocols, under conditions of varying noise, could outperform optimal recovery protocols with constant noise (Theorem~\ref{th:spec_multi}). This is exclusively a multi-cycle phenomenon, where errors from different cycle numbers can cohere \cite{carignan2019bounding}, including errors from incomplete knowledge of the noise parameter. Finally, this is illustrated for the [4,1] code of the amplitude-damping channel (Fig.~\ref{fig:multicycle}). Similar error coherence effects have been seen recently in the context of quantum state transfer in quantum networks, between gate and readout errors \cite{thotakura2022quantum}.
\end{enumerate}

\section{Preliminaries}
\subsection{Quantum States and Channels}
Let $\mathcal{H}$ denote a Hilbert space, and $\mathcal{L(H)}$ be the set of bounded linear operators acting on~$\mathcal{H}$. Denote by $\mathcal{L_{+}(H)}$ the subset of positive semi-definite operators of $\mathcal{L}(\mathcal{H})$. We define the Hilbert-Schmidt inner product between two linear operators $A$, $B \in \mathcal{L}(\mathcal{H})$ to be $\langle A, B \rangle \coloneqq \mathrm{Tr}[A^{\dagger}B]$. The state of a physical system is described by a density matrix $\rho \in \mathcal{D}(\mathcal{H})$, where $\mathcal{D}(\mathcal{H})$ is the subset of positive semi-definite linear operators $\mathcal{L}_{+}(\mathcal{H})$ that have a unit trace. We denote the dimensions of a Hilbert space $\mathcal{H}$ by $d \coloneqq \text{dim}\mathcal{H}$.

A linear map from $\mathcal{L}(\mathcal{H}^{A})$ to  $\mathcal{L}(\mathcal{H}^{B})$ is denoted by $\mathcal{Q}^{A\rightarrow B}:\mathcal{L}(\mathcal{H}^{A})\rightarrow \mathcal{L}(\mathcal{H}^{B})$. We say that a linear map is positive if  $\mathcal{Q}^{A\rightarrow B}(L_{A}) \in \mathcal{L}_{+}(\mathcal{H}^{B})$ for all $L_{A} \in \mathcal{L}_{+}(\mathcal{H}^{A})$, and trace preserving (TP) if $\mathrm{Tr}[\mathcal{Q}^{A\rightarrow B}(L_{A})]=\mathrm{Tr}[L_{A}]$ for all   $L_{A} \in \mathcal{L}(\mathcal{H}^{A})$. A positive linear map $\mathcal{Q}^{A\rightarrow B}$ is called completely positive (CP) if for every Hilbert space $\mathcal{H}^{R}$, the map $\textsf{id}^{R}\otimes \mathcal{Q}^{A\rightarrow B}$ is positive, where $\textsf{id}^{R}$ is the identity map acting on $\mathcal{L}(\mathcal{H}^{R})$. For any $\mathcal{Q}^{A \rightarrow B}$ and $\mathcal{S}^{B \rightarrow C}$ linear maps, their composition is defined to be $(\mathcal{S}\circ \mathcal{Q})^{A \rightarrow C}(L_{A}) \coloneqq \mathcal{S}(\mathcal{Q}(L_{A}))$, for all  $L_{A} \in \mathcal{L}(\mathcal{H}^{A})$. We define the adjoint map $\left( \mathcal{Q}^{A\rightarrow B} \right)^{\dagger}$ of a linear map $\mathcal{Q}^{A\rightarrow B}$ with respect to the Hilbert-Schmidt inner product as $\langle \mathcal{Q}^{A\rightarrow B}(N_{A}), M_{B}\rangle=\langle N_{A}, (\mathcal{Q}^{A\rightarrow B})^{\dagger}(M_{B})\rangle$ for all $N_{A} \in \mathcal{L}(\mathcal{H}^{A})$ and $M_{B} \in \mathcal{L}(\mathcal{H}^{B})$. More explicitly, if $\{Q_{i}\}_{i=1}^{K}$ are the Kraus operators of the CP map $\mathcal{Q}^{A \rightarrow B}$ (see Eq.~\eqref{eqn:kraus}), then the Kraus operators of the adjoint map $\left( \mathcal{Q}^{A\rightarrow B} \right)^{\dagger}$ are given by $\{Q_{i}^{\dagger}\}_{i=1}^{K}$.

The Choi Matrix of any linear map $\mathcal{Q}^{A\rightarrow B}$ is defined to be 
\begin{equation}
    \Gamma^{\mathcal{Q}}_{RB} \coloneqq \textsf{id}^{R}\otimes \mathcal{Q}^{A\rightarrow B}(|\Gamma \rangle \! \langle \Gamma |_{RA}) \; ,
\end{equation}
where $| \Gamma \rangle_{RA}\coloneqq \sum_{i=0}^{d-1}|i\rangle_{R} |i\rangle_{A}$ is the unnormalized maximally entangled state, with $d \equiv \text{dim}\mathcal{H}^{A}=\text{dim}\mathcal{H}^{R}$. The corresponding Choi state is defined as $\Phi_{RB}^{\mathcal{Q}}\coloneqq \Gamma^{\mathcal{Q}}_{RB}/d$. The linear map $\mathcal{Q}^{A \rightarrow B}$ is TP if and only if its Choi matrix satisfies $\mathrm{Tr}_{B}[\Gamma^{\mathcal{Q}}_{RB}]=I_{R}$, and CP if and only if its Choi matrix is positive, i.e. $\Gamma^{\mathcal{Q}}_{RB} \ge 0$.

In what follows, we suppress the system subscript and/or superscript if it does not lead to ambiguities. Every CP map $\mathcal{Q}^{A\rightarrow B}$ admits a Kraus decomposition
\begin{equation}
    \mathcal{Q}(\cdot) = \sum_{i=1}^{K}Q_{i}(\cdot)Q^{\dagger}_{i} \;, \label{eqn:kraus}
\end{equation}
in terms of Kraus operators $\{ Q_{i} \}_{i=1}^{K}$. If $\mathcal{Q}^{A\rightarrow B}$ is also TP, then $\sum_{i=1}^{K}Q_{i}^{\dagger}Q_{i} =I_{A}$ holds.

\subsection{Diamond Distance Between Quantum Channels}
The trace norm of any linear operator $L \in \mathcal{L}(\mathcal{H})$ is given as $\Vert L \Vert_{1} \coloneqq \operatorname{Tr}\left( \vert L\vert \right)$, where we have denoted by $\vert L \vert \coloneqq \sqrt{L^{\dagger}L}$. Therefore, we define the trace distance between any two quantum states $\rho, \sigma \in \mathcal{D}(\mathcal{H})$ to be $\frac{1}{2}\Vert \rho-\sigma \Vert_{1}$. More generally, we define the diamond distance $\frac{1}{2} \Vert \mathcal{Q}^{A}-\mathcal{S}^{A} \Vert_{\diamond}$ between two quantum channels $\mathcal{Q}^{A}$ and $\mathcal{S}^{A}$ as follows
\begin{equation}
    \sup_{\rho_{RA}} \frac{1}{2} \left \Vert \left(\textsf{id}^{R}\otimes \mathcal{Q}^{A} \right)(\rho_{RA})- \left(\textsf{id}^{R}\otimes \mathcal{S}^{A} \right)(\rho_{RA}) \right \Vert_{1} \; .
\end{equation}

\subsection{Quantum Fidelities}
Given any two quantum states $\rho, \sigma \in \mathcal{D}(\mathcal{H})$ with $\text{dim}\mathcal{H} \equiv d$, we define the fidelity function as follows
\begin{equation}
    F(\rho, \sigma)\coloneqq \left(\mathrm{Tr}\left[\sqrt{\sqrt{\rho}\sigma\sqrt{\rho}}\right]\right)^{2}=\Vert \sqrt{\rho} \sqrt{\sigma} \Vert_{2}^{2} \; .
\end{equation}
If one of the two state, say $\sigma \equiv |\psi\rangle \! \langle \psi|$, is pure, then we have $F(\rho, \psi)=\langle \psi|\rho|\psi\rangle$. Based on this definition, various fidelities that are relevant in QEC and other areas of quantum information have been defined. One such quantity is called entanglement fidelity $F_{e}$ of a channel $\mathcal{Q}^{A}\equiv \mathcal{Q}^{A \rightarrow A}$ with respect to a state $\rho \in \mathcal{D}(\mathcal{H}^{A})$ \cite{schumacher1996quantum, schumacher1996sending}, which is given by 
\begin{align}
    F_{e}(\mathcal{Q}, \rho) &\coloneqq F(\textsf{id}^{R}\otimes \mathcal{Q}^{A}(\psi^{\rho}_{RA}), \psi^{\rho}_{RA}) \\ &=\langle \psi^{\rho}|\textsf{id}^{R} \otimes \mathcal{Q}^{A}(\psi_{RA}^{\rho})|\psi^{\rho}\rangle_{RA}\;, \label{eqn:ent_fid}
\end{align}
where $\psi^{\rho}_{RA} \in \mathcal{H}^{R}\otimes \mathcal{H}^{A}$ is a purification of the density matrix $\rho \in D(\mathcal{H}^{A})$, i.e. $\mathrm{Tr}_{R}\psi_{RA}^{\rho}=\rho_{A}$. It can be shown that the entanglement fidelity is independent of the particular choice of the purification, following from the fact that the former can be expressed in terms of the Kraus operators of $\mathcal{Q}^{A}$ as \cite{schumacher1996sending, schumacher1996quantum}
\begin{equation}
    F_{e}(\mathcal{Q}, \rho)=\sum_{i=1}^{K}\mathrm{Tr}[\rho Q_{i}]\mathrm{Tr}[\rho Q_{i}^{\dagger}]=\sum_{i=1}^{K}\vert \mathrm{Tr}[\rho Q_{i}]\vert^{2} \; . \label{eqn:ent_fid_kraus}
\end{equation}

The entanglement fidelity of a quantum channel $\mathcal{Q}^{A}$ is defined to be the entanglement fidelity of $\mathcal{Q}^{A}$ with respect to the maximally mixed state $\rho_{A}=I_{A}/d$ \cite{reimpell2005iterative} (which is purified by the maximally entangled state $|\Phi \rangle_{RA}$). This can also be written in terms of the Choi state of $\mathcal{Q}^{A}$, as follows
\begin{align}
    F_{e}(\mathcal{Q}) &\equiv F_{e}\left(\mathcal{Q}, \frac{I}{d}\right) \\
    &=\langle \Phi|\textsf{id}^{R} \otimes \mathcal{Q}^{A}(\Phi_{RA})|\Phi \rangle_{RA} \\
    &=\langle \Phi| \Phi^{\mathcal{Q}}_{RA}|\Phi \rangle_{RA} \\ &=F(\Phi^{\mathcal{Q}}, \Phi)\; . \label{eqn:ent_choi}
\end{align}
Another important fidelity measure of the form $F(\rho, \psi)=\langle \psi|\rho|\psi\rangle$ is the average channel (gate) fidelity $F_{\text{avg}}(\mathcal{Q})$, defined for any $|\psi\rangle  \in \mathcal{H}^{A}$ and CPTP map $\mathcal{Q}^{A}$ as
\begin{equation}
    F_{\text{avg}}(\mathcal{Q}) \coloneqq \int d \psi \langle \psi | \mathcal{Q}(\psi) |\psi \rangle \; ,
\end{equation}
where the discrete version has appeared in \cite{schumacher1996quantum, schumacher1996sending}. In \cite{horodecki1999general}, the authors have shown that the average and entanglement fidelities are related by
\begin{equation}
    F_{\text{avg}}(\mathcal{Q})=\frac{dF_{e}(\mathcal{Q})+1}{d+1}\; . \label{eqn:HHH}
\end{equation}
Finally, in what follows, we also use the simplifying notation 
\begin{equation}
    \delta^{\mathcal{Q}} \coloneqq \arccos{\sqrt{F_{e}(\mathcal{Q})}} \; \label{eqn:err_angle}, 
\end{equation}
which can be interpreted in the $\chi$-matrix representation of quantum channels (see Appendix~\ref{sec:chi_matrix}) as the ``error angle'' by which the Kraus operators of the noisy channel $\mathcal{Q}$ deviate from the desired ``no error'' normalized basis element $B_{0}=I/\sqrt{d}$ (where $\langle B_{0}, B_{0} \rangle=1$) of the vector space $\mathcal{L}(\mathcal{H})$. It turns out that the error angle notation is very convenient when expressing the average fidelity of composite channels in terms of the individual average channel fidelities \cite{carignan2019bounding}. 

\section{Lower-Bounding Diamond Distance Using Entanglement Fidelity}
We now consider the diamond distance between any two quantum channels and show that it is lower bounded by the difference between their entanglement fidelities. This is generalized in Theorem~\ref{th:main_th} of Appendix~\ref{apx:gen_dist_meas}, where we show that the lower bound for any generalized distinguishability measure between the two channels is still fully determined by their entanglement fidelities. Besides the diamond distance considered here, fidelity and entropy-based distinguishability measures, such as Bures distance and quantum relative entropy, are also used to quantify the performance of recovery protocols, e.g. in Refs.~\cite{kubica2021using, beny2010general} and Refs.~\cite{junge2018universal, buscemi2016approximate}, respectively. Therefore, the results of this section (as well as the following sections) could be generalized for other distinguishability measures, in the light of Appendix~\ref{apx:gen_dist_meas}. 

The diamond distance is especially relevant for two reasons: $(1)$ it has a clear operational meaning in terms of the maximum probability of distinguishing between two channels in a quantum channel discrimination task \cite{acin2001statistical}, and $(2)$ it satisfies the triangle inequality and hence also the chaining property, which is useful for bounding errors in fault-tolerant quantum computing (see Appendix~\ref{sec:chaining}). 

We start by proving the following:

\begin{lemma} \label{le:exact_diamond}
For any two depolarizing channels $\tilde{\mathcal{Q}}^{A\rightarrow A}$ and $\tilde{\mathcal{S}}^{A \rightarrow A}$, the diamond distance between them is equal to the difference between their entanglement fidelities, namely 
\begin{equation}
        \frac{1}{2}\Vert \tilde{\mathcal{Q}}-\tilde{\mathcal{S}} \Vert_{\diamond} =
        \vert F_{e}(\tilde{\mathcal{Q}})-F_{e}(\tilde{\mathcal{S}}) \vert \; .
    \end{equation}
    
\end{lemma}
\begin{proof}
    Assume that $\tilde{\mathcal{Q}}$ and $\tilde{\mathcal{S}}$ are depolarizing channels with depolarizing parameters $p^{\mathcal{Q}}$ and $p^{\mathcal{S}}$, respectively. Namely,
    \begin{align}
        \tilde{\mathcal{Q}}^{A}&=(1-p^{\mathcal{Q}})\textsf{id}^{A}+p^{\mathcal{Q}}\frac{I_{A}}{d_{A}}\mathrm{Tr}_{A} \; ,  \\
        \tilde{\mathcal{S}}^{A}&=(1-p^{\mathcal{S}})\textsf{id}^{A}+p^{\mathcal{S}}\frac{I_{A}}{d_{A}}\mathrm{Tr}_{A} \; ,
    \end{align}
    and hence
    \begin{align}
        \textsf{id}^{R}\otimes \tilde{\mathcal{Q}}^{A}&=(1-p^{\mathcal{Q}})\textsf{id}^{RA}+p^{\mathcal{Q}} \mathrm{Tr}_{A}\otimes \frac{I_{A}}{d_{A}} \; ,  \\
        \textsf{id}^{R}\otimes \tilde{\mathcal{S}}^{A}&=(1-p^{\mathcal{S}})\textsf{id}^{RA}+p^{\mathcal{S}}\mathrm{Tr}_{A}\otimes  \frac{I_{A}}{d_{A}}\; ,
    \end{align}
    yields for the diamond distance
\begin{align}
    \frac{1}{2}\Vert \tilde{\mathcal{Q}}-\tilde{\mathcal{S}} \Vert_{\diamond} &= \kappa\left \vert p^{\mathcal{S}}-p^{\mathcal{Q}} \right \vert \\
    &= \frac{d \kappa}{d-1}\vert F_{\text{avg}}(\tilde{\mathcal{Q}})-F_{\text{avg}}(\tilde{\mathcal{S}}) \vert \\ &= \frac{d^{2} \kappa}{d^{2}-1}\vert F_{e}(\tilde{\mathcal{Q}})-F_{e}(\tilde{\mathcal{S}}) \vert
    \; , \label{eqn:diamond_lower}
\end{align}
where 
\begin{equation}
    \kappa(d) \equiv \frac{1}{2} \sup_{\psi_{RA}}\left\Vert \psi_{RA}-\rho_{R}\otimes \frac{I_{A}}{d_{A}} \right \Vert_{1} \; ,
\end{equation}
and we have used Eq.~\eqref{eqn:HHH}. We can rewrite $\kappa$ using the diamond distance between the identity and the replacement channel, as follows
\begin{equation}
    \kappa(d) = \frac{1}{2}\left\Vert \textsf{id}^{A}-\frac{I_{A}}{d}\mathrm{Tr}_{A} \right\Vert_{\diamond} \; .
\end{equation}
Next, we use the semi-definite program for the normalized diamond norm \cite{watrous2009semidefinite}
    \begin{align}
        \frac{1}{2}&\left\Vert \textsf{id}^{A} -\frac{I_{A}}{d}\mathrm{Tr}_{A} \right\Vert_{\diamond} \\ &=\sup_{\sigma_{RA}, \rho_{R}} \mathrm{Tr}_{RA}\left[ \sigma_{RA}\left( \Gamma_{RA}^{id}-\Gamma_{RA}^{\frac{I_{A}}{d}\mathrm{Tr}_{A}} \right) \right] \\
        &=\sup_{\sigma_{RA}, \rho_{R}} \mathrm{Tr}_{RA}\left[ \sigma_{RA}\left( \Gamma_{RA}-\frac{I_{RA}}{d} \right) \right] \; , \label{eqn:SDP_1}
    \end{align}
where the supremum is taken over all positive matrices $0 \leq \sigma_{RA} \leq \rho_{R}\otimes I_{A} $, and $\rho_{R} \in \mathcal{D}(\mathcal{H}^{R})$. Consider the eigenvalues of the $d^{2}\times d^{2}$ matrix $\Gamma_{RA}-I_{RA}/d$. First, if we denote some fixed eigenvalue of $\Gamma_{RA}$ by $\gamma$, then the corresponding eigenvalue of $\Gamma_{RA}-I_{RA}/d$ is $\gamma-1/d$. Next, let us show that $d^{2}-1$ of the $d^{2}$ eigenvalues of $\Gamma_{RA}$ are zero. This follows by considering the kernel space of $\Gamma_{RA}$ (the zero eigenvalue subspace), denoted by $\text{ker}(\Gamma_{RA}) \subset \mathcal{H}^{R}\otimes \mathcal{H}^{A}$. Namely, for all $|\psi\rangle_{RA} \in \text{ker}(\Gamma_{RA})$ we have, by definition, $\Gamma_{RA}|\psi\rangle_{RA}=0$. This is true for all states $|\psi\rangle_{RA}$ for which $\langle \Gamma|\psi\rangle_{RA}=0$, i.e. $|\psi\rangle_{RA} \in (\text{span}\{|\Gamma\rangle \})^{\perp}$, which is the ($d^{2}-1$)-dimensional orthogonal complement of the one-dimensional subspace $\text{span}\{|\Gamma\rangle\}\subset \mathcal{H}^{R}\otimes \mathcal{H}^{A}$. Finally, we note that the only non-zero eigenvalue $\gamma_{0}$ of $\Gamma_{RA}$ is determined by the trace $\mathrm{Tr}_{RA}\left[\Gamma_{RA}\right]=d$, and hence $\gamma_{0}=d$. Therefore, the eigenvalues of $\Gamma_{RA}-I_{RA}/d$ are given by the list
\begin{equation}
    \text{eigenval}\left(\Gamma-\frac{I}{d}\right)=\left\{\frac{d^{2}-1}{d}, -\frac{1}{d}, -\frac{1}{d}, \cdots, -\frac{1}{d}\right\} \; .
\end{equation}
As we can see, only one of the eigenvalues of $\Gamma_{RA}-I_{RA}/d$ is positive. Therefore, we write the spectral decomposition of the matrix $\Gamma_{RA}-I_{RA}/d$ as follows
\begin{equation}
    \Gamma_{RA}-\frac{I_{RA}}{d}=\frac{d^{2}-1}{d}|\gamma_{0}\rangle \! \langle \gamma_{0} | -\frac{1}{d}\sum_{i=1}^{d^{2}-1}|\gamma_{i}\rangle \! \langle \gamma_{i} | \; ,
\end{equation}
where $\{|\gamma_{i} \rangle\}_{i=0}^{d^{2}-1} $ is its orthonormal eigenbasis.

It follows that, to maximize the argument of Eq.~\eqref{eqn:SDP_1}, we need to consider the support of the positive semi-definite operator $\sigma_{RA}$ to be in the (one-dimensional) support of $\Gamma_{RA}$ (which is orthogonal to $\text{ker}(\Gamma_{RA})$), i.e. we need to search for $\sigma_{RA}$ in the form $\sigma_{RA}=z|\gamma_{0}\rangle \! \langle \gamma_{0} |$ for some $z\ge 0$. Substituting into the constraint $\sigma_{RA} \leq \rho_{R}\otimes I_{A}$ gives
\begin{align}
    &\rho_{R}\otimes I_{A}-z|\gamma_{0}\rangle \! \langle \gamma_{0} | \ge 0 \label{eqn:deriv_1} \\ 
    \Rightarrow & \langle \Gamma|\rho_{R}\otimes I_{A}|\Gamma \rangle_{RA}-z|\langle \Gamma|\gamma_{0}\rangle |^{2} \ge 0 \\
    \Rightarrow & \mathrm{Tr}[\rho]-z\langle \gamma_{0}|\Gamma|\gamma_{0}\rangle_{RA} \ge 0 \\ 
    \Rightarrow &1-zd \ge 0 \Rightarrow z\leq \frac{1}{d} \; .
\end{align}
Consequently, if there exists $ \rho_{R} \in \mathcal{D}(\mathcal{H}^{R})$ such that $\sigma_{RA}$ is fully in the support of $\Gamma_{RA}$, i.e. $\sigma_{RA}=z|\gamma_{0}\rangle \! \langle \gamma_{0}|_{RA}$ (where $\Gamma|\gamma_{0}\rangle=d|\gamma_{0}\rangle$), then it must be the case that the normalization $z \leq 1/d$. The resulting maximization in Eq.~\eqref{eqn:SDP_1} will thus yield for $\kappa(d)$
\begin{align}
    \sup_{\sigma_{RA}, \rho_{R}}\mathrm{Tr}_{RA}\left[ \sigma_{RA}\left(\Gamma_{RA}-\frac{I_{RA}}{d}\right) \right] &=\frac{d^{2}-1}{d}z_{\text{max}} \label{eqn:deriv_2} \\ &=\frac{d^{2}-1}{d^{2}} \; ,
\end{align}
or equivalently,
\begin{align}
    \frac{1}{2}\Vert \tilde{\mathcal{Q}}-\tilde{\mathcal{S}} \Vert_{\diamond} &= \frac{d^{2} \kappa(d)}{d^{2}-1}\vert F_{e}(\tilde{\mathcal{Q}})-F_{e}(\tilde{\mathcal{S}}) \vert \\ &= \vert F_{e}(\tilde{\mathcal{Q}})-F_{e}(\tilde{\mathcal{S}}) \vert \; .
\end{align}
In the above analysis, we presumed the existence of a density matrix $\rho_{R}$ for which $\sigma_{RA}=z|\gamma_{0} \rangle \! \langle \gamma_{0}|_{RA} \leq \rho_{R}\otimes I_{A}$. It is easy to see that the pick $\rho_{R}=I_{R}/d$ satisfies the inequality $z|\gamma_{0} \rangle \! \langle \gamma_{0}|_{RA} \leq \rho_{R}\otimes I_{A}$, as well as allowing the normalization $z$ to reach its maximum value $z_{\text{max}}=1/d$.

Note that, if we extend the support  $\sigma_{RA}=z|\gamma_{0}\rangle \! \langle \gamma_{0}|+\sum_{i=1}^{d^{2}-1}z_{i}|\gamma_{i}\rangle \! \langle \gamma_{i}|$, then the above argument (starting from Eq.~\eqref{eqn:deriv_1}) still yields $z \leq 1/d$, while simultaneously leading to a sub-optimal outcome in Eq.~\eqref{eqn:deriv_2} due to the contribution of the negative eigenvalues of $\Gamma_{RA}-I_{RA}/d$. Taking $\sigma_{RA}$ to be off-diagonal in the $\{|\gamma_{i}\rangle\}_{i=0}^{d^{2}-1}$ does not change this argument.
\end{proof}

It is important to note that this lemma has been known previously for special cases, e.g. in \cite{pirandola2019fundamental, benenti2010computing} between qubit depolarizing maps ($d=2$) and between a qudit depolarizing map and the identity map ($p^{\mathcal{S}}=0$), respectively. However, Pirandola \textit{et al.} in \cite{pirandola2019fundamental} used a different technique to compute essentially the same quantity $\kappa(d)$ appearing in our derivation of Lemma~\ref{le:exact_diamond}, which crucially does not depend on the depolarizing parameters $p^{\mathcal{Q}}$ and $p^{\mathcal{S}}$.

We now prove a lower bound for the diamond distance between any two quantum channels with the same input and output spaces. We frame this as follows

\begin{lemma} \label{co:examples}
    For any two CPTP maps $\mathcal{Q}^{A\rightarrow A}$ and $\mathcal{S}^{A \rightarrow A}$, the diamond distance between them is lower bounded by the difference in their entanglement fidelities, namely
    \begin{equation}
        \frac{1}{2}\Vert \mathcal{Q}-\mathcal{S} \Vert_{\diamond} \ge
        \vert F_{e}(\mathcal{Q})-F_{e}(\mathcal{S}) \vert \; .
    \end{equation}
\end{lemma}
\begin{proof}
    It is known that any quantum supermap (a linear map from one quantum channel to another) that is a convex combination of Pauli unitary supermaps (also known as ``twirling'') renders any input channel $\mathcal{Q}$ into a depolarizing channel $\tilde{\mathcal{Q}}$ \cite{dankert2009exact}, with the same entanglement fidelity. The Lemma is then a direct consequence of applying the data-processing inequality to the diamond distance $\Vert \mathcal{Q}-\mathcal{S} \Vert_{\diamond}$ with respect to the Pauli twirling supermap \cite{gour2019comparison}, which yields  
    \begin{equation}
        \frac{1}{2}\Vert \mathcal{Q}-\mathcal{S} \Vert_{\diamond} \ge \frac{1}{2}\Vert \tilde{\mathcal{Q}}-\tilde{\mathcal{S}} \Vert_{\diamond}
         \; ,
    \end{equation}
    where $\tilde{\mathcal{Q}}$ and $\tilde{\mathcal{S}}$ are the resulting depolarizing channels \cite{horodecki1999general} (also see Lemma~\ref{le:twirling} in Appendix~\ref{apx:gen_dist_meas}). Then, we note that the right-hand side is found from
    Lemma~\ref{le:exact_diamond}. Finally, the proof is completed by the fact that any random unitary supermap preserves the entanglement fidelity \cite{horodecki1999general}, hence $F_{e}(\tilde{\mathcal{Q}})=F_{e}(\mathcal{Q})$ and $F_{e}(\tilde{\mathcal{S}})=F_{e}(\mathcal{S})$.
\end{proof}

\begin{remark}
    If one of the quantum channels is the identity, then a much simpler derivation could be found in the supplementary material of \cite{ouyang2023approximate} using the Fuchs-van de Graaf inequality for quantum channels, which yields a two-sided bound on the diamond distance.
\end{remark}


\section{Fundamental Bounds on Recovery with Incomplete Knowledge} \label{sec:single_cycle}
Here, we are interested in applying the lower bound derived in Lemma~\ref{co:examples} of the previous section to the spectator-based recovery setting, succinctly described in Fig.~\ref{fig:circuit2}. We will assume that we are given a parametric family of quantum channels $\{\mathcal{N}_{\theta}\}_{\theta \in \Theta}$ that is motivated from certain physical assumptions about the memory-environment interaction, where $\Theta$ is the allowed range of values for the noise parameter $\theta$. 

\subsection{The Regime of Validity} \label{sub:validity}
Let us now identify four different noise instability regimes and then expand upon the relevant regime for this article. The four cases are described as follows:
\begin{enumerate}
    \item When neither the noise family $\{\mathcal{N}_{\theta}\}_{\theta \in \Theta}$ nor the true noise parameter $\theta$ change in time. This case is best described by the ``perfect'' knowledge scenario in Fig.~\ref{fig:circuit1}, and is the most common in literature. 
    \item When the noise family $\{\mathcal{N}_{\theta}\}_{\theta \in \Theta}$ remains valid, but the true noise parameter $\theta$ varies \textit{stroboscopically}. Namely, the characteristic timescale $\tau_{\theta}$ for the variations of the true value of $\theta$ is larger compared to the duration of a single recovery cycle $\Delta t_{\mathcal{R}}$. Therefore, variations in the noise parameter are on the timescale of multiple recovery cycles (see e.g. \cite{muller2015interacting} for superconducting qubits). In this regime, performing the real-time quantum estimation of the new value of $\theta$ for every recovery cycle becomes useful, hence the need for a spectator system. This regime is best described by the ``incomplete'' knowledge scenario in Fig.~\ref{fig:circuit2}. 
    \item When the noise family $\{\mathcal{N}_{\theta}\}_{\theta \in \Theta}$ remains valid, but the noise parameter $\theta$ varies non-negligibly during a single recovery cycle (i.e. $\tau_{\theta} \sim \Delta t_{\mathcal{R}}$). In this case, the usefulness of the classical side-information in the spectator-based recovery in Fig.~\ref{fig:circuit2} is no longer clear. Instead, this noise regime might benefit from continuously applied recovery, e.g. \cite{kwon2022reversing}. Alternatively, a robustness approach (as opposed to adaptation) might also be suitable (see below). 
    \item When the noise family $\{\mathcal{N}_{\theta}\}_{\theta \in \Theta}$ is changing within a single recovery cycle. This noise regime will mainly benefit from the design of robust recovery protocols, e.g. in \cite{kosut2008robust, ballo2009robustness, huang2019robustness, layden2020robustness}, rather than the spectator-based recovery presented in Fig.~\ref{fig:circuit2}.
\end{enumerate}
In the rest of the article, we focus on the stroboscopic noise regime, and consider the advantages and limitations of using the spectator-based recovery protocol in Fig.~\ref{fig:circuit2}.

To measure the success of the recovery protocol, we recall that the goal is to achieve a complete (or, at least, an approximate) recovery $\mathcal{R}\circ \mathcal{N}_{\theta}(\rho)=\rho$ of the noisy channel $\mathcal{N}_{\theta}$ for a subset of states $\rho \in \mathcal{D}(\mathcal{C}) \subset \mathcal{D}(\mathcal{H})$ in the codespace $\mathcal{C}$. The ``optimality'' of the recovery map $\mathcal{R}$ for a given noisy channel $\mathcal{N}_{\theta}$ could be quantified in various ways. Motivated by Lemma~\ref{co:examples} for the diamond distance (and more generally Theorem~\ref{th:main_th} for all distinguishability measures), we choose the entanglement fidelity to be the quantifier of the optimal recovery, i.e.
\begin{equation}
    \mathcal{R}_{\theta}\coloneqq \text{argmax}F_{e}(\mathcal{R}\circ \mathcal{N}_{\theta}) \hspace{0.1cm} \text{;} \hspace{0.1cm} \mathcal{R}\in \text{CPTP}(\mathcal{H}) \; . \label{eqn:opt_recovery}
\end{equation}
Indeed, entanglement fidelity has been used as a figure of merit for QEC in e.g. \cite{reimpell2005iterative, fletcher2007optimum, fletcher2008channel, ballo2009robustness, kosut2008robust, kosut2009quantum}. We note that the concatenated form $\mathcal{R}\circ \mathcal{N}_{\theta}$ of the memory dynamics presupposes that the recovery map $\mathcal{R}$ is applied much faster than the noisy dynamics. In what follows, we shall compare two different scenarios:
\begin{itemize}
    \item \textit{Optimal recovery} scenario (Fig.~\ref{fig:circuit1}), which corresponds to the optimal choice of the recovery map $\mathcal{R}_{\theta}$ (as defined in Eq.~\eqref{eqn:opt_recovery}) for the noisy channel $\mathcal{N}_{\theta}$, where the value of $\theta \in \Theta$ is completely known. 
    \item \textit{Best-guess recovery} scenario (Fig.~\ref{fig:circuit2}), which corresponds to the optimal recovery choice $\mathcal{R}_{\hat{\theta}}$ (as defined in Eq.~\eqref{eqn:opt_recovery}) for the estimated noisy channel $\mathcal{N}_{\hat{\theta}}$, where $\hat{\theta}$ is the best estimate of $\theta$. The latter is defined to be the minimum variance unbiased estimator (MVUE). 
\end{itemize}


\begin{remark} \label{rem:non_pauli}
    Not all QEC codes require a recovery channel $\mathcal{R}_{\theta}$ that depends on the noise parameter $\theta$. Such channels are known to be Pauli channels, where the recovery operation is fully determined by a subset of Pauli operators from the general Pauli group, as is known in the stabilizer formalism of QEC \cite{gottesman1997stabilizer, nielsen2002quantum}. Therefore, in what follows, we consider non-Pauli channels, of which, the generalized amplitude damping channel is a prime example \cite{nielsen2002quantum}. 
\end{remark}

\begin{figure}[t]
    \begin{quantikz}
        & \lstick{$\psi^{S}$} & \gate[style={fill=yellow!50},wires=2][1cm]{\mathcal{Z}_{\bm{\theta}}}& \meter{$\bm{\hat{\theta}}^{S}(x)$} & \cwbend{1}\\
        & \lstick{$\rho^{M}$} &  & \qw &  \gate[style={fill=red!50}, wires=1][1cm]{\mathcal{R}_{\bm{\hat{\theta}}^{M}}} \vcw{-1} & \qw \\
    \end{quantikz}
    \caption{The spectator ($S$) and memory ($M$) systems are subject to a single environment, characterized by the parametric family of quantum channels $\{\mathcal{Z}_{\bm{\theta}}\}_{\bm{\theta}\in \Theta}$. The mother channel $\mathcal{Z}^{MS}_{\bm{\theta}}$ generates the two local channels $\mathcal{N}^{MS \rightarrow M}_{\bm{\theta}^{M}}\equiv \operatorname{Tr}_{S}\circ \mathcal{Z}_{\bm{\theta}}$ and $\mathcal{M}^{MS \rightarrow S}_{\bm{\theta}^{S}}\equiv \operatorname{Tr}_{M}\circ \mathcal{Z}_{\bm{\theta}}$ describing the reduced dynamics of the memory and spectator systems, respectively.}
    \label{fig:multi_param}
\end{figure}
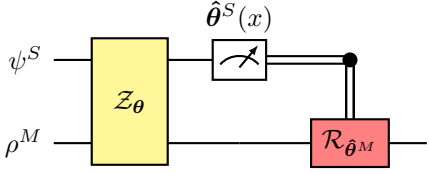

\subsection{A General Framework For Spectator-Based Recovery Protocols}
\label{sec:general_spec_based_recov}
In this section, we formulate the spectator-based recovery protocol (Fig.~\ref{fig:circuit2}) as the combination of two consecutive processes: (i) parameter estimation of the memory noise, using the spectator system, and (ii) application of the corresponding ``best-guess'' recovery map (informed by the estimated value of the noise parameter(s), rather than the true value) to the quantum memory. It is important to emphasize that the parameter estimation stage of spectator-based recovery is more resource efficient than a direct process tomography of the noise. This is due to the prior knowledge of the parametric family of noisy quantum channels, either from physical consideration, or even initial process tomography. This efficiency is especially important for operating a real-time quantum memory, which requires such resources to be replenished after each recovery cycle.

To set up a general framework for spectator-based recovery protocols, we start by considering the origins of the reduced memory ($M$) and spectator ($S$) dynamics in the general multi-parameter regime. In what follows, $MS$ denotes the mother system for both the quantum memory $M$ and the spectator system $S$, and its Hilbert space has the tensor product structure $\mathcal{H}^{MS}=\mathcal{H}^{M}\otimes \mathcal{H}^{S}$. Since both systems are subject to the same local environment, we assume that there is a mother channel $\mathcal{Z}^{MS}_{\bm{\theta}}$ (with a noise parameter vector $\bm{\theta} \in \Theta^{p}$, where $\Theta^{p}$ is a $p$-dimensional parameter space) that acts on the memory and spectator systems collectively, as shown in Fig.~\ref{fig:multi_param}. Without loss of generality, we assume that the components of the multiparameter vector $(\theta_{1}, \cdots, \theta_{p})\in \Theta^{p}$ are treated as independent parameters, namely that none of the components could be expressed as a function of the rest, e.g. $\theta_{i}=f(\theta_{1}, \cdots, \theta_{i-1}, \theta_{i+1}, \cdots, \theta_{p})$. Otherwise, for every such constraint, the number of independent parameters is reduced by one. The mother channel yields the reduced dynamics 
\begin{align}
    \mathcal{N}^{MS \rightarrow M}_{\bm{\theta}^{M}}&\equiv \operatorname{Tr}_{S}\circ \mathcal{Z}^{MS}_{\bm{\theta}} \; , \label{eqn:reduced_N} \\
    \mathcal{M}^{MS \rightarrow S}_{\bm{\theta}^{S}}&\equiv \operatorname{Tr}_{M}\circ \mathcal{Z}^{MS}_{\bm{\theta}} \; , \label{eqn:reduced_M}
\end{align}
where the partial tracing with respect to $S$ in the first equation defines the partition of the global set of noise parameters $\bm{\theta}$ into a subset of relevant parameters $\bm{\theta}^{M} \in \Theta^{p^{M}}$ for the reduced dynamics of $M$ (where in general $p^{M} \leq p$) and a complementary set of parameters $\bm{\theta}^{M\perp} \in \Theta^{p-p^{M}}$ that is irrelevant for the reduced dynamics (namely $\Theta^{p}=\Theta^{p^{M}}\times \Theta^{p-p^{M}}$). A similar partition of $\bm{\theta}=\{\bm{\theta}^{S}, \bm{\theta}^{S\perp}\}$ to relevant and irrelevant parameters for the spectator dynamics occurs when applying the partial trace with respect to $M$ to the mother channel (please see Appendix~\ref{apx:param_indep} for details on when this is possible). With these two natural partitions $\{\bm{\theta}^{M}, \bm{\theta}^{M\perp}\}$ and $\{\bm{\theta}^{S}, \bm{\theta}^{S\perp}\}$ of the parameter vectors $\bm{\theta}$, we can also decompose the latter via a joint partition, as follows
\begin{equation}
    \bm{\theta}=(\bm{\theta}^{S} \cap \bm{\theta}^{M}) \cup (\bm{\theta}^{S} \cap \bm{\theta}^{M\perp}) \cup (\bm{\theta}^{S \perp} \cap \bm{\theta}^{M}) \cup (\bm{\theta}^{S \perp} \cap \bm{\theta}^{M\perp}) \; . \label{eqn:4_partition}
\end{equation}
From this joint partition, it is clear that the spectator system $S$ can only help estimate the subset of memory noise parameters $\bm{\theta}^{S}_{\text{I}} \equiv \bm{\theta}^{S} \cap \bm{\theta}^{M} \subseteq \bm{\theta}^{M}$ (where the subscript ``I'' stands for ``interest'' parameters, as opposed to the ``nuisance'' parameters $\bm{\theta}^{S}_{\text{N}} \equiv \bm{\theta}^{S} \cap \bm{\theta}^{M \perp} \not\subseteq \bm{\theta}^{M}$ of the spectator dynamics \cite{suzuki2020quantum}, which are irrelevant for the memory dynamics). 
Therefore, we see that it is necessary to have $\Theta^{p^{M}} \subseteq \Theta^{p^{S}}$ for the quantum estimation task via the spectator system to be useful for identifying the relevant noise parameters $\bm{\theta}^{M}$ which affect the quantum memory. If this is not the case, then one would require multiple spectators $S_{1}, S_{2}, \cdots, S_{l}$ (from potentially different physical systems) such that the parameter space $\Theta^{p^{M}}$ is contained by the combined parameter subspaces $\cup_{i=1}^{l}\Theta^{p^{S_{i}}}$. For simplicity, we assume in the rest of the article that $p^{S}_{\text{I}}=p^{M}$.

In general, the reduced states of $S$ and $M$ following the application of Eqs.~\eqref{eqn:reduced_N} and \eqref{eqn:reduced_M} will depend on the \textit{global} input state $MS$. This observation still holds even if the input state of $MS$ is of product form $\rho^{MS}=\rho^{M}\otimes \psi^{S}$. Therefore, to arrive at local noise channels $\mathcal{N}^{M}_{\bm{\theta}^{M}}$ and $\mathcal{M}^{S}_{\bm{\theta}^{S}}$ (as shown in Fig.~\ref{fig:circuit2}) that are independent of the input states of $S$ and $M$, respectively, the mother channel itself has to be separable. Namely
\begin{equation}
    \mathcal{Z}^{MS}_{\bm{\theta}}=\mathcal{N}^{M}_{\bm{\theta}^{M}} \otimes \mathcal{M}^{S}_{\bm{\theta}^{S}} \; ,
\end{equation}
similar to the independent noise approximation that is commonly used in QEC literature. The input state $\rho^{MS}=\rho^{M}\otimes \psi^{S}$ is selected such that $\rho^{M}$ is the state of the quantum memory that we would like to protect from the noise, whereas $\psi^{S}$ is the probe state of the spectator system which we are technically free to choose to achieve the optimal precision in parameter estimation.

In Appendix~\ref{apx:QFIM}, we review relevant aspects of multi-parameter quantum estimation theory. It is known that when there are no nuisance parameters present ($p^{S}_{\text{I}}=p^{S}$, $p^{S}_{N}=0$), the quantum estimation limit of the parameters $\bm{\theta}^{S}_{\text{I}}\equiv \bm{\theta}^{M}$ is given by the partial symmetric logarithmic derivative (SLD) quantum Fisher information matrix (QFIM), which has dimensions of $p^{S}\times p^{S}$. However, when $p^{S}_{\text{N}}=p^{S}-p^{S}_{\text{I}}>0$ nuisance parameters are present, then the quantum estimation limit will be given by the $p^{S}_{\text{I}}\times p^{S}_{\text{I}}$ \textit{partial} SLD QFIM, which comprises a tighter lower bound on the estimation variance $\operatorname{Var}(\bm{\hat{\theta}}_{\text{I}}^{S})$ than the $p^{S}_{\text{I}}\times p^{S}_{\text{I}}$ standard SLD QFIM of the parameters $\bm{\theta}^{S}_{\text{I}}$. It is known that these two quantities are equal only if the nuisance parameters are informationally orthogonal to the parameters of interest \cite{barndorff1994inference}, namely when the SLD QFIM of $\bm{\theta}^{S}$ block diagonalizes, with blocks of dimensions $p^{S}_{\text{I}}$ and $p^{S}_{\text{N}}=p^{S}-p^{S}_{\text{I}}$, respectively. Furthermore, it has been shown in \cite{suzuki2020quantum} that, in the single parameter regime $p^{M}=1$ ($p^{S}_{\text{I}}=1$, $p^{S}_{N}=p^{S}-1$), the lower bound in the variance of any locally unbiased estimator $\hat{\theta}^{S}_{\text{I}}$ is achievable via an optimal measurement constructed from the eigenprojectors of the SLD operators of $\{\mathcal{M}_{\bm{\theta}^{S}}(\psi)\}_{\bm{\theta}^{S}}$. A similar optimal measurement construction achieving the QCRB is not known for $p^{M}>1$ in the presence of nuisance parameters. However, when the nuisance parameters are absent, the above optimal measurement saturates the QCRB for any $p^{M} \ge 1$ if and only if the SLD operators of different parameters commute \cite{liu2020quantum}. In the rest of the article, we consider the single parameter case $p^{M}=1$.

Finally, we would like to point out that the main theorems of this manuscript are generalizable to the multiparameter setting $p_{M}>1$, however, a full consideration of all its nuances are left for future work. This includes the incompatibility of different parameters \cite{ragy2016compatibility, belliardo2021incompatibility, lu2021incorporating}, which is an exclusive problem to the multiparameter regime. Further, it is known that the QCRB is a less tight version of the Holevo Cram\'er-Rao bound (HCRB) in multi-parameter quantum estimation theory. The latter is known to be efficiently computable \cite{albarelli2019evaluating, sidhu2021tight}, but it is generally saturated for collective measurements over different probes. Instead, if one is interested in local measurements (which is more practical for parameter estimation), then the Nagaoka-Hayashi (NH) bound \cite{nagaoka2005generalization} is the relevant bound, which is also efficiently computable \cite{conlon2021efficient}. A recent work \cite{hayashi2022tight} shows how these Cram\'er-Rao type bounds are unified under the umbrella of conic linear programming.

\subsection{Spectator Dynamics With No Nuisance Parameters}

In the introduction, as well as Figs.~\ref{fig:circuit2} and \ref{fig:[4,1]_circuit}, we have emphasized the fact that the spectator system need not be the same physical system as the computational or memory system. Consequently, the dynamics of the spectator system $\mathcal{M}_{\theta}$ is generally different from the dynamics $\mathcal{N}_{\theta}=\bigotimes_{i=1}^{n}\mathcal{N}_{\theta}^{(1)}$ of the $n$ memory qubits, though it still depends on the same environment noise parameter $\theta$.

\subsubsection{Independent Noise Approximation}
If the spectator system is made out of $s$ subsystems (e.g. qubits), then the independent noise model reads
\begin{equation}
    \mathcal{M}_{\theta}=\bigotimes_{i=1}^{m}\mathcal{M}^{(1)}_{\theta} \; , \label{eqn:M_spec}
\end{equation}
where $\mathcal{M}_{\theta}^{(1)}$ is a quantum channel acting on the $i$-th subsystem. Note that we assumed negligible spatial variability of the noise parameter $\theta$. The noise separability assumption need not mean that the qubit noises are uncorrelated, as classical correlation between the experienced noise parameters by different qubits is still possible in principle. The above separability assumption only means that the noise correlations are classical and hence non-entangling.

\subsubsection{Spectator Qubits As Memory Qubits With Controllable Environment Coupling}
To perform recovery with incomplete knowledge, we need to hypothesize a relation between the spectator and memory qubit dynamics, i.e. $\mathcal{M}_{\theta}^{(1)}$ and $\mathcal{N}_{\theta}^{(1)}$. Since both types of qubits are subject to the same noisy environment with potentially different coupling strengths, we hypothesise
\begin{equation}
    \mathcal{M}_{\theta}^{(1)}=\mathcal{N}_{f(\theta)}^{(1)} \; , \label{eqn:memo_to_spec}
\end{equation}
where $f(\theta) \in [0, 1]$ is a monotone increasing function of its argument. To justify this choice, consider the case where $\theta$ has the following form
\begin{equation}
    \theta = 1-e^{-t/T_{1}} \; ,
\end{equation}
where $T_{1}$ is the spin relaxation time \cite{breuer2002theory, etxezarreta2021time}. This is the case e.g. for the qubit amplitude-damping channel, which we consider both in Section~\ref{sec:single_cycle} and Section~\ref{sec:multi_cycle}. Then, by expressing $t/T_{1}$ in terms of $\theta$, we arrive at
\begin{equation}
    f_{\gamma}(\theta)=1-(1-\theta)^{\gamma} \; , \label{eqn:f_spec}
\end{equation}
where $\gamma = T_{1}^{\text{memo}}/T_{1}^{\text{spec}}$. The requirement that the spectator qubits should exhibit faster dynamics than the memory qubits translates to $\gamma > 1$. 
Eq.~\eqref{eqn:memo_to_spec} could also be viewed from the point of view of recent progress in quantum control, e.g. via Hamiltonian amplification in bosonic systems \cite{arenz2017dynamical} or decoherence control in NV centers \cite{lei2017decoherence}, which yields controllable qubit coupling strengths. If such qubits are used to build a quantum memory, then a portion of these qubits could be reserved as spectators, and the environment coupling could be adjusted to maximize the sensitivity of the spectator qubits (see below).

\subsubsection{Quantifying The Physical Choice of Spectator Systems}
Given that a particular physical medium (e.g. a spin lattice with multiple spin species A, B, C, $\cdots$) is populated with memory qubits (say, spin species A), a natural question is whether the spectator qubits should be chosen from the same or different spin species. This question is especially relevant for QEC in hybrid spin registers e.g. in diamond \cite{taminiau2014universal, waldherr2014quantum}. To answer this question, we recall that in quantum estimation theory, the sensitivity of the spectator qubit dynamics $\psi \rightarrow \mathcal{M}^{(1)}_{\theta}(\psi)$ to the noise parameter $\theta$ is characterized by the QFI of the output state $\mathcal{M}^{(1)}_{\theta}(\psi)$. The sensing advantage of using a spectator qubit of a different species than the memory qubits is then determined by the ratio SM$=\textsf{I}_{\operatorname{QF}}(\mathcal{M}^{(1)}_{\theta}(\psi))/\textsf{I}_{\operatorname{QF}}(\mathcal{N}^{(1)}_{\theta}(\psi))$ (which we call the spectator multiplier) following the QCRB (please see Appendix~\ref{apx:QFIM} for a self-contained review of QFI and QCRB)
\begin{align}
    \operatorname{Var}(\hat{\theta}) & \ge \frac{1}{\textsf{I}_{\operatorname{QF}}(\mathcal{M}_{\theta}(\psi^{\otimes m}))} \\ &=\frac{1}{m \textsf{I}_{\operatorname{QF}}(\mathcal{M}^{(1)}_{\theta}(\psi))} \\
    & \equiv \frac{1}{n^{\prime}\textsf{I}_{\operatorname{QF}}(\mathcal{N}^{(1)}_{\theta}(\psi))}\; ,
\end{align}
where $n^{\prime}\equiv \text{SM}\times m$ indicates the equivalent number of physical memory qubit species used for sensing. In the case where our hypothesis in Eq.~\eqref{eqn:memo_to_spec} holds, we can use the property of the QFI for the change of parameters \cite{liu2020quantum}
    \begin{equation}
        \textsf{I}_{\operatorname{QF}}(\mathcal{N}^{(1)}_{\theta})=\textsf{I}_{\operatorname{QF}}(\mathcal{N}^{(1)}_{f(\theta)})\left( \frac{df}{d\theta} \right)^{2} \; ,
    \end{equation} 
which yields SM$=1/(df/d\theta)^{2}$. Therefore, the optimal spectator, in this case, is the one that experiences an effective parameter $\theta_{\text{eff}}\equiv f(\theta)$ with $f^{\prime}(\theta)=0$ at the actual value of the parameter (hence effectively yielding an asymptotic estimation regime $n^{\prime} \rightarrow \infty$ for a finite $m$). Realistically, since we do not have prior knowledge of the noise parameter $\theta$, the ideal $f(\theta)$ will be mostly constant, i.e. $f^{\prime}(\theta)=0$ (at least in the relevant variability range of $\theta$), with $f(0)=0$ and $f(1)=1$.  

Finally, we note that the ratio SM quantifies the relative ``speed'' of the dynamics between the spectator and the memory qubits, following the relation between QFI and the Bures distance between two consecutive ``instances'' of a channel \cite{yuan2017fidelity}. This is important for the spectator-based recovery protocol, as the timely feedforward control of the memory qubits based on the classical side information from the spectator is crucial. 

\subsection{Information-Theoretic Bounds}
Here we consider the metrological bounds associated with the spectator-based recovery protocol, following Lemma~\ref{co:examples} of the previous section.

\subsubsection{Fundamental Limitations For All Recovery Protocols}
Assume that we are given a noise channel $\mathcal{N}^{A\rightarrow B}$. By picking $\mathcal{Q}^{A} \equiv \mathcal{R}^{B\rightarrow A} \circ \mathcal{N}^{A \rightarrow B}$ and $\mathcal{S}^{A} \equiv \textsf{id}^{A}$, Lemma~\ref{co:examples} can be reframed in the context of noise recovery to read
\begin{equation}
    \frac{1}{2}\left \Vert \mathcal{R}\circ \mathcal{N}-\textsf{id} \right \Vert_{\diamond} \ge
    1-F_{e}(\mathcal{\mathcal{R}\circ \mathcal{N}})  \; . \label{eqn:QEC_lower}
\end{equation}
We note that this is the lower bound of a two-sided bound on the diamond distance between a quantum channel and the identity channel, recently derived in \cite{ouyang2023approximate}. This lower bound holds for both recovery with perfect and incomplete knowledge scenarios in Figs.~\ref{fig:circuit1} and \ref{fig:circuit2}, respectively. Next, we consider the incomplete knowledge scenario and analyze the contribution of the spectator system to the lower bound.

\subsubsection{Metrological Cost of Spectator-Based Recovery Protocols} \label{sec:spec_sys_contribution}
Consider the scenario described in Fig.~\ref{fig:circuit2}, which is what we expect for real-time quantum memories. The best estimate $\hat{\theta}$ of the unknown $\theta \in \Theta$ is found by the spectator system for each time-interval over which the value of the \textit{stroboscopic} (slowly varying) variable $\theta$ is approximately constant. This characteristic timescale of the stroboscopic noise parameter $\theta$ should be larger than the combined characteristic times of the noisy $\mathcal{N}_{\theta}$ and the best-guess recovery $\mathcal{R}_{\hat{\theta}}$ dynamics. Consequently, the relevant total dynamics of the encoded system is given by $\mathcal{R}_{\hat{\theta}}\circ \mathcal{N}_{\theta}$. Compared to the ideal case where the noise parameter is known perfectly, the metrological cost associated with the adaptation of spectator-based quantum memories is given by Lemma~\ref{co:examples} as
\begin{equation}
    \frac{1}{2}\left\Vert \mathcal{R}_{\theta}\circ \mathcal{N}_{\theta} -\mathcal{R}_{\hat{\theta}}\circ \mathcal{N}_{\theta} \right\Vert_{\diamond} \ge F_{e}(\mathcal{R}_{\theta}\circ \mathcal{N}_{\theta})-F_{e}(\mathcal{R}_{\hat{\theta}}\circ \mathcal{N}_{\theta}) \; , \label{eqn:main_spec_lower}
\end{equation}
where we have substituted for the quantum channels  $\mathcal{Q}^{A}\equiv \mathcal{R}^{B \rightarrow A}_{\theta}\circ \mathcal{N}^{A \rightarrow B}_{\theta}$ and $\mathcal{S}^{A}\equiv \mathcal{R}^{B \rightarrow A}_{\hat{\theta}}\circ \mathcal{N}^{A \rightarrow B}_{\theta}$. Note that an upper bound to the diamond distance is given by 
$\frac{1}{2}\left\Vert \mathcal{R}_{\theta} -\mathcal{R}_{\hat{\theta}} \right\Vert_{\diamond}$, which follows from its definition.

We start by providing a lower bound to the right-hand side in Eq.~\eqref{eqn:main_spec_lower} for \textit{arbitrary finite} estimation errors $\hat{\theta}-\theta$, as follows
\begin{theorem} \label{th:finite_var}
    Consider a parameterized noise channel $\mathcal{N}^{A \rightarrow B}_{\theta}$, and two arbitrary recovery maps $\mathcal{R}^{B \rightarrow A}$ and $\tilde{\mathcal{R}}^{B \rightarrow A}$. Then, the difference between the corresponding entanglement fidelities of recovery is lower bounded by the Choi states of the individual quantum channels $\mathcal{N}_{\theta}$, $\mathcal{R}$, and $\tilde{\mathcal{R}}$, as follows 
    \begin{align}
        &  \vert F_{e}\left( \mathcal{R}\circ \mathcal{N}_{\theta} \right)-F_{e}( \tilde{\mathcal{R}}\circ \mathcal{N}_{\theta}) \vert  \nonumber \\
        & \hspace{0.2cm} \ge \left \Vert \Phi_{AB}^{\mathcal{R}}-\Phi_{AB}^{\tilde{\mathcal{R}}} \right \Vert_{\alpha} \times  \left \Vert \Phi^{\mathcal{N}_{\theta}}_{AB} \right \Vert_{\beta} \; ,
    \end{align}
    where $\alpha \in [0,1)$ and  $1/\alpha+1/\beta=1$ defines the $\beta<0$ H\"older dual to $\alpha$, and $\Vert X \Vert_{\alpha(\beta)}\coloneqq \left(\operatorname{Tr}[\vert X^{\alpha(\beta)} \vert]\right)^{1/\alpha(\beta)}$.
\end{theorem}
\begin{proof}
    Please see Appendix~\ref{axp:finite_est_error}.
\end{proof}

Applying this theorem particularly to the optimal $\mathcal{R}_{\theta}$ and best-guess $\mathcal{R}_{\hat{\theta}}$ recovery maps, defined via the optimization in Eq.~\eqref{eqn:opt_recovery}, yields a lower bound to the right-hand side in Eq.~\eqref{eqn:main_spec_lower} for arbitrary estimation errors $\hat{\theta}-\theta$.

In this article, we will mainly be interested in the small estimation error $\hat{\theta}-\theta$ case. Again, this is motivated by physical considerations. For example, temporal variations of $T_{1}$ and $T_{2}$ times in superconducting qubits have been studied extensively, e.g. in \cite{muller2015interacting, klimov2018fluctuations}. It was observed that such non-negligible temporal variations occur on timescales much longer ($\sim$ 1 second) than the timescale of a single QEC cycle. Therefore, it is sensible to consider the case where the noise parameter does not vary considerably within a single QEC cycle, and hence we expect $\hat{\theta}-\theta$ to also be small for each QEC cycle, and its effect only accumulates on the timescales of multiple QEC cycles, as observed e.g. in \cite{cramer2016repeated}. Hence, in the rest of the manuscript, we provide lower bounds for the spectator-based recovery protocol in this limit, unless stated otherwise.

We now present the main result of the article, which describes a metrological lower bound to the performance of the spectator-based recovery protocol in Fig.~\ref{fig:circuit2}, compared to the perfect knowledge case in Fig.~\ref{fig:circuit1}. For simplicity, this result applies when no nuisance parameters are present. A discussion of how the nuisance parameters will impact the lower bound is given later is Section~\ref{subsec:nuisance}.
\begin{theorem} \label{th:spec_QFI}
    Consider a parameterized noise channel $\mathcal{N}^{A \rightarrow B}_{\theta}$, and the corresponding optimal $\mathcal{R}_{\theta}^{B \rightarrow A}$ and best-guess $\mathcal{R}_{\hat{\theta}}^{B \rightarrow A}$ recovery maps. For small deviations $\hat{\theta} -\theta$ of the locally unbiased estimate $\hat{\theta}$ from the true value $\theta$, the difference between the corresponding entanglement fidelities of recovery is lower bounded by 
    \begin{align}
        &\mathbb{E}\left[F_{e}(\mathcal{R}_{\theta}\circ \mathcal{N}_{\theta})-F_{e}(\mathcal{R}_{\hat{\theta}}\circ \mathcal{N}_{\theta}) \right]_{p(x\vert \theta)} \nonumber \\  
        & \hspace{0.4cm} \ge \frac{g(\theta)}{\textsf{I}_{\operatorname{QF}}(\mathcal{M}_{\theta}(\psi))}-\mathbb{E}\left[R(\hat{\theta}-\theta)\right]_{p(x\vert \theta)} \label{eqn:thm_2_eqn_1}\; ,
    \end{align}
    where $g(\theta)$ denotes
    \begin{align}
        g(\theta) = -\frac{d_{B}}{2d_{A}}\operatorname{Tr}_{AB}\left[\left( \Phi_{AB}^{\mathcal{N}_{\theta}} \right)^{T} \partial_{\theta}^{2}\Phi^{\mathcal{R}_{\theta}}_{BA} \right] \label{eqn:thm_2_eqn_3} \; .
    \end{align}
    Further, $\textsf{I}_{\operatorname{QF}}(\mathcal{M}_{\theta}(\psi))$ denotes the $\operatorname{QFI}$ of the spectator dynamics $\mathcal{M}_{\theta}(\psi)$ (initialized in state $\psi$) and $\mathbb{E}(\cdot)_{p(x\vert \theta)}$ denotes the expectation with respect to the spectator's measurement statistics $p_{X}(x\vert \theta)$, corresponding to the measurement outcomes $x \in \mathcal{X}$ of the spectator observable $X=\sum_{x\in \mathcal{X}}x\Pi_{x}$.  
    Finally, $R(\hat{\theta}-\theta)\equiv \frac{1}{3!}\partial^{3}_{\nu}F_{e}(\mathcal{R}_{\theta + \nu_{0}}\circ \mathcal{N}_{\theta})(\hat{\theta}-\theta)^{3}$ is the Lagrange remainder of the Taylor series expansion of $F_{e}(\mathcal{R}_{\theta +\nu}\circ \mathcal{N}_{\theta})$ with respect to $\nu$, where $\nu_{0} \in [0, \hat{\theta}-\theta]$ is a constant. 
\end{theorem}
\begin{proof}
    The first part of the proof follows directly from Taylor expanding the entanglement fidelity $F_{e}(\mathcal{R}_{\theta +\nu}\circ \mathcal{N}_{\theta})$ with respect to the difference $\nu \equiv \hat{\theta}-\theta$ to the second order, and using the Lagrange form for the remainder, as follows
    \begin{align}
        & F_{e}(\mathcal{R}_{\theta}\circ \mathcal{N}_{\theta})-F_{e}(\mathcal{R}_{\hat{\theta}}\circ \mathcal{N}_{\theta}) \\
        & = F_{e}(\mathcal{R}_{\theta}\circ \mathcal{N}_{\theta})-F_{e}(\mathcal{R}_{\theta+\nu}\circ \mathcal{N}_{\theta}) \\ 
        & = \frac{1}{1!}\left. \left( -\frac{d}{d \nu} F_{e}(\mathcal{R}_{\theta +\nu}\circ \mathcal{N}_{\theta}) \right) \right \vert_{\nu = 0} \nu \nonumber \\ 
        & \hspace{0.2cm} + \frac{1}{2!} \left. \left(- \frac{d^{2}}{d\nu^{2}} F_{e}(\mathcal{R}_{\theta +\nu}\circ \mathcal{N}_{\theta}) \right) \right \vert_{\nu = 0} \nu^{2} \nonumber \\ & \hspace{0.4cm} - \frac{1}{3!}\frac{d^{3}}{d\nu^{3}} F_{e}(\mathcal{R}_{\theta +\nu_{0}}\circ \mathcal{N}_{\theta}) \nu^{3} \; ,
    \end{align}
    where $\nu_{0} \in [0, \nu]$ is a constant. Taking the expectation $\mathbb{E}[\cdot]_{p(x\vert \theta)}$ of both sides with respect to the spectator's measurement statistics $p_{X}(x\vert \theta)$ of the observable $X$, and recalling that $\hat{\theta}$ is a locally unbiased estimate of $\theta$ (and hence the QCRB applies), we see that
    \begin{equation}
        \mathbb{E}[\nu]_{p(x \vert \theta)}=0 \hspace{0.5cm} \text{and} \hspace{0.5cm} \mathbb{E}[\nu^{2}]_{p(x \vert \theta)} \ge \frac{1}{\textsf{I}_{\operatorname{QF}}(\mathcal{M}_{\theta}(\psi))} \; . \label{eqn:QCRB}
    \end{equation}
    This yields Eq.~\eqref{eqn:thm_2_eqn_1} of our theorem. To prove Eq.~\eqref{eqn:thm_2_eqn_3} of this theorem, we first show in Lemma~\ref{le:fid_Choi} of Appendix~\ref{apx:Choi_techniques} that the entanglement fidelity of the composite dynamics $\mathcal{R}_{\theta+\nu}\circ \mathcal{N}_{\theta}$ of the memory qubit is given by the individual Choi states of the noise and recovery maps, as follows
    \begin{equation}
        F_{e}(\mathcal{R}_{\theta+\nu}\circ \mathcal{N}_{\theta})=\frac{d_{B}}{d_{A}}\operatorname{Tr}_{AB}\left[\left( \Phi_{AB}^{\mathcal{N}_{\theta}} \right)^{T} \Phi^{\mathcal{R}_{\theta+\nu}}_{BA} \right] \; . \label{eqn:approx_Choi_int_fid}
    \end{equation}
    Then, Eq.~\eqref{eqn:thm_2_eqn_3} of our theorem is a direct consequence of differentiating this entanglement fidelity formula twice with respect to $\nu$ (assuming the Choi state of the optimal recovery map $\mathcal{R}_{\theta}$ is twice differentiable) for a fixed $\theta$, which yields
    \begin{align}
        \partial_{\nu}\Phi^{\mathcal{R}_{\theta+ \nu}} \vert_{\nu = 0} &=\partial_{\theta+\nu}\Phi^{\mathcal{R}_{\theta+\nu}}\vert_{\theta+\nu=\theta} \\ &=\partial_{\theta^{\prime}}\Phi^{\mathcal{R}_{\theta^{\prime}}}\vert_{\theta^{\prime}=\theta} \equiv \partial_{\theta}\Phi^{\mathcal{R}_{\theta}} \; . 
    \end{align}
\end{proof}


We note that the saturation of this lower bound is based on the saturation of the QCRB, and is discussed in Appendix~\ref{apx:QFIM}. Moreover, although Theorem~\ref{th:spec_QFI} is true for any input state $\psi$ of the spectator system, we would like to use the optimal probe state for the corresponding dynamics $\mathcal{M}_{\theta}$. Furthermore, a sufficient condition for the remainder term in Eq.~\eqref{eqn:thm_2_eqn_1} to be negligible is presented in Appendix~\ref{apx:remainder}. Further, we note that this theorem could also be interpreted from the information geometric perspective, e.g. in \cite{liu2020quantum, sidhu2020geometric}. 

\begin{remark}
    Due to the tensor product property of the QFI, $\textsf{I}_{\operatorname{QF}}(\sigma_{\theta}^{\otimes m})=m\textsf{I}_{\operatorname{QF}}(\sigma_{\theta})$, the QCRB yields zero variance only in the asymptotic limit. However, in realistic spectator-based recovery protocols (described by Fig.~\ref{fig:circuit2}), the asymptotic limit (i.e. implementing $m \rightarrow \infty$ spectator qubits) will necessarily mean that spatial variations of the noise parameter $\theta$ affecting the spectator qubits cannot be neglected. Therefore, we limit ourselves to the non-zero QCRB variance (finite sample case) and instead attempt to saturate this bound via optimal measurements and initial spectator state (see Appendix~\ref{apx:saturation}).
\end{remark}

Finally, note that the expected dependence of the spectator's contribution separates into the product of two functions: the first, $g(\theta)$, depends on the full dynamics of the memory system, and the second, $\operatorname{Var}(\hat{\theta})$,  depends on the full dynamics of the spectator system. It turns out that the function $g(\theta)$ could be computed analytically for simple single-qubit channels, such as for the amplitude-damping channel \cite{zhan2013entanglement}.

\subsubsection{Comparison With The Non-Adaptive Case}
\begin{figure}[t]
    \centering
    \includegraphics[width=0.95\linewidth]{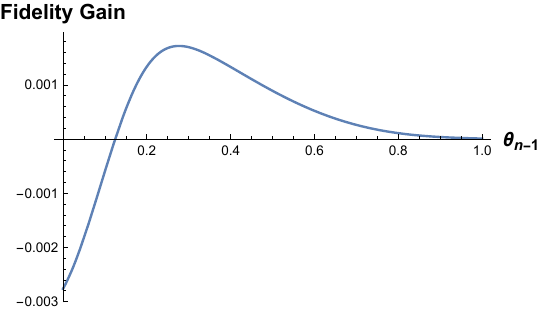}
    \caption{Plot of the fidelity gain Eq.~\eqref{eqn:adapt_adv} due to using a spectator-based recovery protocol for the [4,1] amplitude damping code (see Section~\ref{sec:AD_application_1}). The variation $\theta_{n-1} \rightarrow \theta_{n}$ is modeled to be a truncated Gaussian distribution with a maximum at $\theta_{n-1}$, and 0.1 standard deviation. The region of outperformance under these conditions is $\Theta_{\text{spec}}=[0.125, 1] \subset [0, 1]$.}
    \label{fig:adapt_regime}
\end{figure}

As discussed previously in Section~\ref{sub:validity}, the relative advantage of implementing a spectator-based recovery protocol depends on the characteristics of the noise. For example, if the noise is completely static, then implementing a spectator-based recovery will always be worse than simply characterizing the noise before the experiment (e.g. via process tomography), as there is no advantage to real-time quantum sensing of the noise, where only sparse data is available. Here, we provide a sufficient condition for a spectator-based adaptive protocol to outperform other (non-adaptive) recovery protocols for arbitrary finite estimation error $\hat{\theta}-\theta$.

Consider some parameter variation $\theta_{n-1}\rightarrow \theta_{n}$ from the $(n-1)$-th to the $n$-th QEC cycle. This corresponds to the change in the noise $\mathcal{N}_{\theta_{n-1}} \rightarrow \mathcal{N}_{\theta_{n}}$. In a spectator-based recovery protocol, this change is tracked via the spectator system, and the best-guess recovery is updated accordingly $\mathcal{R}_{\hat{\theta}_{n-1}} \rightarrow \mathcal{R}_{\hat{\theta}_{n}}$. The performance of this protocol will hence be quantified by the entanglement fidelity $F_{e}(\mathcal{R}_{\hat{\theta}_{n}}\circ \mathcal{N}_{\theta_{n}})$. On the other hand, a non-adaptive protocol will include applying a recovery map $\mathcal{R}_{\theta_{n-1}}$ that is (in the ideal case) optimal for the previous noise channel, i.e. $\mathcal{N}_{\theta_{n-1}}$. The performance of an ideal non-adaptive protocol will hence be quantified by the entanglement fidelity $F_{e}(\mathcal{R}_{\theta_{n-1}}\circ \mathcal{N}_{\theta_{n}})$. Therefore, to find a sufficient condition for adaptation to yield an advantage, we need to find a non-negative lower bound to the entanglement fidelity difference
\begin{equation}
     F_{e}(\mathcal{R}_{\hat{\theta}_{n}}\circ \mathcal{N}_{\theta_{n}})- F_{e}(\mathcal{R}_{\theta_{n-1}}\circ \mathcal{N}_{\theta_{n}}) \; . \label{eqn:adapt_adv}
\end{equation}
We accomplish this by using Theorem~\ref{th:finite_var} and Eq.~\ref{eqn:main_spec_lower}, as follows
\begin{align}
    &F_{e}(\mathcal{R}_{\hat{\theta}_{n}}\circ \mathcal{N}_{\theta_{n}})- F_{e}(\mathcal{R}_{\theta_{n-1}}\circ \mathcal{N}_{\theta_{n}}) \nonumber \\
    & = \left[F_{e}(\mathcal{R}_{\theta_{n}}\circ \mathcal{N}_{\theta_{n}})- F_{e}(\mathcal{R}_{\theta_{n-1}}\circ \mathcal{N}_{\theta_{n}})\right] \nonumber \\
    & \hspace{0.2cm} - \left[F_{e}(\mathcal{R}_{\theta_{n}}\circ \mathcal{N}_{\theta_{n}})- F_{e}(\mathcal{R}_{\hat{\theta}_{n}}\circ \mathcal{N}_{\theta_{n}})\right] \\
    & \ge \left \Vert \Phi^{\mathcal{R}_{\theta_{n-1}}}-\Phi^{\mathcal{R}_{\theta_{n}}} \right \Vert_{\alpha} \times  \left \Vert \Phi^{\mathcal{N}_{\theta_{n}}} \right \Vert_{\beta} \nonumber \\
    & \hspace{0.2cm}-\frac{1}{2}\left\Vert \mathcal{R}_{\theta_{n}}\circ \mathcal{N}_{\theta_{n}} -\mathcal{R}_{\hat{\theta}_{n}}\circ \mathcal{N}_{\theta_{n}} \right\Vert_{\diamond} \\
    & \ge \left \Vert \Phi^{\mathcal{R}_{\theta_{n-1}}}-\Phi^{\mathcal{R}_{\theta_{n}}} \right \Vert_{\alpha} \times  \left \Vert \Phi^{\mathcal{N}_{\theta_{n}}} \right \Vert_{\beta} \nonumber \\ & \hspace{0.2cm}-\frac{1}{2}\left\Vert \mathcal{R}_{\theta_{n}}-\mathcal{R}_{\hat{\theta}_{n}}\right\Vert_{\diamond} \; .
\end{align}
Therefore, a \textit{sufficient} (and initial state independent) condition for the spectator-based recovery protocol to be advantageous for an arbitrary estimation error $\hat{\theta}-\theta$, compared to an ideal non-adaptive protocol, is given by
\begin{equation}
    \frac{1}{2}\left\Vert \mathcal{R}_{\theta_{n}}-\mathcal{R}_{\hat{\theta}_{n}}\right\Vert_{\diamond} \leq c \times \left \Vert \Phi^{\mathcal{R}_{\theta_{n-1}}}-\Phi^{\mathcal{R}_{\theta_{n}}} \right \Vert_{\alpha} \; ,
\end{equation}
where $c \equiv \left \Vert \Phi^{\mathcal{N}_{\theta_{n}}} \right \Vert_{\beta}$. It is easy to see that if there is no change, i.e. $\theta_{n}=\theta_{n-1}$, then this condition is not satisfied. On the other hand, for a general (Markovian) stochastic jump model, with a $\theta_{n-1} \rightarrow \theta_{n}$ transition probability of $p(\theta_{n}\vert \theta_{n-1})$, there is a range of values $\Theta_{\text{spec}} \subseteq \Theta$ of the noise parameter $\theta_{n-1} \in \Theta$ such that the spectator based recovery exhibits an advantage. Furthermore, this advantage will accumulate over multiple QEC cycles. As an example, in Fig.~\ref{fig:adapt_regime}, we plot the fidelity gain in Eq.~\eqref{eqn:adapt_adv} for the [4,1] amplitude damping code (see the following section), showcasing the region of outperformance of the spectator-based recovery protocol over standard QEC \cite{fletcher2008channel} for a stroboscopically varying noise $\theta_{n-1} \rightarrow \theta_{n}$. For more details, please see Appendix~\ref{apx:[4,1]_ent_fid}.

\section{Application to The [4, 1] Code of The Amplitude-Damping Channel} \label{sec:AD_application_1}
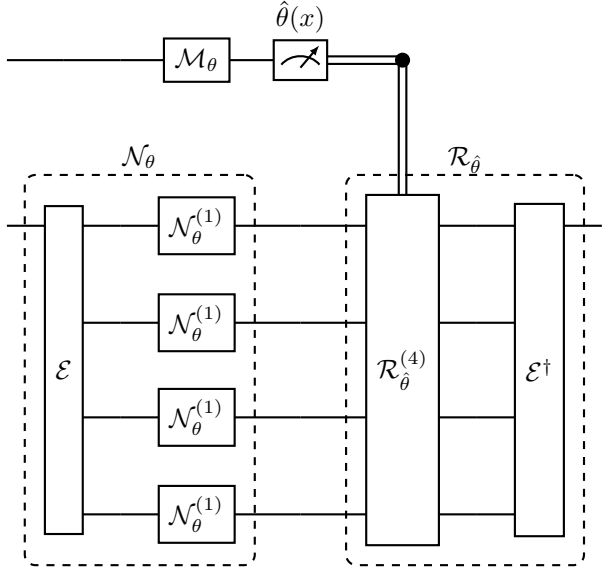
\begin{figure}[t]
    \begin{quantikz}
        & \qw & \qw & \gate{\mathcal{M}_{\theta}} & \meter{$\hat{\theta}(x)$} & \cwbend{3} & \\
        & & & & & & & & \\
        & & & & & & & & \\
        & \gate[4, nwires={2, 3, 4}]{\mathcal{E}}\gategroup[4, steps=3, style={dashed, rounded corners}]{$\mathcal{N}_{\theta}$} & \qw & \gate{\mathcal{N}_{\theta}^{(1)}} & \qw & \gate[4]{\mathcal{R}^{(4)}_{\hat{\theta}}}\gategroup[4, steps=3, style={dashed, rounded corners}]{$\mathcal{R}_{\hat{\theta}}$} & \qw & \gate[4]{\mathcal{E}^{\dagger}} & \qw \\
        & & \qw & \gate{\mathcal{N}_{\theta}^{(1)}} & \qw & & \qw & &  \\
        & & \qw & \gate{\mathcal{N}_{\theta}^{(1)}} & \qw  & & \qw & & \\
        & & \qw & \gate{\mathcal{N}_{\theta}^{(1)}} & \qw & & \qw & &
    \end{quantikz}
    \caption{Spectator-based [4, 1] code of the amplitude-damping channel (time flows from left to right). A single logical qubit (second register) is encoded into four physical qubits using the encoding channel $\mathcal{E}:\mathcal{H}_{2}\rightarrow \mathcal{H}_{2}^{\otimes 4}$ in Eqs.~\eqref{eqn:AD_enc_1} and \eqref{eqn:AD_enc_2}, where $\mathcal{H}_{2}$ denotes the two dimensional Hilbert space of a single qubit system. The spectator system (first register) performs an unbiased estimate $\hat{\theta}(x)$ of the noise parameter $\theta$ using the POVM $\{\Pi_{x}\}$. The estimated value $\hat{\theta}$ is fed into the recovery operation described in detail in \cite{fletcher2008channel} that is adapted for the amplitude-damping channel.}
    \label{fig:[4,1]_circuit}
\end{figure}

In what follows, we derive the entanglement fidelity $F_{e}(\mathcal{R}_{\hat{\theta}}\circ \mathcal{N}_{\theta})$ for the $[4,1]$ code of the amplitude-damping (AD) channel analytically, following the approach developed in \cite{rahn2002exact}, and extending the derivation in \cite{zhan2013entanglement} to the incomplete knowledge recovery scenario. It is worth noting that analytical approaches to the AD channel have also been taken previously e.g. in \cite{cafaro2014approximate, cafaro2014simple}.
 
Since the AD channel is covariant with respect to the group $\{I, Z\}$, the Eastin-Knill theorem \cite{eastin2009restrictions} guarantees that no perfect QEC codes exist. However, approximate codes for the AD channel have been developed in \cite{leung1997approximate} and later on, channel-adapted codes have been developed \cite{fletcher2008channel}, where the recovery depends on the value of the noise parameter. The developed techniques have also been extended beyond the $[4,1]$ code and towards more general $[2k+1, k]$ codes \cite{leung1997approximate, fletcher2008channel} (where $k$ logical qubits are encoded into $n=2k+1$ physical/memory qubits).

\subsection{The Amplitude-Damping Channel}
The single-qubit AD channel is defined as $\mathcal{N}_{\theta}^{(1)}(\cdot)=N_{0}(\cdot)N_{0}^{\dagger}+N_{1}(\cdot)N_{1}^{\dagger}$, where
\begin{equation}
N_{0}= \begin{pmatrix} 1 & 0 \\ 0 & \sqrt{1-\theta} \end{pmatrix}, \hspace{0.4cm} N_{1}= \begin{pmatrix} 0 & \sqrt{\theta} \\ 0 & 0 \end{pmatrix} \; . \label{eqn:AD_channel}
\end{equation}
The Kraus operators $N_{0}$, $N_{1}$ are often called the ``no-damping'' and ``damping'' errors, respectively.
Here, the noise parameter $\theta(t)=1-\exp{(-t/T_{1})}$ depends on time $t$ and the relaxation time $T_{1}$ \cite{etxezarreta2021time, nielsen2002quantum}. We follow the usual notation in quantum information, where the dependence of the noisy channel (and hence also the noise parameter) on time is suppressed.

\subsection{The Approximate [4,1] Code}
Assuming an independent noise model, we recall the encoding $\mathcal{E}:\mathcal{D}(\mathcal{H})\rightarrow \mathcal{D}(\mathcal{C})$ of the [4,1] code \cite{leung1997approximate} from a 1-qubit physical state to a 4-qubit logical state, where $\mathcal{C}=\text{span}\{|0_{L}\rangle, |1_{L}\rangle \} \ \subset \mathcal{H}^{\otimes 4}$, as follows
\begin{align}
    & |0\rangle \rightarrow |0_{L}\rangle \coloneqq \frac{1}{\sqrt{2}}\left(|0000\rangle + |1111\rangle \right) \label{eqn:AD_enc_1} \\
    & |1\rangle \rightarrow |1_{L}\rangle \coloneqq \frac{1}{\sqrt{2}}\left(|1100\rangle + |0011\rangle \right) \; , \label{eqn:AD_enc_2}
\end{align}
and hence $\mathcal{E}(\cdot)\coloneqq C(\cdot)C^{\dagger}$, where $C=|0_{L}\rangle \! \langle 0|+|1_{L}\rangle \! \langle 1| $. The encoded Pauli operators $\sigma_{\text{enc}}=\mathcal{E}(\sigma)$ for $\sigma \in \{ I, X, Y, Z\}$ read
\begin{align}
    & I_{\text{enc}} = |0_{L}\rangle \! \langle 0_{L}|+|1_{L}\rangle \! \langle 1_{L}| \\
    & X_{\text{enc}} = |0_{L}\rangle \! \langle 1_{L}|+|1_{L}\rangle \! \langle 0_{L}| \\
    & Y_{\text{enc}} = -i|0_{L}\rangle \! \langle 1_{L}|+i|1_{L}\rangle \! \langle 0_{L}| \\
    & Z_{\text{enc}} = |0_{L}\rangle \! \langle 0_{L}|-|1_{L}\rangle \! \langle 1_{L}| \; .
\end{align}
By definition, the encoded Pauli operators only act on states in the codespace $\mathcal{C}$. However, in the stabilizer formalism, the logical Pauli operators $I_{L}$, $X_{L}$, $Y_{L}$, and $Z_{L}$ are defined on the full 4-qubit Hilbert space. For example, the generators of the stabilizer set for the $[4,1]$ code is given by $S=\{S_{j}\}_{j=1}^{3}=\{XXXX, ZZII, IIZZ\}$, along with the logical Pauli operators $X_{L}=XXII$, $Y_{L}=YXZI$, and $Z_{L}=ZIZI$. The link between the encoded and logical Pauli operators is found by restricting the action of the latter to the codespace. Namely, $\sigma_{\text{enc}}=\mathcal{E}(\sigma)=\Pi \sigma_{L}$, where $\Pi=\sum_{j=1}^{3}S_{j}/\vert S \vert$ is the projection onto the codespace $\mathcal{C}$ corresponding to the set of stabilizers \cite{rahn2002exact}.

We define the noisy channel $\mathcal{N}_{\theta}$ to be the physical noise experienced by the four physical qubits in the [4,1] code post-encoding, as follows
\begin{equation}
    \mathcal{N}_{\theta}=\left(\mathcal{N}_{\theta}^{(1)}\otimes\mathcal{N}_{\theta}^{(1)}\otimes\mathcal{N}_{\theta}^{(1)}\otimes\mathcal{N}_{\theta}^{(1)}\right) \circ \mathcal{E} \; .
\end{equation}
Furthermore, we define the decoding recovery channel $\mathcal{R}_{\hat{\theta}}$ to be given by
\begin{equation}
    \mathcal{R}_{\hat{\theta}}=\mathcal{E}^{\dagger}\circ \mathcal{R}_{\hat{\theta}}^{(4)} \; ,
\end{equation}
where $\mathcal{R}_{\hat{\theta}}^{(4)}$ is taken from Table~1 of \cite{fletcher2008channel}, which is the channel-adapted recovery of the AD channel (see Fig.~\ref{fig:[4,1]fund}(b) for the performance of this recovery).

\subsection{Entanglement Fidelity}
\begin{figure*}[t]
    \centering
    \includegraphics[width=17cm]{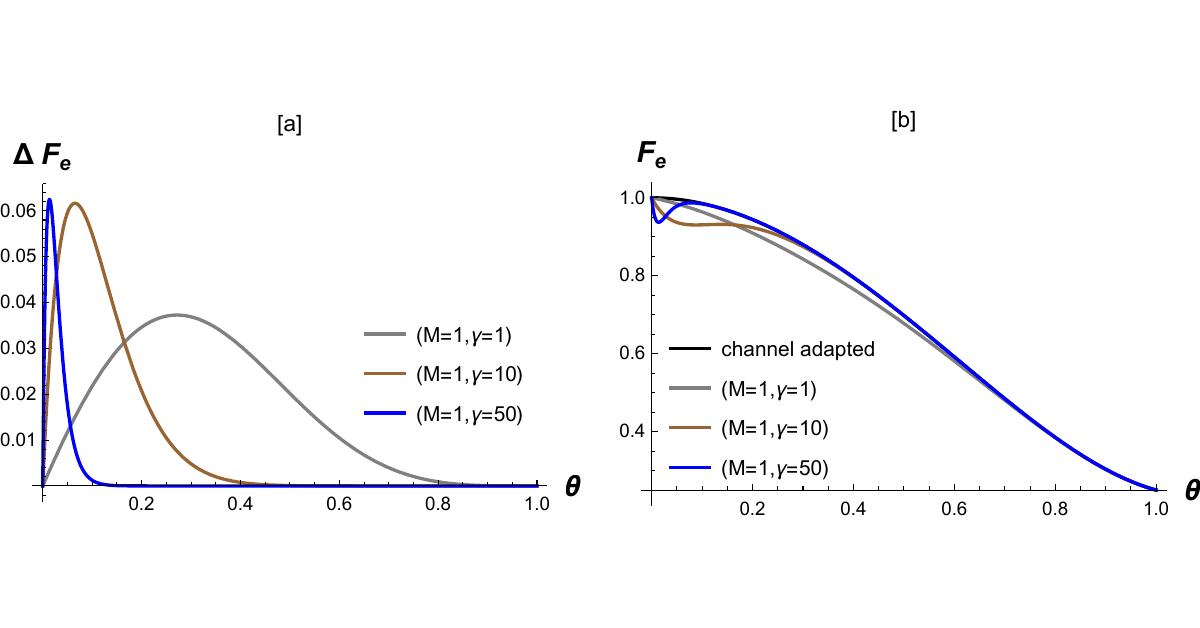}
    \caption{The spectator system is taken to be a single qubit ($m=1$) with varying values of the physical parameter $\gamma$ in Eq.~\eqref{eqn:f_spec}. Both subfigures consider the channel-adapted approximate [4, 1] code of the amplitude damping channel \cite{fletcher2008channel}. (a) Entanglement fidelity difference between the cases of perfect and incomplete knowledge recovery protocols. (b) Comparison between the entanglement fidelities for perfect \cite{fletcher2008channel, zhan2013entanglement} and incomplete knowledge recovery protocols. In both figures, we assume the best-case scenario where the spectator system saturates the QCRB during parameter estimation. 
    }
    \label{fig:[4,1]fund}
\end{figure*}

In Appendix~\ref{apx:[4,1]_ent_fid}, we analytically calculate the numerator $g(\theta)$ in Theorem~\ref{th:spec_QFI} for the [4,1] AD code to be
\begin{equation}
    g(\theta)= \frac{(1-\theta)^{3}}{\sqrt{2}(1+(1-\theta)^{4})^{3/2}} \; . \label{eqn:h_spec}
\end{equation}

It has been shown that the RLD QFI for the AD channel diverges (see the example discussed in \cite{katariya2021geometric} for generalized AD channels), however, the SLD QFI is finite and known to be equal to $\textsf{I}_{\operatorname{QF}}(\mathcal{N}^{(1)}_{\theta})=1/(\theta (1-\theta))=\textsf{I}_{\operatorname{QF}}(\mathcal{N}^{(1)}_{\theta}(\psi))$ for $\psi=|1\rangle \! \langle 1|$ \cite{fujiwara2004estimation}. Consequently, a spectator qubit that satisfies the condition Eq.~\eqref{eqn:memo_to_spec} has an SLD QFI of
\begin{equation}
   \textsf{I}_{\operatorname{QF}}(\mathcal{M}^{(1)}_{\theta}(\psi))=\frac{1}{f_{\gamma}(\theta) (1-f_{\gamma}(\theta))} \; , \label{eqn:QFI_AD}
\end{equation}
where $f_{\gamma}(\theta)$ is given by Eq.~\eqref{eqn:f_spec} for the AD code (or more generally, by Eq.~\eqref{eqn:memo_to_spec}). If the spectator system is made out of $m$ qubits, then the QFI scales linearly with $m$ due to the QFI property $\textsf{I}_{\operatorname{QF}}(\mathcal{M}_{\theta}^{(1)}(\psi)^{\otimes m})=m\textsf{I}_{\operatorname{QF}}(\mathcal{M}_{\theta}^{(1)}(\psi))$, assuming an independent noise model (see Eq.~\eqref{eqn:M_spec}). Note that we can only realistically improve the QCRB to a certain degree by increasing $m$, without dropping the negligible spatial variability assumption of the noise parameter $\theta$ \cite{gupta2020integration}.

Combining Eqs.~\eqref{eqn:h_spec} and \eqref{eqn:QFI_AD} for the [4,1] AD code with Theorem~\ref{th:spec_QFI} yields (for small $\hat{\theta}-\theta$)
\begin{equation}
    \mathbb{E}[ \Delta F_{e}] \ge \frac{ g(\theta)}{mf_{\gamma}(\theta)(1-f_{\gamma}(\theta))} \; . \label{eqn:AD_lower}
\end{equation}
Therefore, the contribution of the spectator system to the entanglement fidelity of the $[4, 1]$ code of the amplitude-damping channel is determined by two parameters: the number of spectator qubits used ($m$) and their physical nature ($\gamma=T_{1}^{\text{memo}}/T_{1}^{\text{spec}}$). When the QCRB is saturated, the resulting entanglement fidelity is illustrated in Fig.~\ref{fig:[4,1]fund} for various values of the spectator parameter $\gamma$.

\subsection{The Effect of Nuisance Parameters}
\label{subsec:nuisance}
In this section, we discuss the effects of nuisance parameters on the lower bound in Theorem~\ref{th:spec_QFI} for the AD code (i.e. Eq.~\eqref{eqn:AD_lower}). We consider three different physical choices of a spectator qubit, which yield one of the following:
\begin{enumerate}
    \item The presence of an additional constant magnetic field $B$
    \begin{equation}
        \mathcal{U}_{\phi}(\cdot)=U^{-i\phi Z}(\cdot)U^{i\phi Z} \; ,
    \end{equation}
    with noise parameter $\phi=\gamma_{\text{spec}} B t$, where $\gamma_{\text{spec}}$ here is the gyromagnetic ratio of the spectator qubit.
    \item The presence of an additional pure dephasing noise
    \begin{equation}
        \mathcal{P}(\cdot)=P_{1}(\cdot)P_{1}^{\dagger}+P_{2}(\cdot)P_{2}^{\dagger} \; ,
    \end{equation}
    where the Kraus operators $P_{1}$ and $P_{2}$ are given by
    \begin{equation}
        P_{1}=
        \begin{pmatrix}
            1 & 0 \\
            0 & \sqrt{1-\lambda} 
        \end{pmatrix}
        \hspace{0.1cm} \text{;} \hspace{0.1cm} 
        P_{2}=
        \begin{pmatrix}
            0 & 0 \\
            0 & \sqrt{\lambda} 
        \end{pmatrix}
        \; ,
    \end{equation}
    with noise parameter $\lambda=1-\exp{(-t/T_{\varphi})}$, which yields the depahsing time $T_{2}$ with an off-diagonal decay rate of $1/T_{2}=1/2T_{1}+1/T_{\varphi}$ \cite{allen1987optical}.
    \item The presence of an additional depolarizing noise
    \begin{equation}
        \mathcal{D}_{q}(\cdot)=(1-q)(\cdot)+q\frac{I}{2} \; ,
    \end{equation}
    with noise parameter $q \in [0,1]$.
\end{enumerate}
When no noise parameters are present, the quantum state of the spectator system
\begin{equation}
    \psi = 
    \begin{pmatrix}
        \psi_{00} & \psi_{01} \\
        \psi_{10} & \psi_{11} 
    \end{pmatrix}
    \; , 
\end{equation}
is transformed to $\mathcal{M}^{(1)}_{\theta}(\psi)=\mathcal{N}^{(1)}_{f(\theta)}(\psi)$, where
\begin{equation}
    \mathcal{N}^{(1)}_{f(\theta)}(\psi)=
    \begin{pmatrix}
    \psi_{00}+f(\theta) \psi_{11} & \psi_{01}\sqrt{1-f(\theta)} \\
    \psi_{10}\sqrt{1-f(\theta)} & (1-f(\theta))\psi_{11} 
    \end{pmatrix} 
    \; .
\end{equation}
However, when a nuisance parameter is present, the output state of the spectator used for quantum estimation limit of $\theta$ will be modified. In the above three cases, the spectator states are given by, respectively,
\begin{equation}
    \mathcal{M}^{(1)}_{\theta, \phi}(\psi)=
    \begin{pmatrix}
    \psi_{00}+f(\theta) \psi_{11} & \psi_{01}e^{-i\phi}\sqrt{1-f(\theta)} \\
    \psi_{10}e^{i\phi}\sqrt{1-f(\theta)} & (1-f(\theta))\psi_{11} 
    \end{pmatrix} 
    \; ,
\end{equation}
\begin{align}
    &\mathcal{M}^{(1)}_{\theta, \lambda}(\psi)= \nonumber \\ 
    & =
    \begin{pmatrix}
    \psi_{00}+f(\theta) \psi_{11} & \psi_{01}(1-\lambda)\sqrt{1-f(\theta)} \\
    \psi_{10}(1-\lambda)\sqrt{1-f(\theta)} & (1-f(\theta))\psi_{11} 
    \end{pmatrix} 
    \; ,
\end{align}
and
\begin{equation}
    \mathcal{M}^{(1)}_{\theta, q}(\psi)=
    q\mathcal{M}^{(1)}_{\theta}(\psi)+(1-q)\frac{I}{2}
    \; .
\end{equation}
The quantum estimation limit of the parameter of interest $\theta$, in the presence of one of the nuisance parameters $\zeta=\{\phi, \lambda, q\}$ above, is given by the partial QFIM $ \textsf{I}_{\theta \vert \zeta}$ via $\operatorname{Var}(\hat{\theta})\ge 1/\textsf{I}_{\theta \vert \zeta}$, where
\begin{equation}
    \textsf{I}_{\theta \vert \zeta} = \textsf{I}_{\theta, \theta}-\textsf{I}_{\theta, \zeta} \left( \textsf{I}_{\zeta, \zeta}\right)^{-1}\textsf{I}_{\zeta, \theta} \leq \textsf{I}_{\theta, \theta} \; ,
\end{equation}
and the right hand side are the block matrices of the QFIM (see Appendix~\ref{apx:QFIM} for more details)
\begin{equation}
    \textsf{I}_{\operatorname{QF}}\left( \mathcal{M}_{\theta,\zeta}^{(1)} (\psi)\right)=
    \begin{pmatrix}
        \textsf{I}_{\theta, \theta} & \textsf{I}_{\theta, \zeta} \\
        \textsf{I}_{\zeta, \theta} & \textsf{I}_{\zeta, \zeta}
    \end{pmatrix} \; .
\end{equation}
It is important to emphasize that the matrix element $\textsf{I}_{\theta, \theta}$ refers to the quantum estimation limit of $\theta$, when the noise parameter $\zeta$ \textit{is known}. Therefore, we expect to have $\lim_{\zeta \rightarrow 0}\textsf{I}_{\theta, \theta}=\textsf{I}_{\operatorname{QF}}(\mathcal{M}_{\theta}^{(1)}(\psi))$, which is easily verified numerically for the above three examples. However, a nuisance parameter $\zeta$, similar to the parameter of interest $\theta$, is, by nature, \textit{unknown}. Hence, the quantum estimation limit would be reduced by $ \textsf{I}_{\theta, \theta}- \textsf{I}_{\theta \vert \zeta}$ due to the presence of the unknown noise parameter $\zeta$, compared to when it is known.

We now present the impact of the nuisance parameter for the spectator qubit initial state $\psi=|1\rangle \! \langle 1|$ that maximizes the QFI $\textsf{I}_{\operatorname{QF}}(\mathcal{M}^{(1)}_{\theta}(\psi))$ for the AD channel $\mathcal{M}_{\theta}^{(1)}$ \cite{fujiwara2004estimation}. Simple computation shows that
\begin{enumerate}
    \item For $\zeta=\phi$, we have 
    \begin{equation}
        \textsf{I}_{\theta \vert \phi}= \textsf{I}_{\theta, \theta}=\frac{1}{f(\theta)(1-f(\theta))}=\textsf{I}_{\operatorname{QF}}(\mathcal{M}_{\theta}^{(1)}(\psi)) \; .
    \end{equation}
    \item For $\zeta=\lambda$, it still holds that 
    \begin{equation}
        \textsf{I}_{\theta \vert \lambda}= \textsf{I}_{\theta, \theta}=\frac{1}{f(\theta)(1-f(\theta))}=\textsf{I}_{\operatorname{QF}}(\mathcal{M}_{\theta}^{(1)}(\psi)) \; .
    \end{equation}
    This is a direct consequence of the fact that the optimal spectator input state $\psi=|1\rangle \! \langle 1|$ has no off-diagonal elements subject to dephasing.
    \item For $\zeta=q$, we have 
    \begin{equation}
        \textsf{I}_{\theta, \theta}=\frac{(1-q)^{2}}{(1-q)^{2}(f(\theta)(1-f(\theta))-\frac{1}{4})+\frac{1}{4}} \; ,
    \end{equation}
    which also yields $\lim_{q \rightarrow 0} \textsf{I}_{\theta, \theta}=\textsf{I}_{\operatorname{QF}}(\mathcal{M}_{\theta}^{(1)}(\psi))=1/(f(\theta)(1-f(\theta)))$. However, the partial QFI is computed to be $ \textsf{I}_{\theta \vert q}=0$.
\end{enumerate}


\section{Recovery Bounds in The Multi-Cycle Scenario} \label{sec:multi_cycle}

\subsection{The Multi-Cycle Case}

\begin{figure*}[!t]
	\centering
	\includegraphics[width=0.9\textwidth]{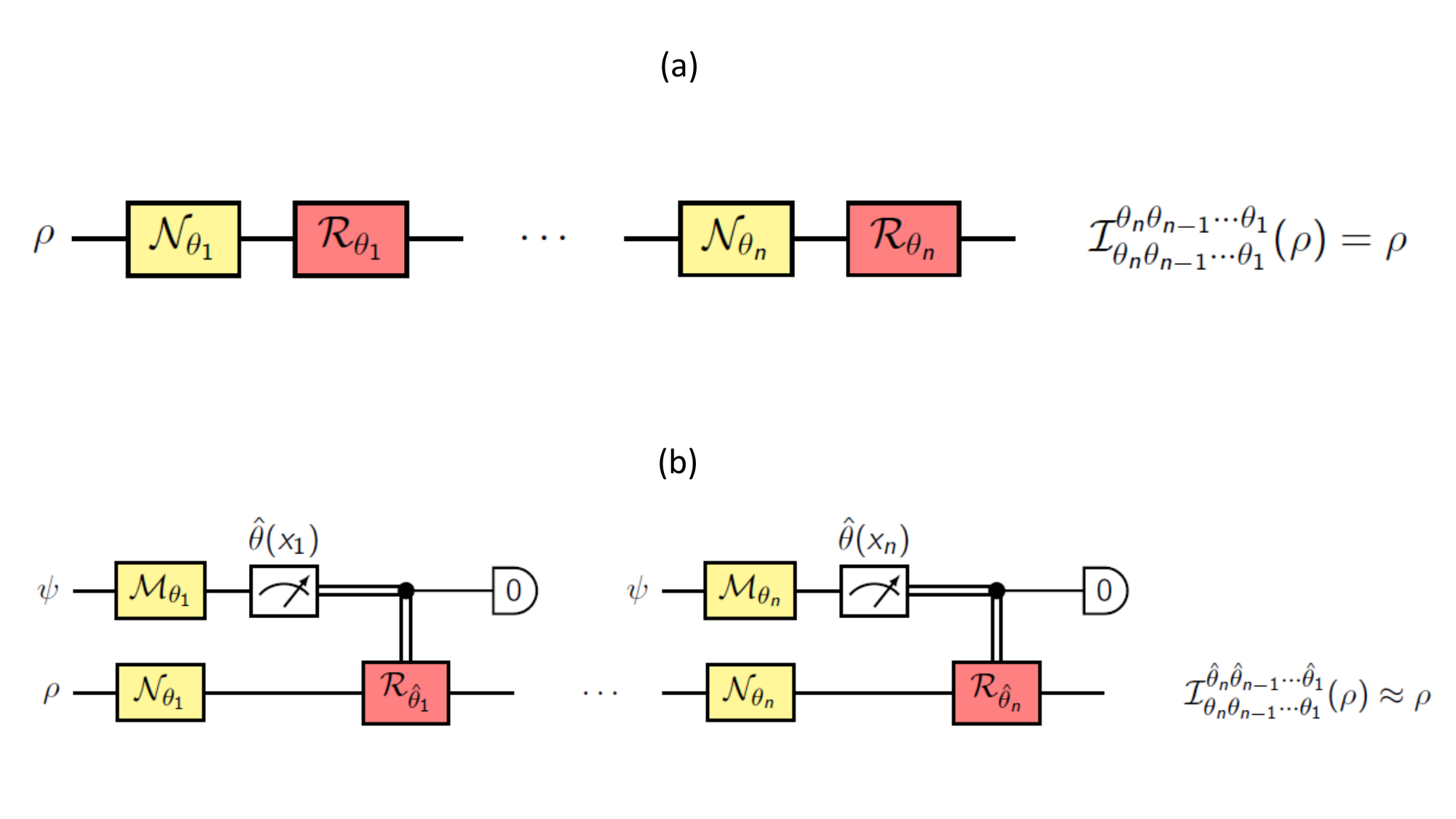}
    \caption{Time flows from left to right. (a) Multi-cycle recovery protocol with perfect knowledge. (b) Multi-cycle recovery protocol with incomplete knowledge. The input state of the quantum memory in both subfigures is given by $\rho$, whereas the input state of the spectator system in subfigure (b) is $\psi$. In the latter case, the state of the spectator is recycled back to $\psi$ after every recovery cycle, via a discarding and preparation channel. The final output states $\mathcal{I}^{\theta_{n}\theta_{n-1}\cdots \theta_{1}}_{\theta_{n}\theta_{n-1}\cdots \theta_{1}}(\rho)$ and $\mathcal{I}^{\hat{\theta}_{n}\hat{\theta}_{n-1}\cdots \hat{\theta}_{1}}_{\theta_{n}\theta_{n-1}\cdots \theta_{1}}(\rho)$ of the quantum memory follow the multi-cycle notation in Eq.~\eqref{eqn:multi_cycle_notation}. }
\end{figure*}

In the article, we have considered a single-cycle recovery, i.e. when the noisy channel $\mathcal{N}_{\theta}$ is applied only once. However, extensions to the multi-cycle regime are also important for real-time applications. A thorough study of the multi-cycle case is beyond the scope of the current article. However, here we present some useful bounds to stimulate future discussions. 

To start, consider a stroboscopically varying noise parameter $\theta$ 
\begin{equation}
    \theta_{1}\rightarrow \theta_{2} \rightarrow \cdots \rightarrow \theta_{n} \; ,
\end{equation}
where $n$ enumerates the recovery cycle in the multi-cycle protocol, executed in the time interval $[(n-1)\Delta t_{\mathcal{R}}, n\Delta t_{\mathcal{R}}]$, where $\Delta t_{\mathcal{R}} << \tau_{\theta}$ is the duration of a single recovery cycle, and $\tau_{\theta}$ is the characteristic time of the noise parameter $\theta$ (i.e. the expected time in which the value of $\theta$ will change appreciably, see Section~\ref{sub:validity}). The corresponding set of real-time spectator estimates of $\theta$ for these $n$ cycles is given by
\begin{equation}
    \hat{\theta}_{1}\rightarrow \hat{\theta}_{2} \rightarrow \cdots \rightarrow \hat{\theta}_{n} \; .
\end{equation}

We introduce the following shorthand notation for multi-cycle recovery protocols 
\begin{equation}
    \mathcal{I}^{\hat{\theta}_{n}\hat{\theta}_{n-1}\cdots \hat{\theta}_{1}}_{\theta_{n}\theta_{n-1}\cdots \theta_{1}}\coloneqq \mathcal{R}_{\hat{\theta}_{n}} \circ \mathcal{N}_{\theta_{n}} \circ \mathcal{R}_{\hat{\theta}_{n-1}} \circ \mathcal{N}_{\theta_{n-1}} \cdots \circ \mathcal{R}_{\hat{\theta}_{1}} \circ \mathcal{N}_{\theta_{1}} \; , \label{eqn:multi_cycle_notation}
\end{equation}
which is an $n$-cycle concatenation between the noisy channel (with changing noise parameter values in each timestep) and the corresponding best-guess recovery.

\subsection{Recurrence Inequalities for Composite Average Channel Fidelity} \label{sec:rec_ineq}
So far, we have found a lower bound on the desired distinguishability measure in terms of the composite channel entanglement fidelity for the single-cycle case. To extend to the multi-cycle scenario, one option is to consider the entanglement fidelity of $\mathcal{I}^{\hat{\theta}_{n}\hat{\theta}_{n-1}\cdots \hat{\theta}_{1}}_{\theta_{n}\theta_{n-1}\cdots \theta_{1}}$. However, a more insightful approach is to express this entanglement fidelity in terms of individual cycle fidelities. Specifically, this is accomplished by the use of the entanglement fidelities of $\mathcal{I}^{\hat{\theta}_{n-1}\cdots \hat{\theta}_{1}}_{\theta_{n-1}\cdots \theta_{1}}$ and $\mathcal{I}^{\hat{\theta}_{n}}_{\theta_{n}}$.  Bounding composite channel fidelities using individual channel fidelities has been studied previously in e.g. \cite{carignan2019bounding}. The following lemma is largely taken from \cite{carignan2019bounding}, using the $\chi$-matrix representation of quantum dynamics (see Appendix~\ref{sec:chi_matrix} for a self-contained review).

\begin{lemma} \label{le:err_angle}
    (\cite{carignan2019bounding}) Given the $\chi$-matrix elements $\chi_{00}^{\mathcal{Q}}$, $\chi_{00}^{\mathcal{S}}$ of the channels $\mathcal{Q}$, $\mathcal{S}$, respectively, the composite channel $\mathcal{S}\circ\mathcal{Q}$ $\chi$-matrix element $\chi^{\mathcal{S}\circ\mathcal{Q}}_{00}$ is bounded from above (and hence the corresponding error angle $\delta^{\mathcal{S}\circ \mathcal{Q}}$ is bounded from below), as follows
    \begin{equation}
        \frac{\chi^{\mathcal{S}\circ \mathcal{Q}}_{00}}{d} \leq \cos^{2} \left( \arccos{\sqrt{\frac{\chi^{\mathcal{S}}_{00}}{d}}}-\arccos{\sqrt{\frac{\chi^{\mathcal{Q}}_{00}}{d}}} \right) \; \label{eqn:chi_upper} ,
    \end{equation}
    or more simply
    \begin{equation}
        \delta^{\mathcal{S}\circ \mathcal{Q}} \ge \vert \delta^{\mathcal{S}}-\delta^{\mathcal{Q}} \vert \; .
    \end{equation}
    The inequality is saturated iff $v_{ij}=1$, $\phi_{i}^{\mathcal{Q}}=\phi_{i^{\prime}}^{\mathcal{Q}}$, and $\phi_{j}^{\mathcal{S}}=\phi_{j^{\prime}}^{\mathcal{S}}$ for all $ i, i^{\prime}=1, \cdots, K(\mathcal{Q})$ and $j, j^{\prime}=1, \cdots, K(\mathcal{S})$. The quantities $v_{ij}$, $\phi_{i}^{\mathcal{Q}}$, and $\phi_{j}^{\mathcal{S}}$ are defined in Appendix~\ref{sec:Chi_upper} in terms of the Kraus operators of $\mathcal{Q}$, $\mathcal{S}$ and the $d^{2}$ matrix basis elements of $\mathcal{L}(\mathcal{H})$. 
\end{lemma}

For completeness, the proof of this lemma is found in Appendix~\ref{sec:Chi_upper}. 

Let us denote by $\mathcal{Q}\equiv \mathcal{I}^{\hat{\theta}_{n-1}\cdots \hat{\theta}_{1}}_{\theta_{n-1}\cdots \theta_{1}}$ and $\mathcal{S}\equiv \mathcal{I}^{\hat{\theta}_{n}}_{\theta_{n}}$, such that $\mathcal{S}\circ \mathcal{Q} = \mathcal{I}^{\hat{\theta}_{n}\hat{\theta}_{n-1}\cdots \hat{\theta}_{1}}_{\theta_{n}\theta_{n-1}\cdots \theta_{1}}$. We further use the notation $\chi_{00}^{1\rightarrow n}$, $\chi_{00}^{1 \rightarrow (n-1)}$, $\chi_{00}^{n}$, $\delta^{1 \rightarrow n}$, $\delta^{1 \rightarrow (n-1)}$, and $\delta^{n}$ to replace $\chi_{00}^{\mathcal{S}\circ \mathcal{Q}}$, $\chi_{00}^{\mathcal{Q}}$, $\chi_{00}^{\mathcal{S}}$, $\delta^{\mathcal{S}\circ \mathcal{Q}}$, $\delta^{\mathcal{Q}}$, and $\delta^{\mathcal{S}}$, respectively. We also use the definition in Eq.~\eqref{eqn:err_angle} to write similar notations for the entanglement fidelities $F_{e}^{1 \rightarrow n}$, $F_{e}^{1 \rightarrow (n-1)}$, and $F_{e}^{n}$, in terms of $\delta^{1 \rightarrow n}$, $\delta^{1 \rightarrow (n-1)}$, and $\delta^{n}$. Therefore, we can reframe Lemma~\ref{le:err_angle} by the authors of \cite{carignan2019bounding} as the following set of \textit{recurrence inequalities} in the context of spectator-based recovery:
\begin{lemma} \label{th:rec_ent_fid}
    Given the entanglement fidelities $F_{e}^{i}$ of the single-cycle recovery protocols at each time step $\Delta t_{\mathcal{R}} = t_{i+1}-t_{i}$, the $n$-cycle entanglement fidelity $F_{e}^{1 \rightarrow n}$ of the multi-cycle recovery protocol is bounded from above by the $(n-1)$-cycle entanglement fidelity $F_{e}^{1 \rightarrow (n-1)}$, as
    \begin{equation}
        F_{e}^{1 \rightarrow n} \leq \cos^{2} \left( \arccos{\sqrt{F_{e}^{1 \rightarrow (n-1)}}}-\arccos{\sqrt{F_{e}^{n}}} \right) \; ,
    \end{equation}
    or equivalently,
    \begin{equation}
        \delta^{1\rightarrow n} \ge \vert \delta^{1 \rightarrow (n-1)}-\delta^{n} \vert \; . \label{eqn:rec_ineq}
    \end{equation}
    The necessary and sufficient conditions for the saturation of this inequality are identical to that of Lemma~\ref{le:err_angle}.
\end{lemma}

\begin{remark} \label{re:coherence}
    As noted in \cite{carignan2019bounding}, the entanglement fidelity of a composite channel exhibits ``constructive'' and ``destructive interference'' with respect to the individual channel entanglement fidelities. In our case, we view the $n$-cycle recovery as a composite channel, where the individual channels are the $(n-1)$-cycle recovery and the $n$-th timestep recovery. Therefore, the same phenomenon of constructive and destructive interference applies here. This is purely a multi-cycle recovery phenomenon that is not present in single-cycle recovery case, which has been the main focus of modern literature in QEC.
\end{remark}

\subsection{Contribution of The Spectator System}
To identify the contribution of the lack of complete knowledge of $\theta$ to the recurrence inequalities, let us consider the error angle 
\begin{equation}
    \hat{\delta} \equiv \delta^{\mathcal{R}_{\hat{\theta}}\circ \mathcal{N}_{\theta}} \coloneqq \arccos{\sqrt{F_{e}(\mathcal{R}_{\hat{\theta}}\circ \mathcal{N}_{\theta})}} \; . 
\end{equation}

\begin{theorem} \label{th:spec_multi}
    Given the multi-cycle entanglement fidelity $F_{e}^{1\rightarrow (n-1)}$ from the previous $n-1$ cycles, the contribution of the spectator system to the upper bound of the total $n$-cycle entanglement fidelity $F_{e}^{1\rightarrow n}$ is given by
    \begin{equation}
        h(\theta_{n}, F_{e}^{1\rightarrow (n-1)})\operatorname{Var}(\hat{\theta}_{n}) \; ,
    \end{equation}
    where
    \begin{equation}
    h \coloneqq \frac{g(\theta_{n})\sin\left( 2\delta^{1\rightarrow (n-1)}-2\arccos{\sqrt{F^{n}_{e}}} \right)}{2\sqrt{F^{n}_{e}(1-F^{n}_{e})}} \; ,
    \end{equation}
    with
    \begin{equation}
        g(\theta_{n}) \coloneqq -\frac{1}{2} \left. \left( \frac{d^{2}}{d\nu^{2}} F_{e}(\mathcal{R}_{\theta_{n} +\nu}\circ \mathcal{N}_{\theta_{n}}) \right) \right \vert_{\nu = 0} \; .
    \end{equation}
\end{theorem}
The proof of this theorem is found in Appendix~\ref{apx:multi_cycle_proof}.

\subsection{Application to [4,1] Code of The Amplitude-Damping Channel} \label{sec:AD_application_2}
The contribution of the spectator system in multi-cycle bounds can also be computed explicitly for the [4,1] code of the AD channel. For a fixed value of the entanglement fidelity $F_{e}^{1\rightarrow (n-1)}$ (or equivalently, $\delta^{1\rightarrow (n-1)}$) at the $(n-1)$-th step, we can plot the total upper bound in the case of both complete and incomplete knowledge. The simplest case where the spectator system's parameters are $\gamma=1$ and $m=1$ is shown in Fig.~\ref{fig:multicycle}. 

Note that, although we expect the incomplete knowledge about the noise parameter to deteriorate the allowed values of the entanglement fidelity (as we have shown for single-cycle QEC of the AD channel in Fig.~\ref{fig:[4,1]fund}(b)), in the multi-cycle scenario, this can play to our advantage due to the coherence between the accumulated error during the prior $(n-1)$ cycles and the error due to the limited knowledge about the noise parameter at the $n$-th cycle (see Remark~\ref{re:coherence}). This observation further supports the potential superiority of spectator-based recovery techniques in maintaining real-time quantum memories.

\begin{figure}[t!] 
    \centering
    \includegraphics[width=9cm]{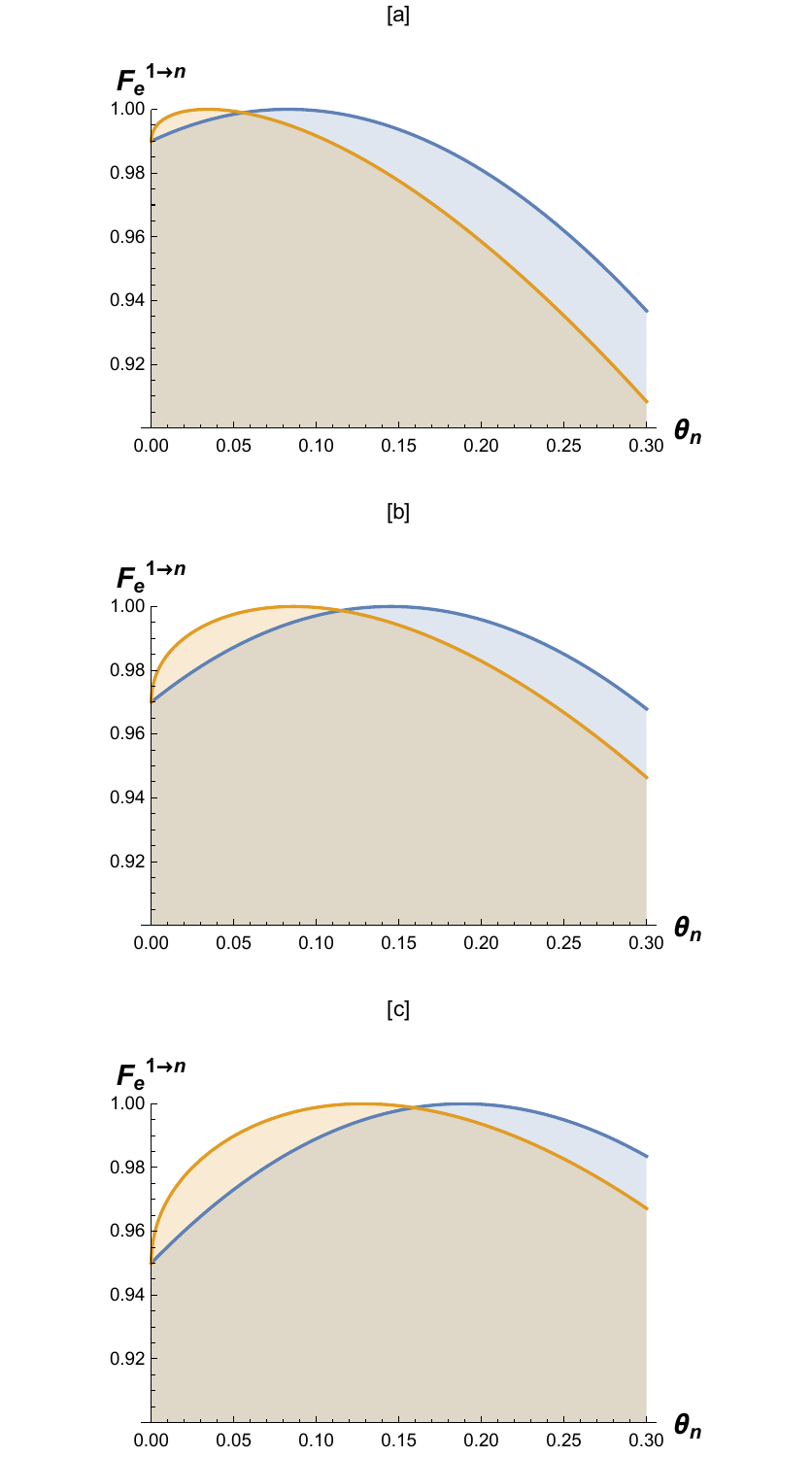}
    \caption{Subplots show the dependence of the accumulated $n$-cycle entanglement fidelity $F_{e}^{1\rightarrow n}$ on the value of the noise parameter $\theta_{n}$ at the $n$-th cycle for the [4,1] code of the amplitude-damping channel. The spectator system is taken to have the simplest characteristic parameters $(\gamma=1, m=1)$. The colored regions indicate allowed values for the entanglement fidelity. The blue color refers to the case of perfect knowledge of $\theta_{n}$ and the orange color to the lack of that knowledge. From top to bottom, the value of the accumulated $(n-1)$-cycle entanglement fidelity $F_{e}^{1\rightarrow (n-1)}$ is picked to be (a) 0.99, (b) 0.97, and (c) 0.95, respectively.}
    \label{fig:multicycle}
\end{figure}

\section{Comparison With Previous Literature} \label{sec:Discussion}

\subsection{Relation to Quantum Information-Theoretic Protocols}
In this article, we focused on the diamond distance due to its operational meaning in terms of a quantum channel discrimination task. In this context, Lemma~\ref{co:examples} could be interpreted as a fundamental bound on the success probability of such a task. A similar bound has already been derived by Pirandola \textit{et al.} in \cite{pirandola2019fundamental} using port-based teleportation \cite{ishizaka2008asymptotic}. In fact, the bound in \cite{pirandola2019fundamental} is valid for general adaptive protocols.

Furthermore, as current techniques of quantum control have matured, the applicability of Lemma~\ref{th:rec_ent_fid} is not only confined to multiple recovery rounds, as demonstrated experimentally in e.g. \cite{cramer2016repeated}. It can also be applied in various quantum information-theoretic tasks where multiple calls to the noisy channel and adaptive feedback are allowed, such as quantum channel discrimination with adaptive feedback \cite{chiribella2008memory, hayashi2009discrimination, pirandola2019fundamental}.
\subsection{Relation to Robustness of Channel-adapted QEC}
Our approach to recovery with incomplete knowledge is closely related to the robustness of channel-adapted QEC codes studied previously in literature \cite{kosut2008robust, ballo2009robustness, huang2019robustness, layden2020robustness}. To elaborate, since QEC codes are designed to correct the most likely errors, an important question to ask is: how resilient (robust) is the designed QEC code with respect to some arbitrary mixing with the next-most likely errors? The authors of \cite{ballo2009robustness} have framed the robustness problem such that it applies both for Pauli and non-Pauli channels, as follows: One first finds the optimum recovery $\mathcal{R}$ of a channel $\mathcal{N}$ (the most likely noise) by maximizing the entanglement fidelity of $\mathcal{R}\circ \mathcal{N}$, and then one mixes the original channel $\mathcal{N}$ with some other channel $\mathcal{N}^{\prime}$ (the next-most likely noise) by taking their convex combination, i.e. $\mathcal{N}_{\mu}\coloneqq (1-\mu)\mathcal{N}+\mu \mathcal{N}^{\prime}$ for some mixing parameter $\mu \in [0, 1]$. Then, the robustness of the recovery $\mathcal{R}$ with respect to $\mu$ is found by considering the entanglement fidelity of $\mathcal{R}\circ \mathcal{N}_{\mu}$ and observing if it has major variations as a function of the mixing parameter $\mu$. This setup shares some similarities with our approach, however, it has a different quantity of interest, namely the \textit{sensitivity} of entanglement fidelity with respect to changes in the mixing parameter, quantified as the first derivative with respect to $\mu$ of
\begin{equation}
    F_{e}(\mathcal{R}_{\mu}\circ \mathcal{N}_{\mu})-F_{e}(\mathcal{R}\circ \mathcal{N}_{\mu}) \; ,
\end{equation}
where $\mathcal{R}_{\mu}$ is the optimum recovery of the mixing $\mathcal{N}_{\mu}$ (Also see Appendix~\ref{subsec:robustness_bound} for bounds on a similar quantity). This is to be contrasted with the quantity of interest in this article (using the parameter notation $\mu$)
\begin{equation}
    F_{e}(\mathcal{R} \circ \mathcal{N})-F_{e}(\mathcal{R}_{\mu}\circ \mathcal{N}) \; .
\end{equation}
Here, $\mu$ plays the role of the uncertainty $\nu \equiv \hat{\theta}-\theta$ in the noise parameter $\theta$, and therefore it has a different interpretation. Namely, there is no next-most likely noise in this description! Instead, $\mu$ is the random variable describing the uncertainty in the environment noise parameter and has a finite variance, by the QCRB.

The possibility of including the channel uncertainty as a probability distribution $p(\mu)$ in the optimization problem of entanglement fidelity has been discussed by Fletcher in \cite{fletcher2007channel}. The question then, as mentioned in \cite{ballo2009robustness}, is: how to pick a physical probability distribution $p(\mu)$? In our picture (spectator-based QEC), this question has a relatively simple answer, as one should always pick the probability distribution that maximizes the Shannon entropy with a fixed expectation and variance (larger or equal to the inverse of the quantum Fisher information). Such a probability distribution is called `` the truncated normal distribution''.

\subsection{Relation to [4,1] AD Code Literature}

As amplitude-damping (qubit decoherence) is one of the most common noises in quantum systems, developing QEC codes for this particular noise has been a major focus of QEC literature since its inception in 1995. The simplest of such QEC codes is the approximate $[4,1]$ code \cite{leung1997approximate}. Since then, QEC methods for the AD noise have been developing in sophistication by using various new techniques, such as channel adaptation \cite{fletcher2008channel}, stabilizer formalism \cite{cafaro2014approximate}, and semi-definite programming \cite{fletcher2007optimum}. These techniques have been steadily improving upon the entanglement fidelity of the original $[4,1]$ code in \cite{leung1997approximate}. However, If we want to implement these techniques for real-time quantum memories (where the decoherence parameter is slowly varying in time), how much of the improvements upon \cite{leung1997approximate} obtained in the last two decades are we likely to retain? The answer to this question, we compare the performance of the $[4,1]$ code in the incomplete knowledge scenario with previous literature. As spectator systems are characterized by their physical nature $\gamma \ge 1$ and the number of independent subsystems $m \in \mathbb{N}_{+}$ (see Eq.~\eqref{eqn:M_spec}), the answer will vary from one physical implementation to another. However, we consider the above question in the case $(\gamma=1, m=1)$. The results of this comparison are summarized in Fig.~\ref{fig:[4,1]ent_fid} and the table below
\vspace{0.5cm}
\begin{figure}
    \centering
    \includegraphics[width=\linewidth]{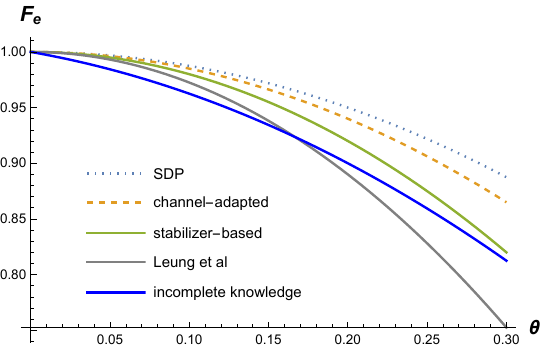}
    \caption{Comparison of the maximum performance of spectator-based recovery (subject to varying noise parameter $\theta$) with various recovery optimization approaches to the [4, 1] code of the amplitude-damping channel (with fixed noise parameter $\theta$). We consider the ``worst case'' spectator parameters $(\gamma=1, m=1)$. Shown are the performances of the well-known approximate QEC code in Leung \textit{et al.} \cite{leung1997approximate}, its channel-adapted version by Fletcher \textit{et al.} \cite{fletcher2008channel}, its SDP optimized version by Fletcher \textit{et al.} \cite{fletcher2007optimum}, its stabilizer-based version \cite{cafaro2014approximate, fletcher2008channel}, and the incomplete knowledge extension of the channel-adapted QEC in \cite{fletcher2008channel}. Here, the difference between the ``channel-adapted'' and ``incomplete knowledge'' entanglement fidelities showcases the fundamental metrological cost of operating a real-time quantum memory. All other $(\gamma \ge 1, m \ge 1)$ spectator-based recoveries lie above the ``incomplete knowledge'' graph.}
    \label{fig:[4,1]ent_fid}
\end{figure}

\begin{table}[ht]
\centering
\begin{tabular}{|c|c|} \hline
     Previous literature & $F_{e}$ to $\mathcal{O}(\theta^{3})$ order \\
    \hline
     Leung \textit{et al.}
     \cite{leung1997approximate} &  $1-2.75\theta^{2}$ \\
     Stabilizer-Based \cite{fletcher2007channel, cafaro2014approximate} & $1-2\theta^{2}$ \\
     Channel-Adapted \cite{fletcher2008channel, zhan2013entanglement} & $1-1.5\theta^{2}$  \\
     SDP \cite{fletcher2007optimum} & $1-1.25\theta^{2}$  \\
     \hline
     Incomplete Knowledge & $1-0.25\theta-1.25\theta^{2}$ \\
     \hline
\end{tabular}
\caption{Comparison between the entanglement fidelities of the [4,1] code for small noise parameter value $\theta$, for different recovery protocols (here SDP stands for ``semi-definite programming''). Note that in the incomplete knowledge scenario, the leading error term in the entanglement fidelity of recovery is linear in $\theta$, as opposed to quadratic, which is the optimal result when the noise parameter $\theta$ is known \textit{apriori}.}
\end{table}

We observe that, due to incomplete knowledge of the noise parameter, the $[4,1]$ code of the AD channel performs suboptimally to \cite{leung1997approximate} for noise parameter values below a certain threshold $\theta \leq 0.17$. However, beyond that point, the improvements introduced by channel-adapted recoveries and semi-definite programming techniques are preserved, as they still outperform \cite{leung1997approximate}, even in the presence of incomplete knowledge about $\theta$. Furthermore, the range of the values of $\theta \in [0, 1]$ for which this outperformance is preserved gets larger the larger we pick $\gamma$ and/or $m$ (see Fig.~\ref{fig:[4,1]fund}(a)).

Let us consider one final observation. We noted in Fig.~\ref{fig:[4,1]fund} that different values of the spectator parameter $\gamma$ yield different regions of $\theta$ where channel-adapted and semi-definite programming techniques in QEC maintain their improvements upon the approximate $[4,1]$ code \cite{leung1997approximate}, provided that a spectator system is implemented in the incomplete-knowledge recovery protocol. One might observe that, since the value of $\gamma$ in Eq.~\eqref{eqn:f_spec} generally depends on the couplings of the spectator and memory systems with the environment, the only way to change $\gamma$ is to change the physical implementation of at least one of these systems. However, recent quantum control techniques, such as Hamiltonian amplification \cite{arenz2020amplification}, allow for the tuning of the coupling strengths between an environment and any continuous variable quantum system, with a quadratic coupling Hamiltonian. Therefore, provided that the implementation of either the spectator or memory system has continuous degrees of freedom \cite{gottesman2001encoding}, the Hamiltonian amplification technique yields a practical advantage for spectator-based recovery architectures, as the resulting entanglement fidelities can be manipulated in experiments for any desired region of the noise parameter $\theta$, as seen in Fig.~\ref{fig:[4,1]fund}(a).

\subsection{Relation to Time-Dependent QEC}
In \cite{fujiwara2014instantaneous}, the author suggests that knowledge of the error rates for Pauli channels is not the most useful side information in QEC. Indeed, as mentioned previously (see Remark~\ref{rem:non_pauli}), the assumption that the \textit{optimal} recovery channel (defined by Eq.~\eqref{eqn:opt_recovery}) depends on the environment noise parameter does not hold for Pauli channels. Nevertheless, it is important to note that optimization-based techniques of QEC for Pauli channels do generally benefit from the knowledge of the noise parameter. This is especially relevant when error identification from syndrome measurements is not unique (e.g. in surface codes \cite{fowler2012surface}). Hence, one can only construct (suboptimal) decoders, rather than the optimal recovery map in Eq.~\eqref{eqn:opt_recovery}. This generally yields decoders that depend on the noise parameters, even for Pauli channels. For example, various types of decoders exist for both repetition codes \cite{spitz2018adaptive, kelly2016scalable} and surface codes \cite{huo2017learning}, where under the presence of a drifting noise parameter, one can design an adaptive decoder that can track this drift while not interrupting the QEC protocol. Therefore, the results of this article could be expanded to include adaptive decoders for repetition and surface codes, rather than the optimal recovery map defined in Eq.~\eqref{eqn:opt_recovery}. Finally, it is worth mentioning that other approaches to adaptation in QEC literature have been pursued, e.g. in Refs.~\cite{taghavi2010channel, florjanczyk2016situ}.

\section{Conclusion and Open Questions}
In this article, I consider the problem of building a real-time (drift-adapting) quantum memory and present it as a spectator-based recovery protocol. To counter noise drift, the spectator system performs a real-time parameter estimation (generally in the presence of nuisance parameters) and feeds forward this classical side information to the ``best-guess'' recovery map. To quantify the single-cycle information-theoretic cost of adaptation in real-time quantum memories, I compute a lower bound for the diamond distance between the optimal (inaccessible) and best-guess (accessible) recovery protocols. This approach is generalized in Appendix~\ref{apx:gen_dist_meas} for other relevant distinguishability measures between arbitrary two quantum channels. For slowly drifting noise parameters, I show that a metrological bound exists, determined by the quantum Fisher information of the spectator dynamics. This bound is demonstrated for the $[4,1]$ code of the amplitude-damping channel, and the effects of various physical choices of spectator qubits and nuisance parameters are discussed. Finally, for multi-cycle recovery, I recall a theorem in \cite{carignan2019bounding} and use it to derive an upper bound to the fidelity of multi-cycle recovery in terms of recurrence inequalities. The contribution of the lack of knowledge of the noise parameters (i.e. noise-drift adaptation) is also derived. This is also showcased for the $[4,1]$ code of the amplitude-damping channel, and regions of outperformance in the spectator-based recovery protocols are highlighted. The advantages of spectator-based recovery compared to non-adaptive recovery protocols, even in the perfect knowledge scenario, is due to the coherence of errors from different cycle numbers as well as the imperfect knowledge (noise estimation) error. This phenomenon is exclusive to multi-cycle QEC.

The results mentioned above are relevant for various research communities, such as quantum error correction, quantum communication, quantum information, quantum control, and quantum computing. To elaborate, the existence of lower bounds on any channel recovery (Eqs.~\eqref{eqn:QEC_lower} and \eqref{eqn:main_spec_lower}, or more generally in Theorem~\ref{th:main_th} in Appendix~\ref{apx:gen_dist_meas}) could be valuable in testing the performance of various optimization-based techniques in QEC to determine if optimal performance is reached. As discussed in Section~\ref{sec:Discussion}, these bounds may also have a broader interest in various domains of quantum information as they hold for any generalized distinguishability measure and between any two quantum channels (see Appendix~\ref{apx:gen_dist_meas}). Multi-cycle bounds (Lemma~\ref{th:rec_ent_fid}) might also be useful in adaptive quantum information-theoretic protocols, where many calls to the noisy channel are made. The analysis made for the $[4,1]$ code of the amplitude-damping channel sheds light on what to expect when implementing such QEC codes in real-time quantum memories \cite{proctor2020detecting, majumder2020real}, while also providing an excited avenue in terms of outperformance in the incomplete knowledge scenario for multi-cycle recovery, which is quickly starting to become a reality \cite{cramer2016repeated}. Finally, implementing novel quantum control techniques, such as Hamiltonian amplification for continuous quantum systems \cite{arenz2020amplification}, might prove useful in controlling the coupling strength of the spectator system. Therefore, one can experimentally optimize over the selection of all possible spectator system parameters without physically changing the spectator system. 

Many questions are left open:
\begin{itemize}
    \item Extension of the information-theoretic and metrological lower bounds to Pauli channels with suboptimal decoders. This is relevant for current surface codes, as the optimal recovery map Eq.~\eqref{eqn:opt_recovery} is inaccessible, due to the probabilistic nature of error identification from syndrome measurements for high error rates \cite{fowler2012surface}.
    \item Incorporation of important theoretical techniques of the maximum overlap problem in quantum information theory, e.g. the two-sided bounds by Tyson \cite{tyson2009two, tyson2010two} using directional iterates. This approach has led to various important results previously \cite{barnum2002reversing, beny2010general}. An interesting proposal (potentially also in quantum metrology) would be to apply the directional iterate technique to the semi-inner product used to define the QFIM, given in Eq.~\eqref{eqn:semi_inner_prod} of Appendix~\ref{apx:QFIM}.
    \item A deeper analysis of the multi-cycle regime is required. This includes, but is not limited to, the study of optimal conditions for achieving the coherent error cancellation, as well as incorporating techniques from asymptotic quantum information theory to gain further insight into the multi-cycle case.
    \item Continuous (dynamical) recovery using Petz recovery maps \cite{kwon2022reversing} could also be considered, as well as other adaptive approaches implementing Petz recovery maps \cite{jayashankar2022quantum, biswas2023noise}. This can potentially extend the temporal range of applicability of spectator-based recovery protocols to faster varying noise. Another interesting dynamical model of real-time quantum memory could be constructed from the open system theory of two subsystems (memory and spectator) with slow and fast dynamics, relative to the environment noise. This has been studied previously in the context of adiabatic elimination in bipartite open quantum systems \cite{azouit2017towards}.
    \item Considerations of spatial variability of the noise parameter are also needed for scalability of the spectator-based recovery protocols \cite{klesse2005quantum, gupta2020integration}. Generalization beyond the independent noise model in Eq.~\eqref{eqn:M_spec} is also of relevance.
    \item Finally, one would be remiss by not considering the large literature on approximate recoverability of quantum states, see e.g. \cite{li2018squashed, wilde2015recoverability, junge2018universal, fawzi2015quantum, buscemi2016approximate} and references within. It is interesting to see whether incomplete knowledge recovery protocols would benefit from a similar approach.
\end{itemize}




\section{Acknowledgements}
I am grateful to Mark M. Wilde, Hwang Lee, Stav Haldar, Kenneth Brown, Milad Marvian, Todd Brun, and the anonymous referees for helpful comments and stimulating discussions. I am also grateful for the early discussions with Lorenza Viola regarding stroboscopically changing noise parameters, which motivated certain parts of this work. I would like to acknowledge the support of Kenneth R. Brown, Mark M. Wilde, Hwang Lee, and the Coates Research Award by the Department of Physics and Astronomy at Louisiana State University. This work was supported by the U.S. Army Research Office through the U.S. MURI Grant No. W911NF-18-1-0218.   

\bibliography{references.bib}
\bibliographystyle{quantum}

\appendix

\section{Lower-Bounding Generalized Distinguishability Measures Using Entanglement Fidelity} \label{apx:gen_dist_meas}

\subsection{Generalized Distinguishability and Distance Measures}
To quantify the success of a recovery protocol (such as QEC), we need to introduce the concepts of generalized distinguishability and distance measures between two states as well as between two channels \cite{kueng2016comparing, aliferis2007level, gilchrist2005distance, smirne2022holevo, leditzky2018approaches, khatri2020principles}. 

We say that $\mathbf{D}:\mathcal{D}(\mathcal{H}) \times \mathcal{L}_{+}(\mathcal{H}) \rightarrow \mathbb{R}^{1}$ is a generalized distinguishability measure between two states if it satisfies the \textit{data-processing inequality} (DPI), i.e. for arbitrary $\mathcal{Q}$ CPTP map and all $ \rho, \sigma \in \mathcal{D}(\mathcal{H})$, we have \footnote{Some previous papers have used the notation $\mathbf{D}(\rho \Vert \sigma)$, rather than $\mathbf{D}(\rho, \sigma)$, to indicate the generalized distinguishability measure between $\rho$ and $\sigma$. Here, we use the latter notation to emphasize the role that $\mathbf{D}$ plays also as a distance measure in deriving standard upper bounds in the context of QEC. Please Appendix~\ref{sec:upper_bound} for more details.}
    \begin{equation}
        \mathbf{D}(\mathcal{Q}(\rho), \mathcal{Q}(\sigma)) \leq \mathbf{D}(\rho, \sigma) \; .
    \end{equation}
An important consequence of DPI is the property of isometric invariance. Namely, for any isometry $V$, the following holds \cite{khatri2020principles}
\begin{equation}
    \mathbf{D}(\mathcal{V}(\rho), \mathcal{V}(\sigma)) = \mathbf{D}(\rho, \sigma) \; , \label{eqn:iso_inv}
\end{equation}
where $\mathcal{V}(\cdot)=V^{\dagger}(\cdot)V$.

Independently, we say that $\mathbf{D}:\mathcal{D}(\mathcal{H}) \times \mathcal{D}(\mathcal{H}) \rightarrow \mathbb{R}^{1}_{+}$ is a generalized distance measure between two states if it satisfies the following three properties for all $\rho, \sigma, \tau \in \mathcal{D}(\mathcal{H})$:
\begin{enumerate}
    \item \textit{Positivity and faithfulness:} 
    \begin{equation}
        \mathbf{D}(\rho, \sigma) \ge 0 \; ,
    \end{equation}
    where the equality holds iff $\rho = \sigma$.
    \item \textit{Symmetry:} $\mathbf{D}(\rho, \sigma)=\mathbf{D}(\sigma, \rho)$.
    \item \textit{Triangle inequality:}
    \begin{equation}
        \mathbf{D}(\rho, \sigma) \leq \mathbf{D}(\rho, \tau) + \mathbf{D}(\tau, \sigma) \; .
    \end{equation}
\end{enumerate}

A common requirement for fault-tolerant QEC and quantum computing is the so-called ``chaining property'' \cite{kueng2016comparing, gilchrist2005distance}. However, this property of generalized distinguishability/distance measures is derivative from more elementary properties, such as DPI and the triangle inequality (see Appendix~\ref{sec:chaining} for a short discussion).

Finally, we say that the map $\mathbf{D}:\mathcal{D}(\mathcal{H}) \times \mathcal{L}_{+}(\mathcal{H}) \rightarrow \mathbb{R}^{1}$ satisfies the \textit{joint convexity} property if for any two ensembles $\{p_{X}(x), \rho^{x}\}_{x \in \mathcal{X}}$ and $\{ p_{X}(x), \sigma^{x} \}_{x \in \mathcal{X}}$, where $p_{X}$ is a probability distribution function of the random variable $X$ over the set $\mathcal{X}$, we have
    \begin{equation}
        \mathbf{D}\left(\sum_{x \in \mathcal{X}}p_{X}(x)\rho^{x}, \sum_{x\in \mathcal{X}}p_{X}(x)\sigma^{x}\right) \leq \sum_{x\in \mathcal{X}}p_{X}(x)\mathbf{D}(\rho^{x}, \sigma^{x}) \; .
    \end{equation}
\begin{table*}[t!]
    \centering
    \begin{tabular}{ |p{5.3cm}||p{3cm}|p{3cm}|p{3cm}|  }
         \hline
         \multicolumn{4}{|c|}{List of different measures} \\
         \hline
         Name &  Data Processing & Distance Measure & joint convexity \\
         \hline
         Quantum Relative Entropy   &  Yes    & No &  Yes \\
         Generalized $\alpha$-Relative Entropies (Petz-Renyi, Sandwiched, etc.) &  Yes  &  No   & Yes \\
         Trace Distance &  Yes & Yes & Yes\\
         Bures Distance  & Yes & Yes & Yes\\
         Sine Distance & Yes & Yes & Yes\\
         Amortized Divergence \cite{wilde2020amortized} & Yes & No & Yes \\
         \hline
    \end{tabular}
    \caption{Summary of properties of various measures used in quantum information theory.}
    \label{tab:measures}
\end{table*}
For fidelity-based distinguishability measures, such as the Bures and Sine distances, this directly follows from the double concavity of the fidelity function (see e.g. \cite{khatri2020principles}). 

Alternatively, it is well-known that the joint convexity property can be derived from the DPI (with respect to the partial trace channel) if we further assume that $\mathbf{D}$ satisfies the direct sum property for classical-quantum states \cite{khatri2020principles}, i.e.
    \begin{align}
        \mathbf{D} &\left(\sum_{x \in \mathcal{X}} p_{X}(x)|x\rangle \! \langle x| \otimes \rho^{x} , \sum_{x\in \mathcal{X}}p_{X}(x)|x\rangle \! \langle x| \otimes \sigma^{x}\right) \nonumber \\ &= \sum_{x\in \mathcal{X}}p_{X}(x)\mathbf{D}(\rho^{x}, \sigma^{x}) \; .
    \end{align}


For a summary of various distinguishability and/or distance measures, as well as which properties they satisfy, please see Table~\ref{tab:measures}. All the above properties are satisfied by \cite{khatri2020principles, gilchrist2005distance}
\begin{enumerate}
    \item Trace Distance: $\mathbf{D}_{\text{Tr}}(\rho, \sigma)=\frac{1}{2}\Vert \rho -\sigma \Vert_{1}$.
    \item Bures Distance: $\mathbf{D}_{\text{B}}(\rho, \sigma)=\sqrt{2-2\sqrt{F(\rho, \sigma)}}$.
    \item Sine Distance: $\mathbf{D}_{\text{S}}(\rho, \sigma)=\sqrt{1-F(\rho, \sigma)}$,
\end{enumerate}
where $F(\rho, \sigma)=\Vert \sqrt{\rho}\sqrt{\sigma} \Vert^{2}_{1}$ is the fidelity function. 


Using the generalized distinguishability (distance) measures between two states,
 we define the generalized distinguishability (distance) measures between two channels $\mathcal{Q}^{A \rightarrow B}$ and $\mathcal{S}^{A \rightarrow B}$, as follows
\begin{equation}
    \mathbf{D}(\mathcal{Q}, \mathcal{S}) \coloneqq \sup_{\rho}\mathbf{D}(\textsf{id}^{R}\otimes \mathcal{Q}^{A\rightarrow B}(\rho), \textsf{id}^{R}\otimes \mathcal{S}^{A\rightarrow B}(\rho)) \; , \label{eqn:channel_dist}
\end{equation}
where $\rho \in \mathcal{D}(\mathcal{H}^{A}\otimes \mathcal{H}^{R})$, for arbitrary Hilbert space dimensions of the reference system $R$. By using joint convexity and the Schmidt decomposition of pure states, it can be shown that the maximization need only be taken over pure states $\psi_{RA}$, with the reference system $R$ having the same Hilbert space dimensions as $A$ \cite{khatri2020principles}, i.e. 
\begin{equation}
    \mathbf{D}(\mathcal{Q}, \mathcal{S}) \coloneqq \sup_{\psi}\mathbf{D}(\textsf{id}^{R}\otimes \mathcal{Q}^{A\rightarrow B}(\psi), \textsf{id}^{R}\otimes \mathcal{S}^{A\rightarrow B}(\psi)) \; .
\end{equation}
Finally, it is important to note that the joint convexity property of generalized distinguishability measures for states implies the same property for channels. This is seen by considering the two channels $\mathcal{Q}^{A\rightarrow B}=\sum_{x\in \mathcal{X}}p_{X}(x)\mathcal{Q}^{A\rightarrow B}_{x}$ and $\mathcal{S}^{A \rightarrow B}=\sum_{x\in \mathcal{X}}p_{X}(x)\mathcal{S}^{A \rightarrow B}_{x}$, and then applying the joint convexity property for states, as follows
\begin{align}
    & \mathbf{D}(\mathcal{Q}, \mathcal{S}) 
    = \sup_{\rho}\mathbf{D}(\textsf{id}^{R}\otimes \mathcal{Q}^{A\rightarrow B}(\rho), \textsf{id}^{R}\otimes \mathcal{S}^{A\rightarrow B}(\rho)) \\ &= \mathbf{D}\left(\textsf{id}^{R}\otimes \mathcal{Q}^{A\rightarrow B}(\rho^{\star}), \textsf{id}^{R}\otimes \mathcal{S}^{A\rightarrow B}(\rho^{\star})\right)  \\ 
    & \leq  \sum_{x\in \mathcal{X}}p_{X}(x)\mathbf{D}\left( \textsf{id}^{R}\otimes \mathcal{Q}^{A \rightarrow B}_{x}(\rho^{\star}), \textsf{id}^{R}\otimes \mathcal{S}^{A \rightarrow B}_{x}(\rho^{\star})\right) \\
    & \leq \sum_{x\in \mathcal{X}}p_{X}(x) \sup_{\rho}\mathbf{D}\left( \textsf{id}^{R}\otimes \mathcal{Q}^{A \rightarrow B}_{x}(\rho), \textsf{id}^{R}\otimes \mathcal{S}^{A \rightarrow B}_{x}(\rho)\right) \\
    & \equiv \sum_{x\in \mathcal{X}}p_{X}(x)\mathbf{D}(\mathcal{Q}_{x}, \mathcal{S}_{x}) \; .
\end{align}
Consequently, we have the joint convexity property
\begin{align}
    & \mathbf{D}\left(\sum_{x\in \mathcal{X}}p_{X}(x)\mathcal{Q}_{x}, \sum_{x\in \mathcal{X}}p_{X}(x)\mathcal{S}_{x}\right) \nonumber \\ & \hspace{1.5cm} \leq   \sum_{x\in \mathcal{X}}p_{X}(x)\mathbf{D}(\mathcal{Q}_{x}, \mathcal{S}_{x}) \; . \label{eqn:double_con}
\end{align}

\subsection{Unitary \textit{t}-designs} \label{sec:t-design}
We call a function $P:\mathbb{U}(d)\rightarrow \mathbb{C}$ acting on any unitary $U$ in $\mathbb{U}(d)$ to be polynomial of degree $t$ if its dependence on the $2d^{2}$ real entries of $U$ is a polynomial of degree at most $t$ in each of its entries. Given a finite set of unitaries $\{U(x)\}_{x \in \mathcal{X}}$ in $\mathbb{U}(d)$, we say that they form a unitary $t$-design \cite{dankert2009exact, gross2007evenly} if the uniform Haar average over $\mathbb{U}(d)$ of any polynomial $P$ of degree $t$ is computed using the uniform average over the finite set $\{U(x)\}_{x \in \mathcal{X}}$ only, as follows
\begin{equation}
    \int_{\mathbb{U}(d)}dUP(U)=\frac{1}{\vert \mathcal{X} \vert}\sum_{x \in \mathcal{X}}P(U(x)) \; .
\end{equation}
It has been shown that for unitary 1 and 2 designs, the above averaging condition can be rewritten in a different form. We say that $\{U(x)\}_{x \in \mathcal{X}}$ forms a unitary 1-design in $\mathbb{U}(d)$ if 
\begin{equation}
    \frac{1}{\vert \mathcal{X} \vert}\sum_{x \in \mathcal{X}}U(x)\rho U^{\dagger}(x)=\pi \; , \label{eqn:unitary_1}
\end{equation}
for all $\rho \in \mathcal{D}(\mathcal{H})$, where $\pi=I/d$ is the maximally mixed state. An example of unitary 1-designs is given by the Pauli group. Further, we say $\{U(x)\}_{x \in \mathcal{X}}$ forms a unitary 2-design in $\mathbb{U}(d)$ if we have the following conditions for twirling of states or channels \cite{dankert2009exact, dankert2005efficient}
\begin{align}
    \int_{\mathbb{U}(d)}dU&(U\otimes U)\rho (U\otimes U)^{\dagger} \nonumber \\ &=\frac{1}{\vert \mathcal{X} \vert}\sum_{x \in \mathcal{X}}(U(x)\otimes U(x))\rho (U(x)\otimes U(x))^{\dagger} \; ,
\end{align}
for all $\rho \in \mathcal{D}(\mathcal{H}\otimes \mathcal{H})$, or equivalently
\begin{align}
    \int_{\mathbb{U}(d)}dU& U^{\dagger}\mathcal{Q}(U\rho U^{\dagger})U \nonumber \\&=\frac{1}{\vert \mathcal{X} \vert}\sum_{x \in \mathcal{X}}U^{\dagger}(x)\mathcal{Q}(U(x)\rho U^{\dagger}(x))U(x) \; , \label{eqn:unitary_2_des}
\end{align}
for all $\rho \in \mathcal{D}(\mathcal{H})$ and quantum channels $\mathcal{Q}$. An example of unitary 2-designs is given by the Clifford group \cite{dankert2005efficient, emerson2007symmetrized}.

\begin{remark}
    Note that, if $\{U(x)\}_{x \in \mathcal{X}}$ is a unitary $t$-design, then it also holds that $\{U(x)\}_{x \in \mathcal{X}}$ is a unitary $(t-1)$-design. For example, the Clifford group forms a unitary 3-design, and hence also a unitary 2-design.
\end{remark}

\subsection{Channel Twirlings}
Generally, channel twirlings can be defined with respect to both discrete and continuous sets of unitaries. In its most simple form, for a set of unitaries $\{U_{A}(x), V_{B}(x)\}_{x \in \mathcal{X}}$ and a probability distribution function $p_{X}$ defined over a finite set $\mathcal{X}$, the twirling of a quantum channel $\mathcal{Q}^{A\rightarrow B}$ (which we denote by a tilde symbol $\tilde{\mathcal{Q}}^{A\rightarrow B}$) is defined as
\begin{equation}
    \tilde{\mathcal{Q}}^{A \rightarrow B} \coloneqq \sum_{x\in \mathcal{X}} p_{X}(x)\mathcal{V}_{x}^{B\dagger }\circ \mathcal{Q}^{A\rightarrow B}\circ \mathcal{U}_{x}^{A} \; ,
\end{equation}
where we have used the notation for the unitary channels $\mathcal{U}_{x}^{A}(\cdot)\coloneqq U^{\dagger}_{A}(x)(\cdot)U_{A}(x)$ and $\mathcal{V}_{x}^{B}(\cdot)\coloneqq V^{\dagger}_{B}(x)(\cdot)V_{B}(x)$, for all $ x \in \mathcal{X}$. Although most of the results presented in this article are valid for any finite set $\mathcal{X}$, the case where it forms a group and $\{U_{A}(x), V_{B}(x)\}_{x \in \mathcal{X}}$ two unitary representations of it are of great interest \cite{leditzky2018approaches} (see Remark~\ref{le:leditzky}). 

Twirlings with continuous sets of unitaries have also been studied extensively in the literature. If we have some probability distribution (measure) $\mu(U)$ over the set of $d\times d$ unitary matrices $\mathbb{U}(d)$, then the continuous twirling of the channel $\mathcal{Q}^{A}$ is defined to be
\begin{equation}
    \tilde{\mathcal{Q}} \coloneqq \int_{ \mathbb{U}(d) } d\mu(U) \mathcal{U}^{\dagger} \circ \mathcal{Q} \circ \mathcal{U} \; .
\end{equation}

Twirling of quantum channels plays an important role in QEC and fault-tolerant quantum computing \cite{eggeling2001separability, silva2008scalable, magesan2008gaining, martinez2020approximating, emerson2007symmetrized, dankert2009exact, meier2013randomized}. Examples include: $(1)$ similarities between QEC codes for channels and their twirled versions \cite{silva2008scalable}, $(2)$ the simulability of twirled quantum channels on a quantum computer, due to the Gottesman-Knill theorem \cite{ gottesman1998heisenberg}, $(3)$ the fact that channels and their twirled versions share the same average and entanglement fidelities \cite{horodecki1999general}, $(4)$ various twirlings (Pauli, Clifford, and uniform Haar) rendering channels depolarizing \cite{horodecki1999general, nielsen2002simple, emerson2007symmetrized, dankert2009exact}, $(5)$ and finally, their close connection to unitary $t$-designs. Due to its importance, I recall some relevant properties of unitary $t$-designs in Appendix~\ref{sec:t-design} (also see \cite{dankert2005efficient, magesan2008gaining} for a brief review).

\subsection{Lower-Bounding Generalized Distinguishability Measures Using Entanglement Fidelity} \label{sec:lower_bound}

We start this section by showing a simple property that all generalized distinguishability measures satisfy with respect to channel twirling if the joint convexity property (or equivalently, if the direct sum property) is further assumed.
\begin{lemma} \label{le:gen_fund_lower}
    Assume that we are given two CPTP maps $\mathcal{Q}^{A \rightarrow B}$ and $\mathcal{S}^{A \rightarrow B}$, a set of unitaries $\{U_{A}(x), V_{B}(x)\}_{x \in \mathcal{X}}$, and a probability distribution function $p_{X}$ defined over the finite set $\mathcal{X}$. If the generalized distinguishability measure $\mathbf{D}$ satisfies the joint convexity property, then $\mathbf{D}(\mathcal{Q}, \mathcal{S})$ is lower bounded by the generalized distinguishability measure between the corresponding twirled channels $\tilde{\mathcal{Q}}^{A \rightarrow B}$ and $\tilde{\mathcal{S}}^{A \rightarrow B}$ with respect to the given weighted set of unitaries above, as follows
    \begin{equation}
    \mathbf{D}(\mathcal{Q}, \mathcal{S}) \ge \mathbf{D}(\tilde{\mathcal{Q}}, \tilde{\mathcal{S}}) \; , \label{eqn:gen_lower}
    \end{equation}
    where the lower bound is saturated iff the joint convexity of the generalized distinguishability measure between the two quantum channels is saturated with respect to the above set of weighted unitaries.
\end{lemma}

\renewcommand\qedsymbol{$\blacksquare$}

\begin{proof}
     Consider the isometric invariance property of $\mathbf{D}(\mathcal{Q}, \mathcal{S})$, namely for any $\mathcal{U}\coloneqq U(\cdot)U^{\dagger}$, where $U \in \mathbf{U}(d)$, we have
\begin{align}
    \mathbf{D}(\mathcal{Q}, \mathcal{S}) &=\mathbf{D}(\mathcal{U}\circ \mathcal{Q}, \mathcal{U}\circ \mathcal{S}) \\
    &=\mathbf{D}(\mathcal{Q}\circ \mathcal{U}, \mathcal{S}\circ \mathcal{U}) \; ,
\end{align}
where the first equality follows from Eq.~\eqref{eqn:iso_inv} and the second equality follows from the definition in Eq.~\eqref{eqn:channel_dist}. This implies that for all $ U_{A} \in \mathbf{U}(d_{A})$ and for all $ V_{B} \in \mathbf{U}(d_{B})$
\begin{equation}
     \mathbf{D}(\mathcal{Q}, \mathcal{S}) = \mathbf{D}(\mathcal{V}^{\dagger} \circ \mathcal{Q} \circ \mathcal{U}, \mathcal{V}^{\dagger}\circ \mathcal{S} \circ \mathcal{U}) \; .
\end{equation}
Consequently, by considering the generalized distinguishability measure $\mathbf{D}(\tilde{\mathcal{Q}}, \tilde{\mathcal{S}}) $ between the twirled channels, we arrive at
\begin{align}
    & \mathbf{D}\left(\sum_{x \in \mathcal{X}}p_{X}(x)\mathcal{V}^{x\dagger}\circ \mathcal{Q} \circ \mathcal{U}^{x}, \sum_{x \in \mathcal{X}}p_{X}(x)\mathcal{V}^{x\dagger}\circ \mathcal{S} \circ \mathcal{U}^{x}\right) \\
    & \leq \sum_{x \in \mathcal{X}}p_{X}(x) \mathbf{D}\left(\mathcal{V}^{x\dagger}\circ \mathcal{Q} \circ \mathcal{U}^{x}, \mathcal{V}^{x\dagger}\circ \mathcal{S} \circ \mathcal{U}^{x}\right) \\
    & = \sum_{x \in \mathcal{X}}p_{X}(x) \mathbf{D}\left(\mathcal{Q}, \mathcal{S} \right) = \mathbf{D}\left(\mathcal{Q}, \mathcal{S} \right) \; ,
\end{align}
where the inequality follows from Eq.~\eqref{eqn:double_con}.
\end{proof}

\begin{remark}
    This lemma can be viewed as a special case of a more general result for quantum supermaps. To elaborate, we recall that a supermap (a linear map from one quantum channel to another) can always be expressed as a pre and post-processing maps concatenated with the input quantum channel, and assisted by a memory \cite{chiribella2008transforming}. Then, Lemma~\ref{le:gen_fund_lower} follows from applying the data-processing inequality for generalized distinguishability measures between two quantum channels \cite{gour2019comparison} with respect to channel twirling, which is a valid quantum supermap.
\end{remark}

\begin{remark} \label{le:leditzky}
    In \cite{leditzky2018approaches}, the authors have shown that for any two covariant channels $\mathcal{F}^{A \rightarrow B}$ and $\mathcal{G}^{A \rightarrow B}$ with respect to $\{U_{A}(x), V_{B}(x)\}_{x \in \mathcal{X}}$ (namely that $\mathcal{V}_{x}\circ \mathcal{F}=\mathcal{F}\circ \mathcal{U}_{x}$ for all $x \in \mathcal{X}$, and similarly for $\mathcal{G}$), the generalized distinguishability measure
    \begin{equation}
        \mathbf{D}(\mathcal{F}, \mathcal{G})= \sup_{\phi}\mathbf{D}\left((\textsf{id}\otimes \mathcal{F})(\phi_{RA}), (\textsf{id}\otimes \mathcal{G})(\phi_{RA})\right) \; ,
    \end{equation}
    can be found by maximizing only over symmetric states $\phi_{RA}$, defined as
    \begin{equation}
        \frac{1}{\vert \mathcal{X} \vert}\sum_{x \in \mathcal{X}}U^{\dagger}_{A}(x)\phi_{RA}U_{A}(x)=\phi_{RA} \; .
    \end{equation}
    However, since the twirlings $\mathcal{F}\equiv \tilde{\mathcal{Q}}$ and $\mathcal{G}\equiv \tilde{\mathcal{S}}$ in Lemma~\ref{le:gen_fund_lower} are trivially covariant with respect to the unitary representations $\{U_{A}(x), V_{B}(x)\}_{x \in \mathcal{X}}$ of the finite group $\mathcal{X}$, this implies that the lower bound in Eq.~\eqref{eqn:gen_lower} need only be computed for such symmetric states. Furthermore, if $\{U_{A}(x)\}_{x\in \mathcal{X}}$ is a unitary 1-design (i.e. it is an irreducible representation of the group $\mathcal{X}$ of degree-$d_{A}$), then, using the property Eq.~\eqref{eqn:unitary_1} of unitary 1-designs, the maximization is found by computing the generalized distinguishability measure exactly for the maximally entangled state
    \begin{equation}
        \mathbf{D}(\tilde{\mathcal{Q}}, \tilde{\mathcal{S}})= \mathbf{D}((\textsf{id}^{R}\otimes \tilde{\mathcal{Q}}^{A})(\Phi_{RA}), (\textsf{id}^{R}\otimes \tilde{\mathcal{S}}^{A})(\Phi_{RA})) \; .
    \end{equation}
\end{remark}

So far, we have shown that the generalized distinguishability measure between $\mathcal{Q}^{A \rightarrow B}$ and $\mathcal{S}^{A \rightarrow B}$ is lower bounded by the corresponding distinguishability measure for arbitrary discrete twirlings of these channels. We now show that a similar lower bound can be derived for the uniform Haar twirling.  But first, we recall the following important result
\begin{lemma} \label{le:twirling}
(\cite{horodecki1999general}) Given a CPTP map $\mathcal{Q}^{A \rightarrow A}$ and for all $ \rho \in \mathcal{D}(\mathcal{H}^{A})$, the continuous twirling $\tilde{\mathcal{Q}}=\int_{\mathbf{U}(d)}dU\mathcal{U}^{\dagger}\circ \mathcal{Q}\circ \mathcal{U}$ over the uniform Haar measure on the set of $d \times d$ unitary matrices $\mathbf{U}(d)$ is given by the depolarizing channel
\begin{equation}
    \tilde{\mathcal{Q}}(\rho)=(1-p^{\mathcal{Q}})\rho+p^{\mathcal{Q}}\frac{I}{d} \; , 
    \label{eqn:depolarizing_ch}
\end{equation}
where the depolarizing parameter $p^{\mathcal{Q}}$ is given by the average fidelity of $\mathcal{Q}$, as follows
\begin{equation}
    p^{\mathcal{Q}}=\frac{d}{d-1}\left(1-F_{\text{avg}}(\mathcal{Q})\right) \; . \label{eqn:dep_parameter}
\end{equation}
\end{lemma}
The proof of Eq.~\eqref{eqn:depolarizing_ch} is shown in \cite{horodecki1999general, nielsen2002simple} for some parameter value $p^{\mathcal{Q}}$. Eq.~\eqref{eqn:dep_parameter} is a direct consequence of the fact that the uniform Haar twirled channel $\tilde{\mathcal{Q}}=\int_{\mathbf{U}(d)}dU\mathcal{U}^{\dagger}\circ \mathcal{Q}\circ \mathcal{U}$ has the same average fidelity as the original channel $\mathcal{Q}$ \cite{horodecki1999general}, along with the fact that the average fidelity of the depolarizing channel is given by 
\begin{equation}
    F_{\text{avg}}(\tilde{\mathcal{Q}})=1-\left(\frac{d-1}{d} \right)p \; ,
\end{equation}
where we have used the normalization $\int d\psi=1$ and the notation $p$ for the depolarizing parameter.

Using the above Lemmas~\ref{le:gen_fund_lower} and \ref{le:twirling}, we now establish a similar lower bound to that in Lemma~\ref{le:gen_fund_lower} for the uniform Haar twirl.

\begin{theorem}\label{th:main_th}
    Assume that we are given two CPTP maps $\mathcal{Q}^{A \rightarrow A}$ and $\mathcal{S}^{A \rightarrow A}$. If the generalized distinguishability measure $\mathbf{D}$ satisfies the joint convexity property, then $\mathbf{D}(\mathcal{Q}, \mathcal{S})$ is lower bounded by some function $l_{\mathbf{D}}$ of the channel entanglement fidelities $F_{e}(\mathcal{Q})$ and $F_{e}(\mathcal{S})$, as follows
    \begin{equation}
        \mathbf{D}(\mathcal{Q}, \mathcal{S}) \geq l_{\mathbf{D}}(F_{e}(\mathcal{Q}), F_{e}(\mathcal{S})) \; ,
    \end{equation}
    where the specific form of the function $l_{\mathbf{D}}$ depends on the choice of the generalized distinguishability measure and is determined by the uniform Haar twirls, as follows
    \begin{equation}
        l_{\mathbf{D}}(F_{e}(\mathcal{Q}), F_{e}(\mathcal{S})) \equiv \mathbf{D}(\tilde{\mathcal{Q}}, \tilde{\mathcal{S}})  \; ,
    \end{equation}
    where $\tilde{\mathcal{Q}}=\int_{\mathbf{U}(d)}dU\mathcal{U}^{\dagger}\circ \mathcal{Q}\circ \mathcal{U}$ and $\tilde{\mathcal{S}}=\int_{\mathbf{U}(d)}dU\mathcal{U}^{\dagger}\circ \mathcal{S}\circ \mathcal{U}$ yield two depolarizing channels. The inequality is saturated if the joint convexity property is saturated for a set of unitary 2-designs and a uniform probability distribution over this set.
\end{theorem}
\begin{proof}
    This is a direct consequence of applying Lemma~\ref{le:gen_fund_lower} to any unitary 2-design $\{ U_{A}(x) \}_{x \in \mathcal{X}}$, e.g. the unitary representation of the Clifford group (see Appendix~\ref{sec:t-design}), along with a uniform distribution on $\mathcal{X}$, and then using the property of unitary 2-designs in Eq.~\eqref{eqn:unitary_2_des}, which finally yields
    \begin{equation}
    \mathbf{D}(\mathcal{Q}, \mathcal{S}) \ge \mathbf{D}(\tilde{\mathcal{Q}}, \tilde{\mathcal{S}}) \; , \label{eqn:uniform_lower}
    \end{equation}
        where $\tilde{\mathcal{Q}}=\int_{\mathbf{U}(d)}dU\mathcal{U}^{\dagger}\circ \mathcal{Q}\circ \mathcal{U}$ and $\tilde{\mathcal{S}}=\int_{\mathbf{U}(d)}dU\mathcal{U}^{\dagger}\circ \mathcal{S}\circ \mathcal{U}$. The proof is completed by applying Lemma~\ref{le:twirling} and plugging in the depolarizing channels into the lower bound in Eq.~\eqref{eqn:uniform_lower}.
\end{proof}

It directly follows from this proof that the image of the function $l_{\mathbf{D}}$ coincides with the image of the corresponding generalized distinguishability measure $\mathbf{D}$.

\begin{remark}
    The lower bound proof does not require faithfulness, symmetry, nor the triangle inequality, which would also make $\mathbf{D}$ a generalized distance measure. However, the triangle inequality becomes necessary when deriving an upper bound for the generalized distinguishability measure for concatenated noisy channels (or gates), as it is relevant to fault-tolerant quantum computing (see Appendix~\ref{sec:upper_bound} for more details). 
\end{remark}

\section{Comment on The Chaining Property} \label{sec:chaining}
In quantum computing literature, one encounters the chaining property for distance measures \cite{kueng2016comparing}, which is useful for computing upper bounds on error propagation in fault-tolerant quantum computing. This property is framed as follows: Assume we want to apply two maps $\mathcal{Q}$ and $\mathcal{S}$ in series, however, we only have access to their noisy versions, which we denote by $\mathcal{Q}^{\prime}$ and $\mathcal{S}^{\prime}$, respectively. If the generalized distance measure $\mathbf{D}$ also satisfies the DPI (i.e. $\mathbf{D}$ is also a distinguishability measure), then the chaining property reads
\begin{align}
    \mathbf{D}(\mathcal{S}\circ \mathcal{Q}(\rho), \mathcal{S}^{\prime}\circ \mathcal{Q}^{\prime}(\rho)) \leq \mathbf{D}(\mathcal{Q}(\rho), \mathcal{Q}^{\prime}(\rho)) \nonumber \\ + \mathbf{D}(\mathcal{S}(\rho), \mathcal{S}^{\prime}(\rho)) \; ,
\end{align}
for all $\rho \in \mathcal{D}(\mathcal{H})$. This is interpreted by saying that the error due to a consecutive application of two faulty channels is no larger than the sum of the errors of applying each of the faulty channels separately. The proof follows by first applying the triangle inequality to the left-hand side of the above inequality, followed up by the date-processing inequality. Therefore, the desirable chaining property is derivative from other, more fundamental, properties of $\mathbf{D}$.

\section{Upper-Bounding Generalized Distance Measures for State Recovery} \label{sec:upper_bound}
Here we present upper bounds on generalized distance and distinguishability measures, showing how they get modified when limited knowledge about the noise parameter $\theta \in \Theta$ is available, both for the single-cycle and multi-cycle cases. Similar to the chaining property, the derivation of upper bounds on generalized distance and distinguishability measures is important for the analysis of error propagation in noisy quantum processes.

\subsection{Single-Cycle Case}
Consider the distance measure $\mathbf{D}$ and assume that for all $ \rho \in \mathcal{D}(\mathcal{C}) \subseteq \mathcal{D}(\mathcal{H})$, approximate recovery from the noise $\mathcal{N}_{\theta}$ is possible in the presence of perfect information about $\theta$, i.e. there exists $\mathcal{R}_{\theta}$ such that
\begin{equation}
    \mathbf{D}(\mathcal{I}^{\theta}_{\theta}(\rho), \rho) \leq \epsilon_{\theta} \hspace{0.2cm} \text{where} \hspace{0.2cm} \mathcal{I}^{\beta}_{\alpha}\equiv \mathcal{R}_{\beta} \circ \mathcal{N}_{\alpha} \; . \label{eqn:upper_1}
\end{equation}
Now consider the distance measure $\mathbf{D}(\mathcal{I}^{\hat{\theta}}_{\theta}(\rho), \rho)$, where $\hat{\theta}$ is the best unbiased estimate of $\theta \in \Theta$. Our goal is to bound this quantity from above by two terms: the first depends on how well we can bound the same distance measure when given perfect knowledge about $\theta$ (see Eq.~\eqref{eqn:upper_1}), and the second should measure our lack of knowledge of the noise parameter $\theta$. This intuition is validated by applying the triangle inequality, as follows
\begin{align}
    \mathbf{D}(\mathcal{I}^{\hat{\theta}}_{\theta}(\rho), \rho) &\leq \mathbf{D}(\mathcal{I}^{\hat{\theta}}_{\theta}(\rho), \mathcal{I}^{\hat{\theta}}_{\hat{\theta}}(\rho)) + \mathbf{D}(\mathcal{I}^{\hat{\theta}}_{\hat{\theta}}(\rho), \rho) \\
    & \leq \mathbf{D}(\mathcal{N}_{\theta}(\rho), \mathcal{N}_{\hat{\theta}}(\rho))+ \mathbf{D}(\mathcal{I}^{\hat{\theta}}_{\hat{\theta}}(\rho), \rho) \\
    & \leq \mathbf{D}(\mathcal{N}_{\theta}, \mathcal{N}_{\hat{\theta}})+ \mathbf{D}(\mathcal{I}^{\hat{\theta}}_{\hat{\theta}}(\rho), \rho) \\
    & \leq \mathbf{D}(\mathcal{N}_{\theta}, \mathcal{N}_{\hat{\theta}})+ \epsilon_{\hat{\theta}} \equiv \epsilon_{\theta, \hat{\theta}} \; , \label{eqn:simple_upper}
\end{align}
where the second inequality follows from the DPI of $\mathbf{D}$, the third follows from the definition of the generalized distance for channels, and the fourth from the assumption of Eq.~\eqref{eqn:upper_1}. It is worth noting that one can derive a similar upper bound using the recoveries, rather than the noisy channels. The advantage of this approach is that we do not need to assume that $\mathbf{D}$ satisfies the DPI, i.e. it suffices for $\mathbf{D}$ to be a distance measure. To see how we simply apply the triangle inequality
\begin{align}
    \mathbf{D}(\mathcal{I}^{\hat{\theta}}_{\theta}(\rho), \rho) &\leq \mathbf{D}(\mathcal{I}^{\hat{\theta}}_{\theta}(\rho), \mathcal{I}^{\theta}_{\theta}(\rho)) + \mathbf{D}(\mathcal{I}^{\theta}_{\theta}(\rho), \rho) \\
    & \leq \mathbf{D}(\mathcal{R}_{\theta}, \mathcal{R}_{\hat{\theta}})+ \epsilon_{\theta} \equiv \epsilon^{\prime}_{\theta, \hat{\theta}} \; ,
\end{align}
where we have used the definition of a distance measure between channels for the second inequality, as well as Eq.~\eqref{eqn:upper_1}. We will shortly show that DPI becomes necessary when considering the multi-cycle case.

\subsection{Multi-Cycle Case}
Let us now extend the upper bound previously derived in the single-cycle case to adaptive multi-cycle recovery. Using the shorthand notation 
\begin{equation}
    \mathbf{D}^{\beta_{n}\beta_{n-1}\cdots \beta_{1}}_{\alpha_{n}\alpha_{n-1}\cdots \alpha_{1}}(\rho) \equiv \mathbf{D}(\mathcal{I}^{\beta_{n}\beta_{n-1}\cdots \beta_{1}}_{\alpha_{n}\alpha_{n-1}\cdots \alpha_{1}}(\rho), \rho) \; ,
\end{equation}
where 
\begin{equation}
    \mathcal{I}^{\beta_{n}\beta_{n-1}\cdots \beta_{1}}_{\alpha_{n}\alpha_{n-1}\cdots \alpha_{1}}\equiv \mathcal{R}_{\beta_{n}} \circ \mathcal{E}_{\alpha_{n}} \circ \mathcal{R}_{\beta_{n-1}} \circ \mathcal{E}_{\alpha_{n-1}} \cdots \circ \mathcal{R}_{\beta_{1}} \circ \mathcal{E}_{\alpha_{1}} \; ,
\end{equation}
and applying the triangle inequality, we get
\begin{align}
    \mathbf{D}^{\hat{\theta}_{n}\hat{\theta}_{n-1}\cdots \hat{\theta}_{1}}_{\theta_{n}\theta_{n-1}\cdots \theta_{1}}(\rho) &\leq  \mathbf{D}(\mathcal{I}^{\hat{\theta}_{n}\hat{\theta}_{n-1}\cdots \hat{\theta}_{1}}_{\theta_{n}\theta_{n-1}\cdots \theta_{1}}(\rho), \mathcal{I}^{\hat{\theta}_{n}\hat{\theta}_{n-1}\cdots \hat{\theta}_{1}}_{\hat{\theta}_{n}\hat{\theta}_{n-1}\cdots \hat{\theta}_{1}}(\rho)) \nonumber \\ &+ \mathbf{D}^{\hat{\theta}_{n}\hat{\theta}_{n-1}\cdots \hat{\theta}_{1}}_{\hat{\theta}_{n}\hat{\theta}_{n-1}\cdots \hat{\theta}_{1}}(\rho) \; . \label{eqn:true_upper}
\end{align}
The second term could be bounded from above by the individual errors $\{ \epsilon_{\hat{\theta}_{i}} \}_{i=1}^{n}$, using only the triangle inequality, as follows
\begin{align}
    \mathbf{D}^{\hat{\theta}_{n}\hat{\theta}_{n-1}\cdots \hat{\theta}_{1}}_{\hat{\theta}_{n}\hat{\theta}_{n-1}\cdots \hat{\theta}_{1}}(\rho) &\leq \mathbf{D}(\mathcal{I}^{\hat{\theta}_{n}\hat{\theta}_{n-1}\cdots \hat{\theta}_{1}}_{\hat{\theta}_{n}\hat{\theta}_{n-1}\cdots \hat{\theta}_{1}}(\rho), \mathcal{I}^{\hat{\theta}_{n-1}\cdots \hat{\theta}_{1}}_{\hat{\theta}_{n-1}\cdots \hat{\theta}_{1}}(\rho)) \nonumber \\ &+ \mathbf{D}^{\hat{\theta}_{n-1}\cdots \hat{\theta}_{1}}_{\hat{\theta}_{n-1}\cdots \hat{\theta}_{1}}(\rho) \; \label{eqn:ideal_upper}.
\end{align}
We assume that
\begin{equation}
    \mathcal{I}^{\hat{\theta}_{n-1}\cdots \hat{\theta}_{1}}_{\hat{\theta}_{n-1}\cdots \hat{\theta}_{1}}(\rho) \in \mathcal{D}(\mathcal{C}) \; , \label{eqn:possible_AQEC}
\end{equation} so that the $n$-th step approximate recovery with perfect knowledge of $\theta$ would be possible, in principle. This leads to
\begin{align}
    &\mathbf{D}(\mathcal{I}^{\hat{\theta}_{n}\hat{\theta}_{n-1}\cdots \hat{\theta}_{1}}_{\hat{\theta}_{n}\hat{\theta}_{n-1}\cdots \hat{\theta}_{1}}(\rho), \mathcal{I}^{\hat{\theta}_{n-1}\cdots \hat{\theta}_{1}}_{\hat{\theta}_{n-1}\cdots \hat{\theta}_{1}}(\rho)) \\ 
    &= \mathbf{D}(\mathcal{I}_{\hat{\theta}_{n}}^{\hat{\theta}_{n}}(\mathcal{I}^{\hat{\theta}_{n-1}\cdots \hat{\theta}_{1}}_{\hat{\theta}_{n-1}\cdots \hat{\theta}_{1}}(\rho)), \mathcal{I}^{\hat{\theta}_{n-1}\cdots \hat{\theta}_{1}}_{\hat{\theta}_{n-1}\cdots \hat{\theta}_{1}}(\rho)) \leq \epsilon_{\hat{\theta}_{n}} \; .
\end{align}
Substituting this result back into Eq.~\eqref{eqn:ideal_upper}, we get
\begin{align}
    \mathbf{D}^{\hat{\theta}_{n}\hat{\theta}_{n-1}\cdots \hat{\theta}_{1}}_{\hat{\theta}_{n}\hat{\theta}_{n-1}\cdots \hat{\theta}_{1}}(\rho) \leq \mathbf{D}^{\hat{\theta}_{n-1}\cdots \hat{\theta}_{1}}_{\hat{\theta}_{n-1}\cdots \hat{\theta}_{1}}(\rho)+\epsilon_{\hat{\theta}_{n}} \; ,
\end{align}
and repeating the above two steps yields
\begin{equation}
    \mathbf{D}^{\hat{\theta}_{n}\hat{\theta}_{n-1}\cdots \hat{\theta}_{1}}_{\hat{\theta}_{n}\hat{\theta}_{n-1}\cdots \hat{\theta}_{1}} \leq \sum_{i=1}^{n} \epsilon_{\hat{\theta}_{i}} \; . \label{eqn:1realAQEC}
\end{equation}
The first term in Eq.~\eqref{eqn:true_upper} is a new error term due to the real-time (drift-adapting) nature of our setup. This term can be bounded from above using the chaining property and the DPI, as follows
\begin{align}
    & \mathbf{D}(\mathcal{I}^{\hat{\theta}_{n}\hat{\theta}_{n-1}\cdots \hat{\theta}_{1}}_{\theta_{n}\theta_{n-1}\cdots \theta_{1}}(\rho), \mathcal{I}^{\hat{\theta}_{n}\hat{\theta}_{n-1}\cdots \hat{\theta}_{1}}_{\hat{\theta}_{n}\hat{\theta}_{n-1}\cdots \hat{\theta}_{1}}(\rho)) \notag \\
    & = \mathbf{D}(\mathcal{I}^{\hat{\theta}_{n}}_{\theta_{n}}\circ \mathcal{I}^{\hat{\theta}_{n-1}\cdots \hat{\theta}_{1}}_{\theta_{n-1}\cdots \theta_{1}}(\rho), \mathcal{I}^{\hat{\theta}_{n}}_{\hat{\theta}_{n}}\circ \mathcal{I}^{\hat{\theta}_{n-1}\cdots \hat{\theta}_{1}}_{\hat{\theta}_{n-1}\cdots \hat{\theta}_{1}}(\rho)) \\
    & \leq \mathbf{D}(\mathcal{I}^{\hat{\theta}_{n}}_{\theta_{n}}(\rho), \mathcal{I}^{\hat{\theta}_{n}}_{\hat{\theta}_{n}}(\rho)) + \mathbf{D}( \mathcal{I}^{\hat{\theta}_{n-1}\cdots \hat{\theta}_{1}}_{\theta_{n-1}\cdots \theta_{1}}(\rho), \mathcal{I}^{\hat{\theta}_{n-1}\cdots \hat{\theta}_{1}}_{\hat{\theta}_{n-1}\cdots \hat{\theta}_{1}}(\rho)) \\
    & \leq \mathbf{D}(\mathcal{N}_{\theta_{n}}(\rho), \mathcal{N}_{\hat{\theta}_{n}}(\rho)) + \mathbf{D}( \mathcal{I}^{\hat{\theta}_{n-1}\cdots \hat{\theta}_{1}}_{\theta_{n-1}\cdots \theta_{1}}(\rho), \mathcal{I}^{\hat{\theta}_{n-1}\cdots \hat{\theta}_{1}}_{\hat{\theta}_{n-1}\cdots \hat{\theta}_{1}}(\rho)) \\
    & \leq \mathbf{D}(\mathcal{N}_{\theta_{n}}, \mathcal{N}_{\hat{\theta}_{n}}) + \mathbf{D}( \mathcal{I}^{\hat{\theta}_{n-1}\cdots \hat{\theta}_{1}}_{\theta_{n-1}\cdots \theta_{1}}(\rho), \mathcal{I}^{\hat{\theta}_{n-1}\cdots \hat{\theta}_{1}}_{\hat{\theta}_{n-1}\cdots \hat{\theta}_{1}}(\rho))\; .
\end{align}
Repeating the above steps $n-1$ times, we arrive at the upper bound 
\begin{equation}
    \mathbf{D}(\mathcal{I}^{\hat{\theta}_{n}\hat{\theta}_{n-1}\cdots \hat{\theta}_{1}}_{\theta_{n}\theta_{n-1}\cdots \theta_{1}}(\rho), \mathcal{I}^{\hat{\theta}_{n}\hat{\theta}_{n-1}\cdots \hat{\theta}_{1}}_{\hat{\theta}_{n}\hat{\theta}_{n-1}\cdots \hat{\theta}_{1}}(\rho)) \leq \sum_{i=1}^{n}\mathbf{D}(\mathcal{N}_{\theta_{i}}, \mathcal{N}_{\hat{\theta}_{i}}) \; . \label{eqn:2realAQEC}
\end{equation}
Combining Eqs.~\eqref{eqn:1realAQEC} and  \eqref{eqn:2realAQEC} with Eq.~\eqref{eqn:true_upper}, we get
\begin{equation}
    D^{\hat{\theta}_{n}\hat{\theta}_{n-1}\cdots \hat{\theta}_{1}}_{\theta_{n}\theta_{n-1}\cdots \theta_{1}} \leq \sum_{i=1}^{n}[\mathbf{D}(\mathcal{N}_{\theta_{i}}, \mathcal{N}_{\hat{\theta}_{i}})+\epsilon_{\hat{\theta}_{i}}] \equiv \sum_{i=1}^{n}\epsilon_{\theta_{i}, \hat{\theta}_{i}} \; ,
\end{equation}
which generalizes Eq.~\eqref{eqn:simple_upper} for real-time approximate recovery. This result says that, if AQEC is possible in principle (see Eq.~\eqref{eqn:possible_AQEC}) when perfect knowledge of $\theta$ is available, then AQEC is also possible when knowledge about $\theta$ is limited. As we have shown, this holds for both the single-cycle and multi-cycle regimes.

Alternatively, we can derive an upper bound that is a function of the recoveries, rather than the noisy channels. This is accomplished as follows
\begin{align}
    \mathbf{D}^{\hat{\theta}_{n}\hat{\theta}_{n-1}\cdots \hat{\theta}_{1}}_{\theta_{n}\theta_{n-1}\cdots \theta_{1}}(\rho) &\leq  \mathbf{D}(\mathcal{I}^{\hat{\theta}_{n}\hat{\theta}_{n-1}\cdots \hat{\theta}_{1}}_{\theta_{n}\theta_{n-1}\cdots \theta_{1}}(\rho), \mathcal{I}^{\theta_{n}\theta_{n-1}\cdots \theta_{1}}_{\theta_{n}\theta_{n-1}\cdots \theta_{1}}(\rho)) \nonumber \\ &+ \mathbf{D}^{\theta_{n}\theta_{n-1}\cdots \theta_{1}}_{\theta_{n}\theta_{n-1}\cdots \theta_{1}}(\rho) \; , \label{eqn:2AQEC}
\end{align}
where the second term is similarly bounded from above by $\sum_{i=1}^{n}\epsilon_{\theta_{i}}$, based only on the triangle inequality (see Eq.~\eqref{eqn:1realAQEC}). We now upper bound the first term in the above inequality as
\begin{align}
    & \mathbf{D}(\mathcal{I}^{\hat{\theta}_{n}\hat{\theta}_{n-1}\cdots \hat{\theta}_{1}}_{\theta_{n}\theta_{n-1}\cdots \theta_{1}}(\rho), \mathcal{I}^{\theta_{n}\theta_{n-1}\cdots \theta_{1}}_{\theta_{n}\theta_{n-1}\cdots \theta_{1}}(\rho)) \nonumber \\
    & \leq \mathbf{D}(\mathcal{I}^{\hat{\theta}_{n}\hat{\theta}_{n-1}\cdots \hat{\theta}_{1}}_{\theta_{n}\theta_{n-1}\cdots \theta_{1}}(\rho), \mathcal{I}^{\hat{\theta}_{n}\theta_{n-1}\cdots \theta_{1}}_{\theta_{n}\theta_{n-1}\cdots \theta_{1}}(\rho)) \nonumber \\
    & +\mathbf{D}(\mathcal{I}^{\hat{\theta}_{n}\theta_{n-1}\cdots \theta_{1}}_{\theta_{n}\theta_{n-1}\cdots \theta_{1}}(\rho), \mathcal{I}^{\theta_{n}\theta_{n-1}\cdots \theta_{1}}_{\theta_{n}\theta_{n-1}\cdots \theta_{1}}(\rho)) \\
    & \leq \mathbf{D}(\mathcal{I}^{\hat{\theta}_{n-1}\cdots \hat{\theta}_{1}}_{\theta_{n-1}\cdots \theta_{1}}(\rho), \mathcal{I}^{\theta_{n-1}\cdots \theta_{1}}_{\theta_{n-1}\cdots \theta_{1}}(\rho))+\mathbf{D}(\mathcal{R}_{\hat{\theta}_{n}}, \mathcal{R}_{\theta_{n}}) \; ,
\end{align}
where we have used the triangle inequality for the first inequality and the DPI and the definition of generalized distance measure between channels for the second inequality. By repeating these two steps $n-1$ times, we arrive at
\begin{align}
    & \mathbf{D}(\mathcal{I}^{\hat{\theta}_{n}\hat{\theta}_{n-1}\cdots \hat{\theta}_{1}}_{\theta_{n}\theta_{n-1}\cdots \theta_{1}}(\rho), \mathcal{I}^{\theta_{n}\theta_{n-1}\cdots \theta_{1}}_{\theta_{n}\theta_{n-1}\cdots \theta_{1}}(\rho)) \leq \sum_{i=1}^{n} \mathbf{D}(\mathcal{R}_{\hat{\theta}_{i}}, \mathcal{R}_{\theta_{i}}) \; .
\end{align}
Consequently, Eq.~\eqref{eqn:2AQEC} yields
\begin{equation}
    \mathbf{D}^{\hat{\theta}_{n}\hat{\theta}_{n-1}\cdots \hat{\theta}_{1}}_{\theta_{n}\theta_{n-1}\cdots \theta_{1}} \leq \sum_{i=1}^{n}[\mathbf{D}(\mathcal{R}_{\theta_{i}}, \mathcal{R}_{\hat{\theta}_{i}})+\epsilon_{\theta_{i}}] \equiv \sum_{i=1}^{n}\epsilon_{\theta_{i}, \hat{\theta}_{i}} \; .
\end{equation}

\section{Necessary and Sufficient Condition for Independence of the Reduced Channel From a Mother Channel Parameter} \label{apx:param_indep}
In this appendix, we are interested in proving the following:
\begin{lemma}
    Consider the mother channel $\mathcal{Z}_{\bm{\theta}}^{MS}$ of a bipartite system $MS$, where $\bm{\theta} \in \Theta^{p}$ is a $p$-dimensional parameter space. The reduced dynamics of subsystem $S$, defined by partial tracing over the subsystem $M$ via $\mathcal{M}_{\bm{\theta}}^{S} \equiv \operatorname{Tr}_{M}\circ \mathcal{Z}_{\bm{\theta}}^{MS}$, is independent of the $\alpha$-th component of the $p$-dimensional vector $\bm{\theta}$ if and only if the Choi matrix $\Gamma^{\mathcal{Z}_{\bm{\theta}}}$ of the mother channel satisfies
    \begin{equation}
        \operatorname{Tr}_{M^{\prime \prime} S^{\prime \prime}}\left[  \left( \frac{\partial}{\partial \theta_{\alpha}}\Gamma^{\mathcal{Z}_{\bm{\theta}}}_{MS,M^{\prime \prime }S^{\prime \prime}} \right)(2P_{S^{\prime \prime}S^{\prime}}^{\text{sym}}-I_{S^{\prime \prime}S^{\prime}}) \right]=0 \; , 
    \end{equation}
    where $P^{\text{sym}}_{S^{\prime \prime}S^{\prime}}$ is the projector onto the symmetric subspace 
	\begin{equation}
		\text{span}(|\mu \rangle_{S^{\prime \prime}} |\nu \rangle_{S^{\prime}} + |\nu \rangle_{S^{\prime \prime}} |\mu \rangle_{S^{\prime}}) \subset \mathcal{H}^{S^{\prime \prime}}\otimes \mathcal{H}^{S^{\prime}} \; ,
	\end{equation}
    and $\{|\mu \rangle\}$ is a basis set of the Hilbert space $\mathcal{H}^{S}$.
\end{lemma}
\begin{proof}
    We consider the dynamics of the combined memory-spectator $(MS)$ system, and express it in terms of the Choi matrix of the mother channel $\mathcal{Z}_{\bm{\theta}}^{MS}$, as follows \cite{khatri2020principles}
\begin{align}
    &\mathcal{Z}_{\bm{\theta}}^{MS \rightarrow M^{\prime}S^{\prime}}(\rho_{MS}) \nonumber \\ &= \mathrm{Tr}_{MS}\left[ (\mathbb{T}_{MS}(\rho_{MS})\otimes I_{M^{\prime}S^{\prime}}) \Gamma^{\mathcal{Z}_{\bm{\theta}}}_{MS,M^{\prime}S^{\prime}} \right] \; ,
\end{align}
then the reduced dynamics of the spectator system yields
\begin{align}
    &\operatorname{Tr}_{M^{\prime}} \circ \mathcal{Z}_{\bm{\theta}}^{MS \rightarrow M^{\prime}S^{\prime}}(\rho_{MS}) \nonumber \\ &= \mathrm{Tr}_{MS}\left[ (\mathbb{T}_{MS}(\rho_{MS})\otimes I_{S^{\prime}}) \Gamma^{\operatorname{Tr}_{M}\circ \mathcal{Z}_{\bm{\theta}}}_{MS,S^{\prime}} \right] \; .
\end{align}
Therefore, for the reduced dynamics to be independent of the noise parameter $\theta_{\alpha}$ for some $\alpha=1, \cdots, p$ and any joint input state $\rho_{MS}$, we must have
\begin{equation}
    \frac{\partial}{\partial \theta_{\alpha}}\Gamma^{\operatorname{Tr}_{M}\circ \mathcal{Z}_{\bm{\theta}}}_{MS,S^{\prime}}=0 \; , \label{eqn:sufficient_cond}
\end{equation}
for all $\bm{\theta} \in \Theta^{p}$. We now derive a necessary and sufficient condition for this equality to hold, in terms of the Choi matrix $\Gamma^{\mathcal{Z}_{\bm{\theta}}}_{MS,M^{\prime}S^{\prime}}$ of the mother channel. We start by recalling the formula for the Choi matrix of the composite channel in terms of the Choi matrices of the individual channels \cite{khatri2020principles}
\begin{align}
    & \Gamma^{\operatorname{Tr}_{M}\circ \mathcal{Z}_{\bm{\theta}}}_{MS,S^{\prime}} \nonumber \\ &= \operatorname{Tr}_{M^{\prime \prime} S^{\prime \prime}}\left[ \mathbb{T}_{M^{\prime \prime}S^{\prime \prime}}\left( \Gamma^{\mathcal{Z}_{\bm{\theta}}}_{MS,M^{\prime \prime }S^{\prime \prime}} \right) \Gamma_{M^{\prime \prime} S^{\prime \prime}, S^{\prime}}^{\operatorname{Tr}_{M}} \right] \\ &=
    \operatorname{Tr}_{M^{\prime \prime} S^{\prime \prime}}\left[ \Gamma^{\mathcal{Z}_{\bm{\theta}}}_{MS,M^{\prime \prime }S^{\prime \prime}} \mathbb{T}^{\dagger}_{M^{\prime \prime}S^{\prime \prime}}\left( \Gamma_{M^{\prime \prime} S^{\prime \prime}, S^{\prime}}^{\operatorname{Tr}_{M}} \right)  \right] \\ &=
    \operatorname{Tr}_{M^{\prime \prime} S^{\prime \prime}}\left[ \Gamma^{\mathcal{Z}_{\bm{\theta}}}_{MS,M^{\prime \prime }S^{\prime \prime}} \mathbb{T}_{M^{\prime \prime}S^{\prime \prime}}\left( \Gamma_{M^{\prime \prime} S^{\prime \prime}, S^{\prime}}^{\operatorname{Tr}_{M}} \right)  \right] \; . \label{eqn:sufficient_param_indep}
\end{align}
We now compute the basis dependent matrix $\mathbb{T}_{M^{\prime \prime}S^{\prime \prime}}\left( \Gamma_{M^{\prime \prime} S^{\prime \prime}, S^{\prime}}^{\operatorname{Tr}_{M}} \right)$ in the separable memory-spectator basis $|i\rangle_{MS}\equiv |i(a,\mu)\rangle_{MS}=|a\rangle_{M}|\mu\rangle_{S}$ of the Hilbert space $\mathcal{H}^{M}\otimes \mathcal{H}^{S}$, as follows
\begin{align}
    \Gamma_{M^{\prime \prime} S^{\prime \prime}, S^{\prime}}^{\operatorname{Tr}_{M}} &=\sum_{ij}|i\rangle \! \langle j|_{M^{\prime \prime}S^{\prime \prime}} \otimes \operatorname{Tr}_{M^{\prime}}\left( |i\rangle \! \langle j|_{M^{\prime}S^{\prime}} \right) \\
    &= I_{M^{\prime \prime}}\otimes \Gamma_{S^{\prime \prime}S^{\prime}} \; ,
\end{align}
where $|\Gamma \rangle_{S^{\prime \prime}S^{\prime}} = \sum_{\mu}|\mu \rangle_{S^{\prime \prime}} |\mu \rangle_{S^{\prime}}$ is the maximally entangled state in the special spectator basis $\{|\mu \rangle\}$. In the same $\{|a\rangle \otimes |\mu \rangle\}$ basis, the partial transpose yields
\begin{align}
    \mathbb{T}_{M^{\prime \prime}S^{\prime \prime}}\left( \Gamma_{M^{\prime \prime} S^{\prime \prime}, S^{\prime}}^{\operatorname{Tr}_{M}} \right) &=\mathbb{T}_{M^{\prime \prime}S^{\prime \prime}}\left(  I_{M^{\prime \prime}}\otimes \Gamma_{S^{\prime \prime}S^{\prime}} \right) \\
    &= I_{M^{\prime \prime}}\otimes((T_{S^{\prime \prime}}\otimes \textsf{id}^{S^{\prime}})(\Gamma_{S^{\prime \prime} S^{\prime}})) \\
    &= I_{M^{\prime \prime}}\otimes((I^{S^{\prime \prime}}\otimes \mathbb{T}_{S^{\prime}})(\Gamma_{S^{\prime \prime} S^{\prime}})) \\
    &= I_{M^{\prime \prime}}\otimes \Gamma_{S^{\prime \prime}S^{\prime}}^{\mathbb{T}} \; ,
\end{align}
where $\Gamma_{S^{\prime \prime}S^{\prime}}^{\mathbb{T}}$ is the Choi matrix of the partial transpose channel. It has been shown in cite{johnston2011quantum} that the Choi matrix $\Gamma_{S^{\prime \prime}S^{\prime}}^{\mathbb{T}}$ is related to the projector $P^{\text{sym}}_{S^{\prime \prime}S^{\prime}}=(I_{S^{\prime \prime} S^{\prime}}+\Gamma_{S^{\prime \prime}S^{\prime}}^{\mathbb{T}})/2$ onto the symmetric subspace 
	\begin{equation}
		\text{span}(|\mu \rangle_{S^{\prime \prime}} |\nu \rangle_{S^{\prime}} + |\nu \rangle_{S^{\prime \prime}} |\mu \rangle_{S^{\prime}}) \subset \mathcal{H}^{S^{\prime \prime}}\otimes \mathcal{H}^{S^{\prime}} \; ,
	\end{equation}
where $d_{S}(d_{S}+1)/2$ is the dimensions of the symmetric subspace in the $d_{S}^{2}$ dimensional Hilbert space $\mathcal{H^{S}}\otimes \mathcal{H^{S}}$. This finally yields
\begin{align}
    \mathbb{T}_{M^{\prime \prime}S^{\prime \prime}}\left( \Gamma_{M^{\prime \prime} S^{\prime \prime}, S^{\prime}}^{\operatorname{Tr}_{M}} \right) &= I_{M^{\prime \prime}}\otimes \Gamma_{S^{\prime \prime}S^{\prime}}^{\mathbb{T}} \\
    &= I_{M^{\prime \prime}}\otimes (2P_{S^{\prime \prime}S^{\prime}}^{\text{sym}}-I_{S^{\prime \prime}S^{\prime}}) \; .
\end{align}
Substituting back into Eq.~\eqref{eqn:sufficient_param_indep} leads to 
\begin{align}
    & \Gamma^{\operatorname{Tr}_{M}\circ \mathcal{Z}_{\bm{\theta}}}_{MS,S^{\prime}} \nonumber \\ &=
    \operatorname{Tr}_{M^{\prime \prime} S^{\prime \prime}}\left[ \Gamma^{\mathcal{Z}_{\bm{\theta}}}_{MS,M^{\prime \prime }S^{\prime \prime}}(2P_{S^{\prime \prime}S^{\prime}}^{\text{sym}}-I_{S^{\prime \prime}S^{\prime}}) \right] \; .
\end{align}
Therefore, Eq.~\eqref{eqn:sufficient_cond} holds if and only if
\begin{equation}
    \operatorname{Tr}_{M^{\prime \prime} S^{\prime \prime}}\left[  \left( \frac{\partial}{\partial \theta_{\alpha}}\Gamma^{\mathcal{Z}_{\bm{\theta}}}_{MS,M^{\prime \prime }S^{\prime \prime}} \right)(2P_{S^{\prime \prime}S^{\prime}}^{\text{sym}}-I_{S^{\prime \prime}S^{\prime}}) \right]=0 \; . \label{eqn:iff}
\end{equation}
\end{proof}
The last equation in the proof can be equivalently written as (using the identity $P_{S^{\prime \prime}S^{\prime}}^{\text{sym}}+(P_{S^{\prime \prime}S^{\prime}}^{\text{sym}})^{\perp}=I_{S^{\prime \prime}S^{\prime}}$)
\begin{align}
    &\operatorname{Tr}_{M^{\prime \prime} S^{\prime \prime}}\left[ \left( \frac{\partial}{\partial \theta_{\alpha}}\Gamma^{\mathcal{Z}_{\bm{\theta}}}_{MS,M^{\prime \prime }S^{\prime \prime}} \right)P_{S^{\prime \prime}S^{\prime}}^{\text{sym}}\right] \nonumber \\ 
    &= \operatorname{Tr}_{M^{\prime \prime} S^{\prime \prime}}\left[ \left( \frac{\partial}{\partial \theta_{\alpha}}\Gamma^{\mathcal{Z}_{\bm{\theta}}}_{MS,M^{\prime \prime }S^{\prime \prime}} \right)\left(P_{S^{\prime \prime}S^{\prime}}^{\text{sym}}\right)^{\perp}\right] \\
    &=\frac{1}{2}\frac{\partial}{\partial \theta_{\alpha}}\operatorname{Tr}_{M^{\prime \prime} S^{\prime \prime}}\left[ \Gamma^{\mathcal{Z}_{\bm{\theta}}}_{MS,M^{\prime \prime }S^{\prime \prime}}\right] \\ &=\frac{1}{2}\frac{\partial}{\partial \theta_{\alpha}}I_{MS}=0 \; .
\end{align}
\begin{remark}
    Note that the condition Eq.~\ref{eqn:iff} is weaker than
\begin{equation}
    \frac{\partial}{\partial \theta_{\alpha}}\Gamma^{\mathcal{Z}_{\bm{\theta}}}_{MS,M^{\prime \prime }S^{\prime \prime}}=0 \; ,
\end{equation}
for all $\bm{\theta}\in \Theta^{p}$, which holds when the mother channel $\mathcal{Z}_{\bm{\theta}}^{MS}$ itself does not depend on the noise parameter $\theta_{\alpha}$, and hence trivially also the reduced channel $\mathcal{M}^{S}_{\bm{\theta}}\equiv \operatorname{Tr}_{M}\circ \mathcal{Z}_{\bm{\theta}}^{MS}$.
\end{remark}

\section{Quantum Fisher Information Matrix}
\label{apx:QFIM}
We review the relevant definitions and results regarding quantum Fisher Information Matrix (QFIM) and the partial QFIM, following \cite{suzuki2020quantum, liu2020quantum}.
\subsection{Useful Definitions}
For a family of parameterized quantum states $\{ \rho_{\bm{\theta}} \}_{\bm{\theta}\in \Theta}$ with a $p$-dimensional parameter space $\Theta^{p} \subseteq \mathbb{R}^{p}$, we define the symmetric inner product between two linear operators $A$ and $B$, with respect to the parameterized family of states, as
\begin{equation}
    \langle A, B\rangle_{\rho_{\bm{\theta}}} \coloneqq \operatorname{Tr}\left[\rho_{\bm{\theta}} \left(\frac{1}{2}\{A^{\dagger},B \}\right) \right] \; , \label{eqn:semi_inner_prod}
\end{equation}
where $\{a, b\}\coloneqq ab+ba$ is the anti-commutator. The symmetric logarithmic derivative (SLD) is a Hermitian operator $L_{\bm{\theta};i}$ that is defined by the solution to the Lyapunov type equation \cite{bhatia2013matrix}
\begin{equation}
    \frac{\partial}{\partial \theta_{\alpha}}\rho_{\bm{\theta}} \eqqcolon \frac{1}{2}\{ L_{\bm{\theta};\alpha}, \rho_{\bm{\theta}} \} \; , \label{eqn:SLD}
\end{equation}
The SLD QFIM corresponding to the parameterized family of states is defined as the $p\times p$ matrix
\begin{equation}
    \left[\textsf{I}_{\operatorname{QF}}(\bm{\theta}; \{ \rho_{\bm{\theta}} \})\right]_{\alpha \beta} \equiv  \textsf{I}^{\bm{\theta}}_{\alpha \beta} \coloneqq \langle L_{\bm{\theta};\alpha}, L_{\bm{\theta};\beta} \rangle_{\rho_{\bm{\theta}}} \; . \label{eqn:QFIM}
\end{equation}
Next, assume that a quantum measurement of an observable $X$ is performed on $\rho_{\bm{\theta}}$, described by a POVM $\Pi \equiv \{\Pi_{x}\}_{x \in \mathcal{X}}$. This yields the statistics $p_{X}(x\vert \bm{\theta})=\operatorname{Tr}[\rho_{\bm{\theta}}\Pi_{x}]$ for the measurement outcomes $x \in \mathcal{X}$. We define an estimate $\bm{\hat{\theta}}:\mathcal{X}\rightarrow \Theta^{p}$ as a mapping from the set of measurement outcomes to the parameter space. We say that the pair $(\Pi,\bm{\hat{\theta}})$ is an estimator, and call it \textit{unbiased} if 
\begin{equation}
    \mathbb{E}\left[ \bm{\hat{\theta}}(X) \right]_{p(x\vert \bm{\theta})} \coloneqq \sum_{x \in \mathcal{X}}\bm{\hat{\theta}}(x)\operatorname{Tr}[\rho_{\bm{\theta}}\Pi_{x}]=\bm{\theta} \; .
\end{equation}
In general, such an estimator does not exist for all $\bm{\theta} \in \Theta^{p}$. Instead, it is customary to use a weaker condition on our estimator, namely that it is \textit{locally unbiased}. This is defined as follows: at a fixed $\bm{\theta}$, we require that the following two conditions are satisfied
\begin{equation}
    \mathbb{E}\left[ \bm{\hat{\theta}}(X) \right]_{p(x\vert \bm{\theta})}=\sum_{x \in \mathcal{X}}\bm{\hat{\theta}}(x)\operatorname{Tr}[\rho_{\bm{\theta}}\Pi_{x}]=\bm{\theta} \; , \label{eqn:loc_unbias_1}
\end{equation}
and \begin{equation}
    \frac{\partial}{\partial \theta_{\beta}}\mathbb{E}\left[ \bm{\hat{\theta}}_{\alpha}(X) \right]_{p(x\vert \bm{\theta})}=\sum_{x \in \mathcal{X}}\hat{\theta}_{\alpha}(x)\operatorname{Tr}\left[\frac{\partial}{\partial \theta_{\beta}}\rho_{\bm{\theta}}\Pi_{x}\right]=\delta_{\alpha \beta}\; . \label{eqn:loc_unbias_2}
\end{equation}
Finally, we define the mean-square error (MSE) matrix corresponding to an estimator $(\Pi, \bm{\hat{\theta}})$ as follows
\begin{equation}
    \operatorname{Var}\left[ \bm{\hat{\theta}}(X) \right]\coloneqq \mathbb{E}\left[ \left( \bm{\hat{\theta}}(X)-\bm{\theta} \right)^{T} \left( \bm{\hat{\theta}}(X)-\bm{\theta} \right) \right]_{p(x\vert \bm{\theta})} \; , \label{eqn:var_matrix}
\end{equation}
with the $(\alpha, \beta)$ entry of this matrix given by 
\begin{equation}
    \mathbb{E}\left[ \left( \hat{\theta}_{\alpha}(X)-\theta_{\alpha} \right) \left( \hat{\theta}_{\beta}(X)-\theta_{\beta} \right) \right]_{p(x\vert \bm{\theta})} \; .
\end{equation}

\subsection{Saturation of QCRB}
\label{apx:saturation}
The quantum Cram\'er-Rao bound (QCRB) provides a lower bound to the variance matrix defined in Eq.~\eqref{eqn:var_matrix} using the QFIM in Eq.~\eqref{eqn:QFIM} \cite{petz2011introduction, suzuki2020quantum}
\begin{equation}
    \operatorname{Var}\left[ \bm{\hat{\theta}} \right] \ge \left(  \textsf{I}^{\bm{\theta}} \right)^{-1} \; .
\end{equation}
The QCRB holds for any locally unbiased estimator, and is a direct consequence to applying the Schwartz inequality for the inner product defined in Eq.~\eqref{eqn:semi_inner_prod}. Here, we are interested in the saturation condition for this inequality. A necessary and sufficient condition for the saturation of the multi-parameter QCRB is given by (see e.g. \cite{hayashi2005asymptotic, liu2020quantum})
\begin{equation}
    \operatorname{Tr}[\rho_{\bm{\theta}}[L_{\bm{\theta};\alpha}, L_{\bm{\theta};\beta}]] = 0 \hspace{0.2cm} \text{for all} \hspace{0.2cm} \alpha, \beta=1, \cdots, p \; .
\end{equation}
To design the optimal measurements for the saturation of the multi-parameter QCRB, we conduct the following: (i) find an SLD $\{L_{\bm{\theta}; \alpha}\}$ that mutually commute, (ii) using the matrices $ \textsf{I}^{\bm{\theta}}$ and $\{L_{\bm{\theta}; \alpha}\}$, construct the (commuting) linear combinations
\begin{equation}
    \tilde{L}_{\bm{\theta}; \alpha} \coloneqq \sum_{\beta=1}^{p}\left(  \textsf{I}^{\bm{\theta}} \right)^{-1}_{\alpha \beta}L_{\bm{\theta}; \beta} \; ,
\end{equation}
for all $\alpha=1, \cdots, p$, and (iii) write the spectral decomposition of the mutually commuting operators
\begin{equation}
    \tilde{L}_{\bm{\theta}; \alpha}=\sum_{i}l_{\alpha i}P_{i} \; ,
\end{equation}
for $\alpha=1, \cdots, p$, where $\{P_{i}\}$ are the projectors onto the simultaneous eigenspaces of $\{\tilde{L}_{\bm{\theta};\alpha}\}$ (or equivalently for $\{L_{\bm{\theta};\alpha}\}$). Then, the QCRB is saturated if we pick the locally unbiased estimator $(\Pi, \bm{\hat{\theta}})$ to be \cite{suzuki2020quantum}
\begin{align}
    \Pi_{x} &\equiv P_{i=x} \; , \\
    \hat{\theta}_{\alpha}(x) &\equiv \theta_{\alpha}+l_{\alpha x} \; .
\end{align}
In the case of a single parameter family, this yields the locally unbiased estimator
\begin{equation}
    \hat{\theta}(x)=\theta+\left(\textsf{I}^{\theta}\right)^{-1}\frac{d}{d\theta}\log p(x\vert \theta) \; ,
\end{equation}
where $p(x\vert \theta)=\operatorname{Tr}[\rho_{\theta}\Pi_{x}]$, which explicitly depend on $\theta$. Although the optimal measurements described above saturate the QCRB, they require prior knowledge of the noise parameters, which defeats the point of implementing spectator systems. In the single-parameter case, this can be remedied.
\subsubsection{Parameter-Independent Estimator}
Nagaoka has shown that, in the single parameter case, a $\theta$-independent locally unbiased estimator exists that saturates the QCRB \cite{hayashi2005asymptotic}. This is possible only for an exponential family $\{\rho_{\theta}\}_{\theta \in \Theta}$ of parameterized states \cite{hasegawa1997exponential}:
\begin{equation}
    \rho_{\theta}=e^{\frac{1}{2}\int_{0}^{\theta}\psi(\theta^{\prime})d\theta^{\prime}\left(O-\theta_{\psi}\right)}\rho_{0}e^{\frac{1}{2}\int_{0}^{\theta}\psi(\theta^{\prime})d\theta^{\prime}\left(O-\theta_{\psi}\right)} \; , \label{eqn:exp_fam}
\end{equation}
for some $\rho_{0}$, where we have assumed for convenience that $\theta=0 \in \Theta$, $\psi(\theta)$ is some function of $\theta$,
\begin{equation}
    \theta_{\psi}=\frac{\int_{0}^{\theta}\theta^{\prime}\psi(\theta^{\prime})d\theta^{\prime}}{\int_{0}^{\theta}\psi(\theta^{\prime})d\theta^{\prime}} \; ,
\end{equation}
and $O$ is an unbiased observable of $\theta$, i.e. $\operatorname{Tr}[\rho_{\theta}O]=\theta$. As such, the SLD of this parametric family is given by $L_{\theta}=\psi(\theta)(O-\theta)$, which guarantees that the Schwartz inequality for the two vectors $L_{\theta}$ and $O-\theta$ is saturated, and hence the saturation of the QCRB \cite{petz2011introduction, hayashi2005asymptotic}. The optimal measurement POVM (as described above) is given by the (parameter-independent) eigenvectors of $O$. Therefore, we see that achieving exponential family of states, as defined in Eq.~\eqref{eqn:exp_fam}, for the output states $\psi \rightarrow \mathcal{M}_{\theta}(\psi)$ of the spectator system is generally helpful for our application. Finally, note that for non-full rank parameterized density matrices, the optimal measurements described above are not unique.

\subsubsection{Maximum Likelihood Estimator}
The maximum likelihood estimator $\hat{\theta}_{\text{MLE}}$ corresponding to the choice of POVM $\Pi\equiv \{\Pi_{x}\}_{x \in \mathcal{X}}$ is defined as
\begin{equation}
    \hat{\theta}_{\text{MLE}}(x) \coloneqq \arg \max_{\theta}p(x\vert \theta) \; ,
\end{equation}
where $p(x\vert \theta)=\operatorname{Tr}[\rho_{\theta}\Pi_{x}]$ for all $x \in \mathcal{X}$.
Although the above definition seems natural, the MLE is known to be a biased estimator for a general parametric family $\{\rho_{\theta}\}_{\theta \in \Theta}$. However, the MLE becomes unbiased, and further, saturates the classical CRB in the asymptotic limit \cite{wasserman2004all}. We recall that a necessary condition for the saturation of the QCRB is that the classical and quantum Fisher informations must coincide \cite{liu2020quantum}. Hence, the MLE is also relevant for the asymptotic saturation of the QCRB. In the context of our article, the asymptotic limit necessarily implies that the spatial dependence of the noise parameter cannot be neglected, as we are performing quantum parameter estimation on a large number of spectator qubits that must be spatially distributed within the quantum memory device. Therefore, to retain the spatial homogeneity assumption of the noise parameters used in the main text, we refrain from considering the asymptotic saturation of the QCRB. Hence, the MLE choice is inappropriate within the context of our manuscript, as it is a biased estimator in the non-asymptotic regime.

\subsection{Partial QFIM}
Now we consider the bipartition of the parameter space $\Theta^{p}$ as $\bm{\theta}=(\bm{\theta}^{\text{I}}, \bm{\theta}^{\text{N}})$, with the number of parameters in each partition is given by $p_{\text{I}}$ and $p_{\text{N}}=p-p_{\text{I}}$, respectively. Here, the subscripts ``I'' and ``N'' stand for ``interest'' and ``nuisance'', respectively. We can then write the $p\times p$ SLD QFIM in a block form
\begin{equation}
     \textsf{I}^{\bm{\theta}}=
    \begin{pmatrix}
         \textsf{I}^{\bm{\theta}}_{\text{I},\text{I}} &  \textsf{I}^{\bm{\theta}}_{\text{I},\text{N}} \\
         \textsf{I}^{\bm{\theta}}_{\text{N},\text{I}} &  \textsf{I}^{\bm{\theta}}_{\text{N},\text{N}} 
    \end{pmatrix} \; ,
\end{equation}
where the upper block diagonal matrix $ \textsf{I}^{\bm{\theta}}_{\text{I},\text{I}}$ is $p_{\text{I}}\times p_{\text{I}}$ and the lower block diagonal matrix $ \textsf{I}^{\bm{\theta}}_{\text{N},\text{N}}$ is of size $p_{\text{N}}\times p_{\text{N}}$. We also write the inverse of the SLD QFIM in a similar block form
\begin{equation}
    \left( \textsf{I}^{\bm{\theta}}\right)^{-1}=
    \begin{pmatrix}
         \textsf{I}^{\bm{\theta};\text{I},\text{I}} &  \textsf{I}^{\bm{\theta}; \text{I},\text{N}} \\
         \textsf{I}^{\bm{\theta};\text{N},\text{I}} &  \textsf{I}^{\bm{\theta};\text{N},\text{N}} 
    \end{pmatrix} \; .
\end{equation}
Using these block forms, the partial SLD QFIM is defined as
\begin{equation}
     \textsf{I}^{\bm{\theta}}_{\text{I}\vert \text{N}} \coloneqq \left(  \textsf{I}^{\bm{\theta};\text{I},\text{I}} \right)^{-1} =  \textsf{I}^{\bm{\theta}}_{\text{I},\text{I}}- \textsf{I}^{\bm{\theta}}_{\text{I},\text{N}}\left(  \textsf{I}^{\bm{\theta}}_{\text{N},\text{N}} \right)^{-1} \textsf{I}^{\bm{\theta}}_{\text{N},\text{I}} \; .
\end{equation}
Let us further define the $p_{I}\times p_{I}$ MSE matrix for the parameters of interest $\bm{\theta}_{\text{I}}$ as 
\begin{equation}
    \operatorname{Var}\left[ \bm{\hat{\theta}}_{\text{I}}(X) \right]\coloneqq 
    \mathbb{E}\left[ \left( \hat{\theta}_{\alpha}(X)-\theta_{\alpha} \right) \left( \hat{\theta}_{\beta}(X)-\theta_{\beta} \right) \right]_{p(x\vert \bm{\theta})} \; ,
\end{equation}
for $\alpha$, $\beta=1, \cdots, p_{I}$. The pair $(\Pi, \bm{\hat{\theta}}_{\text{I}})$ is said to be a locally unbiased estimator for $\bm{\theta}_{I}$ at $\bm{\theta}$ when it satisfies the following two conditions \cite{suzuki2020quantum} (analogous to Eq.~\eqref{eqn:loc_unbias_1} and \eqref{eqn:loc_unbias_2})
\begin{equation}
    \mathbb{E}\left[ \bm{\hat{\theta}}_{\text{I}}(X) \right]_{p(x\vert \bm{\theta})}=\sum_{x \in \mathcal{X}}\bm{\hat{\theta}}_{\text{I}}(x)\operatorname{Tr}[\rho_{\bm{\theta}}\Pi_{x}]=\bm{\theta}_{\text{I}} \; , \label{eqn:loc_unbias_3}
\end{equation}
and \begin{equation}
    \frac{\partial}{\partial \theta_{\beta}}\mathbb{E}\left[ \bm{\hat{\theta}}_{\alpha}(X) \right]_{p(x\vert \bm{\theta})}=\sum_{x \in \mathcal{X}}\hat{\theta}_{\alpha}(x)\operatorname{Tr}\left[\frac{\partial}{\partial \theta_{\beta}}\rho_{\bm{\theta}}\Pi_{x}\right]=\delta_{\alpha \beta} \; , \label{eqn:loc_unbias_4}
\end{equation}
where $\alpha=1, \cdots, p_{I}$ and $\beta=1, \cdots, p$. For locally unbiased estimators, the following QCRB holds in the presence of nuisance parameters
\begin{equation}
    \operatorname{Var}\left[ \bm{\hat{\theta}}_{\text{I}}(X) \right] \ge \left[  \textsf{I}^{\bm{\theta}}_{\text{I}|\text{N}} \right]^{-1} \; ,
\end{equation}
which is a modification of the standard QCRB when nuisance parameters are present.

\section{Choi Matrix Methods for Entanglement Fidelity of Composite Parameterized Channels} \label{apx:Choi_techniques}
In what follows, we present a useful lemma for the entanglement fidelity of composite parameterized channels and then dedicate the rest of this appendix to demonstrating its wide range of applicability in the context of the main text.
\begin{lemma} \label{le:fid_Choi}
    The entanglement fidelity of the composite channel $\mathcal{R}^{B\rightarrow A} \circ \mathcal{N}^{A \rightarrow B}$ is given by the individual Choi states $\Phi^{\mathcal{N}}_{AB}$ and $\Phi^{\mathcal{R}}_{BA}$, as follows
    \begin{equation}
        F_{e}(\mathcal{R}\circ \mathcal{N})=\frac{d_{B}}{d_{A}}\operatorname{Tr}\left[ \Phi^{\mathcal{N}}\left( \Phi^{\mathcal{R}} \right)^{T} \right] \; ,
    \end{equation}
    where $T$ indicates matrix transposition.
\end{lemma}
\begin{proof}
    By definition, the entanglement fidelity of the composite channel can be written in terms of its Choi matrix, as follows
    \begin{align}
    F_{e}\left( \mathcal{R}^{B \rightarrow A^{\prime}} \circ \mathcal{N}^{A \rightarrow B} \nonumber \right) &= \operatorname{Tr}_{AA^{\prime}}\left[ \Phi_{AA^{\prime}} \Phi^{\mathcal{R}\circ \mathcal{N}}_{AA^{\prime}} \right] \\ & =\frac{1}{d_{A}^{2}}\operatorname{Tr}_{AA^{\prime}}\left[ \Gamma_{AA^{\prime}} \Gamma^{\mathcal{R}\circ \mathcal{N}}_{AA^{\prime}} \right] \; . \label{eqn:ent_fid_choi}
\end{align}
where $\mathcal{H}^{A}$ and $\mathcal{H}^{A^{\prime}}$ are isomorphic Hilbert spaces. One can easily verify that the Choi matrix of the composite channel can be written in terms of the Choi matrices of the individual channels, as follows \cite{khatri2020principles}
\begin{equation}
    \Gamma_{AC}^{\mathcal{R}^{B \rightarrow C} \circ \mathcal{N}^{A \rightarrow B}} =\operatorname{Tr}_{B}\left[ \mathbb{T}_{B}\left( \Gamma_{AB}^{\mathcal{N}} \right) \Gamma_{BC}^{\mathcal{R}}\right] \; ,
\end{equation}
where $\mathbb{T}_{B}$ is the partial transpose defined with respect to the same basis as the maximally entangled state $|\Gamma \rangle$. By substituting this form into the entanglement fidelity formula, we arrive at
\begin{align}
    F_{e} &=\frac{1}{d_{A}^{2}}\operatorname{Tr}_{AA^{\prime}}\left[ \Gamma_{AA^{\prime}} \operatorname{Tr}_{B}\left[ \mathbb{T}_{B}\left( \Gamma_{AB}^{\mathcal{N}} \right) \Gamma_{BA^{\prime}}^{\mathcal{R}}\right] \right] \\
    &=\frac{1}{d_{A}^{2}}\operatorname{Tr}_{A^{\prime}B}\left[ \left(\operatorname{Tr}_{A}\left[\Gamma_{AA^{\prime}} \mathbb{T}_{B}\left( \Gamma_{AB}^{\mathcal{N}} \right) \right]\right) \Gamma_{BA^{\prime}}^{\mathcal{R}}\right] \; .
\end{align}
Next, we make standard simplifications for any $\Lambda_{AB}$:
    \begin{align}
        & \operatorname{Tr}_{A}\left[\Gamma_{AA^{\prime}} \Lambda_{AB} \right] \\ &=\sum_{i}\langle i|_{A}\Gamma_{AA^{\prime}}\Lambda_{AB}|i\rangle_{A} \\
        &= \sum_{ij}|i\rangle_{A^{\prime}} \! \langle j|_{A^{\prime}}\langle j|_{A}\Lambda_{AB}|i\rangle_{A} \\
        &= \sum_{ij}|i\rangle_{A^{\prime}} \! \langle j|_{A^{\prime}}\langle i|_{A}\mathbb{T}_{A}\left(\Lambda_{AB}\right)|j\rangle_{A} \\
        &= \sum_{i}|i\rangle_{A^{\prime}} \! \langle i|_{A}\mathbb{T}_{A}\left(\Lambda_{AB}\right)\sum_{j}|j\rangle_{A}\! \langle j|_{A^{\prime}} \\
        &=\mathbb{T}_{A^{\prime}}\left(\Lambda_{A^{\prime}B}\right) \; .
    \end{align}
    This yields for $\Lambda_{AB}\equiv \mathbb{T}_{B}\left( \Gamma_{AB}^{\mathcal{N}}\right)$ the following
    \begin{align}
        \operatorname{Tr}_{A}\left[\Gamma_{AA^{\prime}} \mathbb{T}_{B}\left( \Gamma_{AB}^{\mathcal{N}} \right)\right] &=\mathbb{T}_{A^{\prime}}\left(  \mathbb{T}_{B}\left( \Gamma_{A^{\prime}B}^{\mathcal{N}} \right) \right) \\ &=\left( \Gamma_{A^{\prime}B}^{\mathcal{N}} \right)^{T} \; .
    \end{align}
Substituting back into the entanglement fidelity formula completes the proof.
\end{proof}
Therefore, according to this lemma, the entanglement fidelity of the memory dynamics is generally written as follows for any recovery map
\begin{align}
    F_{e}\left( \mathcal{R}\circ \mathcal{N}_{\bm{\theta}} \right) &=\frac{1}{d_{A}^{2}}\operatorname{Tr}_{AB}\left[ \left( \Gamma_{AB}^{\mathcal{N}_{\bm{\theta}}} \right)^{T}\Gamma_{BA}^{\mathcal{R}}\right] \label{eqn:ent_fid_simplified} \\ 
    &=\frac{1}{d_{A}^{2}}\operatorname{Tr}_{AB}\left[ \Gamma_{AB}^{\mathcal{N}_{\bm{\theta}}} \left(\Gamma_{BA}^{\mathcal{R}}\right)^{T}\right] \; .
\end{align}
In particular, this implies that we can search for a recovery map $\mathcal{R}_{\bm{\phi}}$ that is parameterized by some number of parameters $\bm{\phi}$ and maximize over them (for a fixed $\bm{\theta}$) to arrive at an optimal choice $\bm{\phi}(\bm{\theta})$, see e.g. \cite{fletcher2007channel} in terms of the natural representation of quantum channels (which is related, but not the same as, the Choi representation adopted in this article).


\subsection{Zeroth Derivative: H\"older Type Upper Bounds} \label{sub:Holder_tech}
A useful upper bound on the entanglement fidelity in Eq.~\eqref{eqn:opt_recovery} of the quantum memory dynamics can be given in terms of the Choi state of the recovery map, by applying the H\"older inequality \cite{khatri2020principles}, as follows
\begin{align}
    F_{e}\left( \mathcal{R}\circ \mathcal{N}_{\bm{\theta}} \right) 
    &=\frac{d_{B}}{d_{A}}\left \vert \operatorname{Tr}_{AB}\left[ \Phi_{AB}^{\mathcal{N}_{\bm{\theta}}} \left(\Phi_{BA}^{\mathcal{R}}\right)^{T}\right] \right \vert \\
    &\leq \frac{d_{B}}{d_{A}}\operatorname{Tr}_{AB}\left \vert \Phi_{AB}^{\mathcal{N}_{\bm{\theta}}} \left(\Phi_{BA}^{\mathcal{R}}\right)^{T}\right \vert \\
    & \leq \frac{d_{B}}{d_{A}} \left \Vert \Phi_{AB}^{\mathcal{N}_{\bm{\theta}}} \right \Vert_{\alpha} \times \left \Vert \left(\Phi_{BA}^{\mathcal{R}}\right)^{T} \right \Vert_{\beta} \\ 
    & =\frac{d_{B}}{d_{A}} \left \Vert \Phi_{AB}^{\mathcal{N}_{\bm{\theta}}} \right \Vert_{\alpha} \times \left \Vert \Phi_{BA}^{\mathcal{R}}\right \Vert_{\beta} \; ,
\end{align}
for $\alpha \in [1, \infty)$ and its H\"older dual $\beta$ defined via $1/\alpha+1/\beta=1$. Here, $\Vert X \Vert_{\alpha}\coloneqq \left(\operatorname{Tr}[\vert X^{\alpha} \vert]\right)^{1/\alpha}$ defines the Schatten norms for $\alpha \in [1, \infty)$.

\subsection{First Derivative: Bound on Robustness of a Recovery Map}\label{subsec:robustness_bound}
Consider, for a given recovery map $\mathcal{R}$, the partial derivatives of the entanglement of fidelity in Eq.~\eqref{eqn:ent_fid_simplified} with respect to the components of the noise parameter vector $\bm{\theta}$. Using the inner product definition in Eq.~\eqref{eqn:semi_inner_prod} and the definition of SLD operator for the Choi state of $\mathcal{N}_{\bm{\theta}}$, as given in Eq.~\eqref{eqn:SLD}, we have for the partial derivatives 
\begin{align}
     & \frac{\partial}{\partial \theta_{\alpha}}F_{e}\left( \mathcal{R}\circ \mathcal{N}_{\bm{\theta}} \right) \nonumber \\
     & = \frac{1}{d_{A}^{2}}\operatorname{Tr}_{AB}\left[\frac{\partial}{\partial \theta_{\alpha}} \Gamma_{AB}^{\mathcal{N}_{\bm{\theta}}}\left( \Gamma_{BA}^{\mathcal{R}}\right)^{T}\right] \\ & = \frac{d_{B}}{d_{A}}\operatorname{Tr}_{AB}\left[ \frac{1}{2}\{ \Phi^{\mathcal{N}_{\bm{\theta}}}_{AB}, L^{\bm{\theta};\alpha}_{AB} \}\left( 
     \Phi_{BA}^{\mathcal{R}} \right)^{T}\right] \\ 
     & = \frac{d_{B}}{d_{A}}\operatorname{Tr}_{AB}\left[\Phi_{AB}^{\mathcal{N_{\bm{\theta}}}} \frac{1}{2}\{ \left( 
     \Phi_{BA}^{\mathcal{R}} \right)^{T}, L^{\bm{\theta};\alpha}_{AB} \}\right] \\
     & \equiv \frac{d_{B}}{d_{A}}\left \langle \left(\Phi^{\mathcal{R}}_{BA}\right)^{T}, L^{\bm{\theta};\alpha}_{AB} \right \rangle_{ \Phi_{AB}^{\mathcal{N}_{\bm{\theta}}}} \; ,
\end{align}
where $d_{B}=2^{n}$ and $d_{A}=2^{k}$ for a generic $[n,k]$ QEC code, therefore $d_{B}/d_{A}=2^{n-k}$. Using the fact that the robustness of $\mathcal{R}$ (which is the derivative of the entanglement fidelity with respect to the parameters pertaining to $\mathcal{N}_{\bm{\theta}}$) is written in terms of the inner product in Eq.~\eqref{eqn:semi_inner_prod}, we can use the Cauchy-Schwartz inequality to arrive at an upper bound, as follows
\begin{align}
    \left \vert \left \langle \left(\Phi^{\mathcal{R}}_{BA}\right)^{T}, L^{\bm{\theta};\alpha}_{AB} \right \rangle_{ \Phi_{AB}^{\mathcal{N}_{\bm{\theta}}}} \right \vert^{2} & \leq \left \langle \left(\Phi^{\mathcal{R}}_{BA}\right)^{T},\left(\Phi^{\mathcal{R}}_{BA}\right)^{T}  \right \rangle_{ \Phi_{AB}^{\mathcal{N}_{\bm{\theta}}}} \times \nonumber \\ & \times  \left \langle L^{\bm{\theta};\alpha}_{AB}, L^{\bm{\theta};\alpha}_{AB} \right \rangle_{ \Phi_{AB}^{\mathcal{N}_{\bm{\theta}}}} \\
    &=\left \langle \left(\Phi^{\mathcal{R}}_{BA}\right)^{T},\left(\Phi^{\mathcal{R}}_{BA}\right)^{T}  \right \rangle_{ \Phi_{AB}^{\mathcal{N}_{\bm{\theta}}}} \times \nonumber \\ &\times \left[\textsf{I}_{\operatorname{QF}}(\bm{\theta}; \{ \Phi^{\mathcal{N}_{\bm{\theta}}} \})\right]_{\alpha \alpha} \; ,
\end{align}
which yields the upper bound
\begin{align}
     \left \vert \frac{\partial}{\partial \theta_{\alpha}}F_{e}\left( \mathcal{R}\circ \mathcal{N}_{\bm{\theta}} \right) \right \vert &  \leq \frac{d_{B}}{d_{A}}\sqrt{\left \langle \left(\Phi^{\mathcal{R}}_{BA}\right)^{T},\left(\Phi^{\mathcal{R}}_{BA}\right)^{T}  \right \rangle_{ \Phi_{AB}^{\mathcal{N}_{\bm{\theta}}}} } \times \nonumber \\ &\times \sqrt{\left[\textsf{I}_{\operatorname{QF}}(\bm{\theta}; \{ \Phi^{\mathcal{N}_{\bm{\theta}}} \})\right]_{\alpha \alpha}} \; .
\end{align}

\subsection{Proof of Theorem~\ref{th:finite_var}: Bounds on Spectator-Based Recovery For Finite Estimation Errors}
\label{axp:finite_est_error}
Schatten norms $\Vert X \Vert_{\alpha}\coloneqq \left(\operatorname{Tr}[\vert X^{\alpha} \vert]\right)^{1/\alpha}$, where $\alpha \in [1, \infty)$, are often used to bound trace quantities, via the well-known H\"older inequality \cite{khatri2020principles}
\begin{equation}
    \Vert X Z \Vert_{1} \leq \Vert X \Vert_{\alpha} \times \Vert Z \Vert_{\beta} \; ,
\end{equation}
where $1/\alpha+1/\beta=1$ (a pair $(\alpha, \beta)$ satisfying this equality is called an H\"older pair). It is easy to show that the following statement also applies \cite{khatri2020principles} 
\begin{equation}
    \vert \operatorname{Tr}[X Z] \vert \leq \Vert X \Vert_{\alpha} \times \Vert Z \Vert_{\beta} \; . \label{eqn:modified_Holder}
\end{equation}
For $\alpha \in [0,1)$, the Schatten norm $\Vert \cdot \Vert_{\alpha}$ is no longer a norm (e.g. it does not satisfy the triangle inequality). However, if $Z>0$ (along with $0 \leq \alpha<1$), then the above inequality is reversed \cite{muller2013quantum} (the H\"older dual $\beta$ becomes negative)
\begin{equation}
    \vert \operatorname{Tr}[X Z] \vert \ge \Vert X \Vert_{\alpha} \times \Vert Z \Vert_{\beta} \; . \label{eqn:muller_th}
\end{equation}
We can use this inequality to find a lower bound on the difference between the entanglement fidelities of any two recovery maps $\mathcal{R}^{B \rightarrow A}$, $\tilde{\mathcal{R}}^{B \rightarrow A}$ for the parameterized noise channel $\mathcal{N}^{A \rightarrow B}_{\bm{\theta}}$, as follows
\begin{align}
    & \vert F_{e}\left( \mathcal{R} \circ \mathcal{N}_{\bm{\theta}} \right)-F_{e}( \tilde{\mathcal{R}}\circ \mathcal{N}_{\bm{\theta}}) \vert \nonumber \\
    & =\left \vert \operatorname{Tr}_{AA^{\prime}}\left[ \Phi_{AA^{\prime}}\left(\textsf{id}^{A}\otimes \left( \mathcal{R}^{B \rightarrow A^{\prime}}-\tilde{\mathcal{R}}^{B \rightarrow A^{\prime}} \right) \right)\left(\Phi^{\mathcal{N}_{\theta}}_{AB} \right) \right] \right \vert \\
    & =\left \vert \operatorname{Tr}_{AB}\left[ \left( \Phi_{AB}^{\mathcal{R} ^{\dagger}}-\Phi_{AB}^{\tilde{\mathcal{R}} ^{\dagger}} \right) \left(\Phi^{\mathcal{N}_{\bm{\theta}}}_{AB} \right) \right] \right \vert \\
    & \ge \left \Vert \Phi_{AB}^{\mathcal{R}^{\dagger}}-\Phi_{AB}^{\tilde{\mathcal{R}}^{\dagger}} \right \Vert_{\alpha} \times  \left \Vert \Phi^{\mathcal{N}_{\bm{\theta}}}_{AB} \right \Vert_{\beta} \; . \label{eqn:finite_deriv_1}
\end{align}
Next, we use a theorem relating the Choi matrix of a quantum channel $\mathcal{Q}^{A}$ to its adjoint \cite{johnston2011quantum}
\begin{equation}
	\Gamma^{\mathcal{Q}^{\dagger}}_{AA^{\prime}}=\Gamma^{\mathbb{T}}_{AA^{\prime}}\left(\Gamma^{\mathcal{Q}}_{AA^{\prime}}\right)^{T}\Gamma^{\mathbb{T}}_{AA^{\prime}} \; , \label{eqn:choi_adj}
\end{equation}
where $\mathbb{T}$ is the partial transpose channel, and its Choi matrix yields a SWAP unitary \cite{johnston2011quantum}. Even though $\Vert \cdot \Vert_{\alpha}$ is not a norm for $\alpha \in [0,1)$, its definition is still invariant with respect to a unitary transformation and transposition \cite{khatri2020principles}. This yields
\begin{equation}
    \left \Vert \Phi_{AB}^{\mathcal{R} ^{\dagger}}-\Phi_{AB}^{\tilde{\mathcal{R}}^\dagger} \right \Vert_{\alpha}=\left \Vert \Phi_{AB}^{\mathcal{R}} -\Phi_{AB}^{\tilde{\mathcal{R}}} \right \Vert_{\alpha} \; .
\end{equation}
Substituting back into Eq.~\eqref{eqn:finite_deriv_1} yields the inequality in Theorem~\ref{th:finite_var}.


\section{Sufficient Condition For a Negligible Remainder Term in Theorem~\ref{th:spec_QFI}}
\label{apx:remainder}
To quantify the ``smallness'' of $\hat{\theta}-\theta$, for the remainder term in Eq.~\eqref{eqn:thm_2_eqn_1} of Theorem~\ref{th:spec_QFI} to be negligible (given the noise channel $\mathcal{N}_{\theta}$), we first use Eq.~\eqref{eqn:approx_Choi_int_fid} to arrive at a useful bound, as follows:
\begin{align}
	& \left \vert \frac{1}{3!}\partial^{3}_{\nu}F_{e}(\mathcal{R}_{\theta + \nu_{0}}\circ \mathcal{N}_{\theta})(\hat{\theta}-\theta)^{3} \right \vert \nonumber \\ &= \frac{d_{B}}{3!d_{A}}\left \vert \operatorname{Tr}_{AB}\left[\left( \Phi_{AB}^{\mathcal{N}_{\theta}} \right)^{T} \partial_{\theta}^{3}\Phi^{\mathcal{R}_{\theta}}_{BA} \right]\right \vert \vert \hat{\theta}-\theta \vert^{3} \\
	& \leq \frac{d_{B}}{3!d_{A}} \operatorname{Tr}_{AB}\left[ \left \vert \left( \Phi_{AB}^{\mathcal{N}_{\theta}} \right)^{T} \partial_{\theta}^{3}\Phi^{\mathcal{R}_{\theta}}_{BA}\right \vert \right] \vert \hat{\theta}-\theta \vert^{3} \\
	& \leq \frac{d_{B}}{3!d_{A}} \left \Vert \left(  \Phi_{AB}^{\mathcal{N}_{\theta}} \right)^{T} \right \Vert_{1} \left \Vert \partial_{\theta}^{3}\Phi^{\mathcal{R}_{\theta}}_{BA}\right \Vert_{\infty} \vert \hat{\theta}-\theta \vert^{3} \\
	& = \frac{d_{B}}{3!d_{A}} \left \Vert  \Phi_{AB}^{\mathcal{N}_{\theta}} \right \Vert_{1} \left \Vert \partial_{\theta}^{3}\Phi^{\mathcal{R}_{\theta}}_{BA}\right \Vert_{\infty} \vert \hat{\theta}-\theta \vert^{3} \\
	& = \frac{d_{B}}{3!d_{A}} \left \Vert \partial_{\theta}^{3}\Phi^{\mathcal{R}_{\theta}}_{BA}\right \Vert_{\infty} \vert \hat{\theta}-\theta \vert^{3} \; ,
\end{align} 
where we have used the H\"older inequality \cite{khatri2020principles}, the invariance of the trace norm $\Vert \cdot \Vert_{1}$ under transposition, and the unit trace of the Choi state $\Phi^{\mathcal{N}_{\theta}}_{AB}$, in the third, fourth, and fifth lines, respectively. Then, the above bound implies that the remainder term in Eq.~\ref{eqn:thm_2_eqn_1} of Theorem~\ref{th:spec_QFI} is negligible if the following sufficient condition holds
\begin{equation}
	\frac{d_{B}}{3!d_{A}} \left \Vert \partial_{\theta}^{3}\Phi^{\mathcal{R}_{\theta}}_{BA}\right \Vert_{\infty} \mathbb{E}\left[\vert \hat{\theta}-\theta \vert^{3}\right]_{p(x\vert \theta)} << \frac{g(\theta)}{I_{\text{QF}}(\mathcal{M}_{\theta})} \; .
\end{equation}
We can rewrite this condition as
\begin{align}
\mathbb{E}\left[\vert \hat{\theta}-\theta \vert^{3}\right]_{p(x\vert \theta)}\textsf{I}_{\text{QF}}(\mathcal{M}_{\theta}(\psi)) &<< \frac{6d_{A}g(\theta)}{d_{B}\left \Vert \partial_{\theta}^{3}\Phi^{\mathcal{R}_{\theta}}_{BA}\right \Vert_{\infty}} \nonumber \\ &\equiv g^{\prime}(\theta) \; , \label{eqn:negligible_remainder}
\end{align}
where the right-hand-side $g^{\prime}(\theta)$ is fully determined by the parameterized noise channel $\mathcal{N}^{A \rightarrow B}_{\theta}$ (since $\mathcal{R}^{B \rightarrow A}_{\theta}$ is found from Eq.~\eqref{eqn:opt_recovery}, given $\mathcal{N}^{A \rightarrow B}_{\theta}$), similar to $g(\theta)$ in Theorem~\ref{th:spec_QFI}. The bound in Eq.~\eqref{eqn:negligible_remainder} could be understood by saying that, although $\vert \hat{\theta}-\theta \vert$ cannot be arbitrarily small (due to Eq.~\eqref{eqn:QCRB}), it should not be too large so that the small error expansion that is important for the proof of Theorem~\ref{th:spec_QFI} will hold.

\section{Entanglement Fidelity For The [4,1] Code}
\label{apx:[4,1]_ent_fid}
The entanglement fidelity for single-qubit $(k=1)$ $[n, k]$ codes is given by \cite{zhan2013entanglement}
\begin{equation}
    F_{e}(\mathcal{R}_{\theta}\circ \mathcal{N}_{\theta})=\frac{1}{4}\mathrm{Tr}[G] \; ,
\end{equation}
where the matrix elements $G_{\sigma \sigma^{\prime}}=\mathrm{Tr}\left[ D_{\hat{\theta}, \sigma} \mathcal{N}_{\theta}[\sigma^{\prime}_{L}] \right]$, with $D_{\hat{\theta}, \sigma}\equiv 2\sum_{i}R_{\hat{\theta}, \sigma}^{(4)}\sigma_{L}R_{\hat{\theta}, \sigma}^{(4)\dagger}$, describes the effective dynamics of the Bloch coefficients for the encoded single qubit \cite{rahn2002exact}, i.e. if $\rho_{i}=\frac{1}{2}\sum_{\sigma}u_{\sigma}\sigma$ and $\rho_{f}=\frac{1}{2}\sum_{\sigma}v_{\sigma}\sigma$, then $\Vec{v}=G\Vec{u}$.

In \cite{zhan2013entanglement}, the authors derived an analytical formula for the entanglement fidelity $F_{e}(\mathcal{R}_{\theta}\circ \mathcal{N}_{\theta})$ as
\begin{align}
    F_{e}= \frac{1}{4} \left[1 +\sqrt{2}\text{Re}[\alpha]\tau+8\tau^{2}+(\sqrt{2}\text{Re}[\beta]-8)\tau^{3}+\tau^{4} \right] \; , \label{eqn:AD_ent_fid}
\end{align}
where $\tau=1-\theta$ and $\alpha$, $\beta$ (with $\vert \alpha \vert^{2}+\vert \beta \vert^{2}=1$) are the complex parameters that the recovery channel depends on. The optimum recovery channel $\mathcal{R}_{\theta}(\alpha, \beta)=\mathcal{R}(\alpha(\theta), \beta(\theta))$ in \cite{fletcher2008channel} is the one that maximizes the entanglement fidelity with respect to $\alpha$, $\beta$ for the given value of the noise parameter $\theta$.

To find the dependence of $\alpha$ and $\beta$ on the noise parameter $\theta$ for the optimum recovery, we first rewrite Eq.~\eqref{eqn:AD_ent_fid} using $\alpha = \vert \alpha \vert e^{i\psi}$ and $\beta = \vert \beta \vert e^{i\phi}=\sqrt{1-\vert \alpha \vert^{2}}e^{i\phi}$, which yields
\begin{align}
    F_{e}(\vert \alpha \vert, \psi, \phi ; \theta)&=\frac{1}{4}+\frac{\sqrt{2}}{4}\vert \alpha \vert \tau \cos{\psi} +2\tau^{2} \nonumber \\ &+(\sqrt{2(1-\vert \alpha \vert^{2})}\cos{\phi}-8)\frac{\tau^{3}}{4}+\frac{\tau^{4}}{4} \; .
\end{align}
Then we take the partial derivatives of this function with respect to the independent parameters $\vert \alpha \vert \in [0, 1]$ and $\psi, \phi \in [0, 2\pi)$, to arrive at
\begin{equation}
    \vert \alpha_{\text{opt}}(\theta) \vert=\frac{1}{\sqrt{1+\tau^{4}}} \hspace{0.2cm} \text{and} \hspace{0.2cm} (\psi_{\text{opt}}, \phi_{\text{opt}})=\{(0, 0), (\pi, \pi)\} \; .
\end{equation}
By simple substitution, we can check that $(\psi, \phi)=(0, 0)$ is the pair that maximizes the entanglement fidelity function.

Now let us find an analytical formula for the incomplete knowledge scenario $ F_{e}(\mathcal{R}_{\hat{\theta}}\circ \mathcal{N}_{\theta})$. Note that in Eq.~\eqref{eqn:AD_ent_fid}, the dependence of the recovery $\mathcal{R}_{\hat{\theta}}(\alpha, \beta)=\mathcal{R}(\alpha(\hat{\theta}), \beta(\hat{\theta}))$ on the estimated noise parameter $\hat{\theta}$ enters only through $\alpha$ and $\beta$ \cite{fletcher2008channel}. Therefore, we can use the optimum values $\vert \alpha_{\text{opt}}(\hat{\theta}) \vert = 1/\sqrt{1+\tau(\hat{\theta})}$ and $\psi_{\text{opt}}=\phi_{\text{opt}}=0$ and plug it back into Eq.~\eqref{eqn:AD_ent_fid}, which yields
\begin{align}
    F_{e}(&\vert \alpha_{\text{opt}}(\hat{\theta}) \vert  ; \theta)=\frac{1}{4}+\frac{\sqrt{2}}{4}\vert \alpha_{\text{opt}}(\hat{\theta}) \vert \tau(\theta) +2\tau^{2}(\theta) \nonumber \\ &+\left(\sqrt{2(1-\vert \alpha_{\text{opt}}(\hat{\theta}) \vert^{2})}-8\right) \frac{\tau^{3}(\theta)}{4}+\frac{\tau^{4}(\theta)}{4} \; .
\end{align}
This yields for arbitrary finite differences $\theta-\hat{\theta}$ and any estimate $\hat{\theta}$ the following exact formula
\begin{align}
    & F_{e}(\vert \alpha_{\text{opt}}(\theta) \vert ; \theta)-F_{e}(\vert \alpha_{\text{opt}}(\hat{\theta}) \vert  ; \theta) \nonumber \\ &= \frac{\tau(\theta)}{2\sqrt{2}}\left(\frac{1}{\sqrt{1+\tau^{4}(\theta)}}-\frac{1}{\sqrt{1+\tau^{4}(\hat{\theta})}} \right) \nonumber \\
    & \hspace{0.2cm} + \frac{\tau^{3}(\theta)}{2\sqrt{2}}\left(\frac{\tau^{2}(\theta)}{\sqrt{1+\tau^{4}(\theta)}}-\frac{\tau^{2}(\hat{\theta})}{\sqrt{1+\tau^{4}(\hat{\theta})}} \right) \; .
\end{align}
When adaptation is implemented, the following derivative is relevant
\begin{equation}
    \left . \frac{d^{2}F_{e}(\vert \alpha_{\text{opt}}(\theta+\nu) \vert  ; \theta)}{d\nu^{2}} \right \vert_{\nu=0}=  -\frac{\tau^{3}(\theta)}{\sqrt{2}(1+\tau^{4}(\theta))^{3/2}} \; .
\end{equation}
On the other hand, if no adaptation is implemented, the relevant derivative becomes
\begin{align}
    \left . \frac{d^{2}F_{e}(\vert \alpha_{\text{opt}}(\eta+\nu) \vert  ; \theta)}{d\nu^{2}} \right \vert_{\nu=0}=  -\frac{\tau^{3}(\eta)q\left(\tau(\eta), \frac{\tau(\theta)}{\tau(\eta)}\right)}{\sqrt{2}(1+\tau^{4}(\eta))^{3/2}} \; ,
\end{align}
where
\begin{equation}
    q(x,y) \equiv \frac{\left(\frac{5}{2}y^{3}-\frac{3}{2}y\right)x^{4}+\left(\frac{3}{2}y-\frac{1}{2}\right)}{x^{4}+1} \; ,
\end{equation}
which satisfies $q(x,1)=1$ in the adaptive regime $\eta=\theta$. Therefore, further algebraic simplifications yield
\begin{align}
    & F_{e}(\vert \alpha_{\text{opt}}(\theta) \vert ; \theta)-F_{e}(\vert \alpha_{\text{opt}}(\hat{\theta}) \vert  ; \theta) \nonumber \\ &= \frac{\tau^{3}(\theta)}{\sqrt{2}(1+\tau^{4}(\theta))^{3/2}}(\theta-\hat{\theta})^{2} - R(\hat{\theta}-\theta) \; ,
\end{align}
where $R(\hat{\theta}-\theta)\equiv \frac{1}{3!}\partial_{\nu}^{3}F_{e}(\vert \alpha_{\text{opt}}(\theta+\nu_{0}) \vert ; \theta)(\hat{\theta}-\theta)^{3}$ is the Lagrange form of the Taylor series expansion remainder of $F_{e}(\vert \alpha_{\text{opt}}(\theta+\nu) \vert ; \theta)$ with respect to $\nu$, where $\nu_{0} \in [0, \hat{\theta}-\theta]$ is a constant. Finally, taking the expectation of both sides with respect to the spectator system's probability distribution function $p_{X}(x \vert \theta)$ (where $x$ is the measurement outcome of the spectator observable $X=\sum_{x\in \mathcal{X}}x\Pi_{x}$) yields
\begin{align}
    & \mathbb{E}\left[ F_{e}(\vert \alpha_{\text{opt}}(\theta) \vert ; \theta)-F_{e}(\vert \alpha_{\text{opt}}(\hat{\theta}) \vert  ; \theta) \right] \nonumber \\ & \hspace{0.2cm}=g(\theta)\operatorname{Var}(\hat{\theta})-\mathbb{E}\left[R(\hat{\theta}-\theta)\right] \; , \label{eqn:[4,1]fund}
\end{align}
where 
\begin{equation}
    g(\theta)= \frac{(1-\theta)^{3}}{\sqrt{2}(1+(1-\theta)^{4})^{3/2}} \; , 
\end{equation}
and we have used the fact that $\hat{\theta}$ is an unbiased estimate of $\theta$. 

\section{\mathinhead{\chi}{chi}-Matrix Representation of Quantum Channels}
\label{sec:chi_matrix}
Besides the well-known Kraus and Stienspring representations of a CPTP map, a lesser-known representation, called the $\chi$-matrix representation \cite{chuang1997prescription}, is also useful in practice. This is most commonly used in quantum state tomography \cite{nielsen2002quantum} and is extended to quantum process tomography \cite{mohseni2008quantum} where state tomography of the Choi state of a quantum channel is conducted. This is to be contrasted with other approaches in measuring noise, such as randomized benchmarking \cite{knill2008randomized} and QEC itself \cite{combes2014situ}. Interestingly, the $\chi$-matrix representation can be well motivated in the context of QEC by noting that we can rewrite the ``error operators'' $\{ Q_{i} \}$ of any noisy map $\mathcal{Q}(\cdot)=\sum_{i}Q_{i}(\cdot)Q_{i}^{\dagger}$ in terms of a pre-selected ``error basis'' $\{ B_{k} \}_{k=0}^{d^{2}-1}$ in $\mathcal{L}(\mathcal{H})$, where $d \equiv \text{dim}\mathcal{H}$. It is particularly useful to pick one of the basis elements, e.g. $B_{0}$, as the ``desirable'' error (such as being proportional to the unit matrix). Consequently, the coefficient associated with this error component indicates how likely it is that the given Kraus operators of the noisy map will change the state of our quantum system in a ``desirable way''. An additional benefit of the $\chi$-matrix representation is that the effects of channel twirling are especially clear \cite{magesan2008gaining}, as ``diagonalization'' with respect to the generalized Pauli group. Therefore, the rest of the appendix is devoted to recalling the $\chi$-matrix representation in a self-contained way.

Recall that every CP map $\mathcal{Q}^{A\rightarrow B}$ admits a Kraus decomposition
\begin{equation}
    \mathcal{Q}(\cdot) = \sum_{i=1}^{K}Q_{i}(\cdot)Q^{\dagger}_{i} \;. \label{eqn:Kraus}
\end{equation}
in terms of Kraus operators $\{ Q_{i} \}_{i=1}^{K}$ satisfying $\sum_{i=1}^{K}Q_{i}^{\dagger}Q_{i} \leq I_{A}$, where the equality holds for TP maps. Let us consider a CP map $\mathcal{Q}^{A\rightarrow A} \equiv \mathcal{Q}$, where by denoting $d\equiv \text{dim}(\mathcal{H}^{A})$, we can decompose each of the Kraus operators $\{ Q_{i} \}_{i=1}^{K}$ as a linear combination of some orthonormal operator basis $\{ B_{k} \}_{k=0}^{d^{2}-1}$ in $\mathcal{L}(\mathcal{H}^{A})$, as follows
\begin{equation}
    Q_{i}=\sum_{k=0}^{d^{2}-1}\langle B_{k}, Q_{i}\rangle B_{k} \; , \label{eqn:basis_decom}
\end{equation}
where $\langle B_{k}, B_{l}\rangle = \delta_{kl}$, and $\langle \cdot \rangle $ being the Hilbert-Schmidt inner product in $\mathcal{L}(\mathcal{H}^{A})$. One could take $B_{k}\equiv B_{(m,n)}=|m\rangle \! \langle n|$ for $m,n= \{1, \cdots, d$\}, which is known as the standard basis in $\mathcal{L}(\mathcal{H}^{A})$. Substituting Eq.~\eqref{eqn:basis_decom} in the Kraus representation of $\mathcal{Q}$, we get
\begin{equation}
    \mathcal{Q}(\cdot) = \sum_{k=0}^{d^{2}-1}\sum_{l=0}^{d^{2}-1}\chi_{kl}^{\mathcal{Q}}B_{k}(\cdot)B^{\dagger}_{l} \; , \label{eqn:chi_repr}
\end{equation}
where
\begin{equation}
    \chi_{kl}^{\mathcal{Q}}\coloneqq \sum_{i=1}^{K}\langle B_{k}, Q_{i}\rangle \! \langle Q_{i}, B_{l}\rangle \; ,
\end{equation}
is called the $\chi$ matrix of the CP map $\mathcal{Q}$. It is easy to see that the $\chi$ matrix is a positive semi-definite matrix. This matrix has $d^{4}$ complex entries, corresponding to the matrix entries of the superoperator $\mathcal{Q}$ in the Liouville representation (see, e.g. \cite{carignan2019bounding, kimmel2014robust}), namely
\begin{equation}
    \hat{\mathcal{Q}}\coloneqq \sum_{k=0}^{d^{2}-1}\sum_{l=0}^{d^{2}-1}\chi_{kl}^{\mathcal{Q}} |B_{k}\rangle \rangle \! \langle \langle B_{l}| \; ,
\end{equation}
where $|B_{k}\rangle \rangle$ is the $d^{2}\times 1$ vector corresponding to the $d\times d$ matrix $B_{k}$. The number of independent entries of the $\chi$ matrix is reduced from $d^{4}$ to $d^{4}-d^{2}$ complex numbers if the CP map $\mathcal{Q}$ is also TP, since for each of the $d^{2}$ standard basis elements $|n\rangle \! \langle m|$ in $\mathcal{L}(\mathcal{H}^{A})$, the map $\mathcal{Q}$ must also preserve the trace, which leads to $d^{2}$ constraints.

In the context of QEC, it is convenient to choose our operator basis in $\mathcal{L}(\mathcal{H}^{A})$ such that $B_{0}$ indicates a ``desired effect'' on a quantum state. Here $B_{0}\equiv I/\sqrt{d}$ is desirable for QEC, but for other applications, $B_{0}$ could be chosen differently. Next, we write Eq.~\eqref{eqn:basis_decom} for a fixed $i=1, \cdots, K$, as 
\begin{equation}
    Q_{i}=\langle B_{0}, Q_{i}\rangle B_{0}+\sum_{k=1}^{d^{2}-1}\langle B_{k}, Q_{i}\rangle B_{k} \; .
\end{equation}
Then, by multiplying both sides on the left by $Q_{i}^{\dagger}$ and taking the trace, we arrive at
\begin{equation}
    \langle Q_{i}, Q_{i}\rangle =\vert \langle B_{0}, Q_{i}\rangle \vert ^{2}+\sum_{k=1}^{d^{2}-1}\vert \langle B_{k}, Q_{i}\rangle \vert ^{2} \; ,
\end{equation}
or equivalently,
\begin{equation}
    \frac{\vert \langle B_{0}, Q_{i}\rangle \vert^{2}}{\langle Q_{i}, Q_{i}\rangle}+\sum_{k=1}^{d^{2}-1}\frac{\vert \langle B_{k}, Q_{i}\rangle \vert^{2}}{\langle Q_{i}, Q_{i}\rangle}=1 \; .
\end{equation}
By denoting $q_{i}^{2}\coloneqq \langle Q_{i}, Q_{i}\rangle$, $\vert \cos(\phi_{i}) \vert \coloneqq \vert \langle B_{0}, Q_{i}\rangle \vert / q_{i}$, and $v_{i,k} \vert \sin(\phi_{i}) \vert \coloneqq \vert \langle B_{k}, Q_{i}\rangle \vert / q_{i}$ with some real weights $\{ v_{i, k} \}_{k=1}^{d^{2}-1}$ satisfying $\sum_{k=1}^{d^{2}-1}v^{2}_{i, k}=1$, we can rewrite the previous equation in a simple form
\begin{equation}
    \cos^{2}(\phi_{i})+\sin^{2}(\phi_{i})=1 \hspace{0.2cm} \text{for} \hspace{0.2cm} \text{for all} \hspace{0.1cm} i=1, \cdots, K \; ,
\end{equation}
where the angle $\phi_{i}$ indicates how ``close'' the error $Q_{i}$ is the the ``desirable error'' $B_{0}$. Note that, if $\mathcal{Q}$ is also TP, then $\sum_{i=1}^{K}q^{2}_{i}=\sum_{i=1}^{K}\langle Q_{i}, Q_{i}\rangle=\mathrm{Tr}[I]=d$.

Due to the unitary freedom of choosing the Kraus operators of any fixed quantum channel from its Steinspring dilation \cite{nielsen2002quantum}, the $\chi$ matrix is not uniquely determined. Therefore, using the phase freedom $Q_{i}\rightarrow Q_{i}e^{i\omega_{i}}$ (which is a special case of the unitary freedom mentioned above), we can always choose the phases $\omega_{i}$ for $i=1,\cdots, K$ such that all the inner products $\langle B_{0}, Q_{i}\rangle $ with the basis element $B_{0}$ are all non-negative. This means that we can pick $\phi_{i} \in [0, \pi/2]$. Finally, the additional phase in $\langle B_{k}, Q_{i}\rangle$ can always be placed in the vector $v_{k}$, which leads to the final decomposition
\begin{equation}
    Q_{i}=q_{i}\left(\cos(\phi_{i})B_{0}+\sin(\phi_{i})\sum_{k=1}^{d^{2}-1}v_{i,k}B_{k}\right) \; , \label{eqn:N_i}
\end{equation}
where $\{v_{i,k}\}_{k=1}^{d^{2}-1}$ are now complex numbers with $\sum_{k=1}^{d^{2}-1}\vert v_{i,k} \vert^{2}=1$. From here, we can easily compute the matrix element $\chi^{\mathcal{Q}}_{00}$, after taking $B_{0}=I/\sqrt{d}$, as
\begin{equation}
    \chi^{\mathcal{Q}}_{00}=\sum_{i=1}^{K}q_{i}^{2}\cos^{2}(\phi_{i})=\frac{1}{d}\sum_{i=1}^{K}\left \vert \mathrm{Tr}[Q_{i}] \right \vert^{2} \; , \label{eqn:chi_00}
\end{equation}
which is consistent with $\chi^{\mathcal{Q}}_{00}/d=F_{e}(\mathcal{Q}, I/d)$ \cite{schumacher1996quantum, schumacher1996sending}. Using Eq.~\eqref{eqn:N_i} to compute $Q^{\dagger}_{i}Q_{i}$, taking the trace, and using the orthonormality of the operator basis $\{ B_{k} \}_{k=1}^{d^{2}-1}$, we arrive at
\begin{equation}
    \sum_{i=1}^{K}\langle Q_{i}, Q_{i}\rangle = \sum_{i=1}^{K}q_{i}^{2} \leq d \; ,
\end{equation}
where the inequality follows from $\sum_{i}Q_{i}^{\dagger}Q_{i} \leq I$.
Combined with Eq.~\eqref{eqn:chi_00}, this implies that $ 0 \leq \chi_{00}^{\mathcal{Q}} \leq d$, or equivalently $0 \leq F_{e}(\mathcal{Q}) \leq 1$.

\section{Proof of Lemma~\ref{le:err_angle}}
\label{sec:Chi_upper}
Here we derive an upper bound on the matrix element $\chi_{00}^{\mathcal{S}\circ \mathcal{Q}}$ of the composite channel $\mathcal{S}\circ \mathcal{Q}$, given the corresponding $\chi$ matrix elements $\chi_{00}^{\mathcal{Q}}$ and $\chi_{00}^{\mathcal{S}}$ of the individual channels $\mathcal{Q}$ and $\mathcal{S}$, respectively. The technique used for the following derivation is based on \cite{carignan2019bounding}.

Given the Kraus operators $\{ Q_{i} \}_{i=1}^{K(\mathcal{Q})}$, $\{ S_{j} \}_{j=1}^{K(\mathcal{S})}$ of the individual channels $\mathcal{Q}$ and $\mathcal{S}$, the Kraus operators of the composite channel $\mathcal{S}\circ \mathcal{Q}$ are given by $\{ S_{j}Q_{i} \}_{i,j}$ for $i=1, \cdots, K(\mathcal{Q})$ and $j=1, \cdots, K(\mathcal{S})$. Therefore, by using Eq.~\eqref{eqn:N_i} for the individual Kraus operators and using the notation $\langle Q_{i}, Q_{i} \rangle = q_{i}$ and $\langle S_{i}, S_{i} \rangle = s_{i}$, we find for the Kraus operators of the composite channel
\begin{align}
    S_{j}Q_{i} &=s_{j}q_{i} \cos(\phi_{j}^{\mathcal{S}})\cos(\phi_{i}^{\mathcal{Q}})B_{0}^{2} \nonumber \\
    &+ s_{j}q_{i}\sin(\phi_{j}^{\mathcal{S}})\cos(\phi_{i}^{\mathcal{Q}}) \sum_{k=1}^{d^{2}-1}v^{\mathcal{S}}_{j,k}B_{k}B_{0} \nonumber \\ &+ s_{j}q_{i}\cos(\phi_{j}^{\mathcal{S}})\sin(\phi_{i}^{\mathcal{Q}}) \sum_{k=1}^{d^{2}-1}v^{\mathcal{Q}}_{i,k}B_{0}B_{k} \nonumber \\ &+ s_{j}q_{i}\sin(\phi_{j}^{\mathcal{S}})\sin(\phi_{i}^{\mathcal{Q}}) \sum_{k, k^{\prime}=1}^{d^{2}-1}v^{\mathcal{S}}_{j,k}v^{\mathcal{Q}}_{i,k^{\prime}}B_{k}B_{k^{\prime}} \; . 
\end{align}
By substituting into Eq.~\eqref{eqn:chi_00} and choosing the operator basis elements to be Hermitian (hence $\mathrm{Tr}[B_{k}B_{k^{\prime}}]=\langle B_{k}, B_{k^{\prime}} \rangle = \delta_{kk^{\prime}} $, for $k, k^{\prime}=0, 1, \cdots, n^{2}-1$), we arrive at
\begin{align}
    \chi_{00}^{\mathcal{S}\circ \mathcal{Q}} &=\frac{1}{d}\sum_{i,j}s_{j}^{2}q_{i}^{2}\vert \cos(\phi_{j}^{\mathcal{S}})\cos(\phi_{i}^{\mathcal{Q}}) \nonumber \\ &+(v^{\mathcal{S}}_{j}\bullet v^{\mathcal{Q}}_{i})\sin(\phi_{j}^{\mathcal{S}})\sin(\phi_{i}^{\mathcal{Q}}) \vert^{2} \; , \label{eqn:composite_chi}
\end{align}
where we have denoted by $v^{\mathcal{S}}_{j}\bullet v^{\mathcal{Q}}_{i}\equiv \sum_{k=1}^{d^{2}-1}v^{\mathcal{S}}_{j,k}v^{\mathcal{Q}}_{i,k}$, with $\{v^{\mathcal{Q}}_{i,k}\}_{k=0}^{d^{2}-1}$ and $\{v^{\mathcal{S}}_{j,k}\}_{k=0}^{d^{2}-1} \in \mathbb{C}^{d^{2}-1}$. By denoting $c_{ij}\equiv\cos(\phi_{j}^{\mathcal{S}})\cos(\phi_{i}^{\mathcal{Q}})$, $s_{ij}\equiv \sin(\phi_{j}^{\mathcal{S}})\sin(\phi_{i}^{\mathcal{Q}})$, and $v_{ij}\equiv v^{\mathcal{S}}_{j}\bullet v^{\mathcal{Q}}_{i}$, we can use the (forward and reversed) triangle inequality, as follows
\begin{equation}
    \vert \vert c_{ij} \vert - \vert v_{ij}s_{ij} \vert \vert \leq \vert c_{ij}+v_{ij}s_{ij} \vert \leq \vert c_{ij} \vert + \vert v_{ij}s_{ij} \vert \; . \label{eqn:triangle}
\end{equation}
By recalling that $c_{ij}$, $s_{ij} \ge 0$, since $\phi_{i}^{\mathcal{Q}}$, $\phi_{j}^{\mathcal{S}} \in [0, \pi/2]$, the second inequality yields $\vert c_{ij}+v_{ij}s_{ij} \vert \leq \vert c_{ij} \vert + \vert v_{ij}s_{ij} \vert = c_{ij} + \vert v_{ij} \vert s_{ij}$. If we assume that $\vert v_{ij} \vert\leq 1 $, then the inequality becomes $\vert c_{ij}+v_{ij}s_{ij} \vert \leq c_{ij} + s_{ij}$. By squaring both sides, we get
\begin{equation}
    \vert c_{ij}+v_{ij}s_{ij} \vert^{2} \leq  c_{ij}^{2} + s_{ij}^{2} + 2 c_{ij}s_{ij} \; ,
\end{equation}
where this inequality is saturated iff $v_{ij} =1$ for all $i=1,\cdots, K(\mathcal{Q})$ and $j=1, \cdots, K(\mathcal{S})$. Substituting back into Eq.~\eqref{eqn:composite_chi} and using the definitions of $\chi_{00}^{\mathcal{Q}}$ and $\chi_{00}^{\mathcal{S}}$ from Eq.~\eqref{eqn:chi_00}, as well as the fact that $\sum_{i}q_{i}^{2}=\sum_{j}s_{j}^{2}=d$ for CPTP maps $\mathcal{Q}$ and $\mathcal{S}$, we arrive at
\begin{align}
    \chi_{00}^{\mathcal{S}\circ \mathcal{Q}} &= \frac{1}{d}\sum_{i,j}s_{j}^{2}q_{i}^{2}\vert c_{ij}+v_{ij}s_{ij} \vert^{2} \\
    &\leq \frac{1}{d}\sum_{i,j}s_{j}^{2}q_{i}^{2}(c_{ij}^{2} + s_{ij}^{2} + 2 c_{ij}s_{ij}) \\
    & = \frac{1}{d} \left[\chi_{00}^{\mathcal{S}}\chi_{00}^{\mathcal{Q}}+(d-\chi_{00}^{\mathcal{S}})(d-\chi_{00}^{ \mathcal{Q}})\right] \nonumber \\
    & \hspace{0.2cm}+ \frac{2}{d}\sum_{i,j}s_{j}^{2}q_{i}^{2}c_{ij}s_{ij} \\ & = \frac{1}{d} \left[\chi_{00}^{\mathcal{S}}\chi_{00}^{\mathcal{Q}}+(d-\chi_{00}^{\mathcal{S}})(d-\chi_{00}^{ \mathcal{Q}})\right] \nonumber \\
    & \hspace{0.2cm} + \frac{2}{d}\left(\sum_{j}s_{j}^{2}\cos(\phi_{j}^{\mathcal{S}})\sin(\phi_{j}^{\mathcal{S}})\right)\times  \nonumber \\ & \hspace{0.4cm} \times \left(\sum_{i}q_{i}^{2}\cos(\phi_{i}^{\mathcal{Q}})\sin(\phi_{i}^{\mathcal{Q}})\right) \; . \label{eqn:cauchy_0}
\end{align}
Next, we use the Cauchy-Schwartz inequality
\begin{align}
    &\sum_{i}q_{i}^{2}\cos(\phi_{i}^{\mathcal{Q}})\sin(\phi_{i}^{\mathcal{Q}}) \\ & = \sum_{i}\left(q_{i}\cos(\phi_{i}^{\mathcal{Q}})\right) \left( q_{i}\sin(\phi_{i}^{\mathcal{Q}})\right) \\ & \leq \sqrt{\sum_{i} q^{2}_{i}\cos^{2}(\phi_{i}^{\mathcal{Q}})} \sqrt{\sum_{i} q^{2}_{i}\sin^{2}(\phi_{i}^{\mathcal{Q}})} \\ & 
    = \sqrt{\chi_{00}^{\mathcal{Q}}(d-\chi_{00}^{\mathcal{Q}})} \; , \label{eqn:cauchy_1}
\end{align}
where the Cauchy-Schwartz inequality is saturated when the vectors $\{q_{i}\cos(\phi_{i}^{\mathcal{Q}})\}_{i=1}^{K(\mathcal{Q})}$ and $\{q_{i}\sin(\phi_{i}^{\mathcal{Q}})\}_{i=1}^{K(\mathcal{Q})}$ are linearly dependent, therefore $\tan(\phi_{1}^{\mathcal{Q}})=\cdots=\tan(\phi_{K(\mathcal{Q})}^{\mathcal{Q}})$. Since $\phi_{i}^{\mathcal{Q}} \in [0, \pi/2]$ for all $ i=1, \cdots ,K(\mathcal{Q})$ and the function $\tan(x)$ is one-to-one in that region, it follows that the above inequality is saturated iff $\phi_{1}^{\mathcal{Q}}=\cdots=\phi_{K(\mathcal{Q})}^{\mathcal{Q}}$. The exact same argument for the channel $\mathcal{S}$ yields
\begin{equation}
    \sum_{j}s_{j}^{2}\cos(\phi_{j}^{\mathcal{S}})\sin(\phi_{j}^{\mathcal{S}})
    \leq \sqrt{\chi_{00}^{\mathcal{S}}(d-\chi_{00}^{\mathcal{S}})} \; , \label{eqn:cauchy_2}
\end{equation}
where the inequality is saturated iff $\phi_{1}^{\mathcal{S}}=\cdots=\phi_{K(\mathcal{S})}^{\mathcal{S}}$. Substituting Eqs.~\eqref{eqn:cauchy_1} and \eqref{eqn:cauchy_2} into Eq.~\eqref{eqn:cauchy_0}, we arrive at
\begin{align}
    d\chi_{00}^{\mathcal{S}\circ \mathcal{Q}} &\leq \chi_{00}^{\mathcal{S}}\chi_{00}^{\mathcal{Q}}+(d-\chi_{00}^{\mathcal{S}})(d-\chi_{00}^{ \mathcal{Q}}) \nonumber \\ &+2\sqrt{\chi_{00}^{\mathcal{S}}\chi_{00}^{\mathcal{Q}}(d-\chi_{00}^{\mathcal{S}})(d-\chi_{00}^{ \mathcal{Q}})} \; ,
\end{align}
or equivalently,
\begin{equation}
    \sqrt{d\chi_{00}^{\mathcal{S}\circ \mathcal{Q}}} \leq \sqrt{\chi_{00}^{\mathcal{S}}}\sqrt{\chi_{00}^{\mathcal{Q}}}+\sqrt{d-\chi_{00}^{\mathcal{S}}}\sqrt{d-\chi_{00}^{\mathcal{Q}}} \; .
\end{equation}
Dividing both sides by $d$ and redefining $\chi^{\mathcal{Q}}_{00}/d \equiv \cos^{2}(\delta^{\mathcal{Q}})$ for $\delta^{\mathcal{Q}} \in [0, \pi/2]$, as suggested in Eq.~\eqref{eqn:err_angle}, we arrive at
\begin{equation}
    \cos(\delta^{\mathcal{S}\circ \mathcal{Q}}) \leq \cos(\delta^{\mathcal{S}})\cos(\delta^{\mathcal{Q}})+\sin(\delta^{\mathcal{S}})\sin(\delta^{\mathcal{Q}}) \; .
\end{equation}
Using a triangle identity to simplify the right hand side, we get $\cos(\delta^{\mathcal{S}\circ \mathcal{Q}}) \leq \cos(\delta^{\mathcal{S}}-\delta^{\mathcal{Q}}) $, i.e. 
\begin{equation}
    \delta^{\mathcal{S}\circ \mathcal{Q}} \ge \vert \delta^{\mathcal{S}}-\delta^{\mathcal{Q}} \vert \; ,
\end{equation}
which allows $\delta^{\mathcal{S}\circ \mathcal{Q}} \in [0, \pi/2]$ given $\delta^{\mathcal{S}}$, $\delta^{\mathcal{Q}} \in [0, \pi/2]$. In other words, given $\chi^{\mathcal{Q}}_{00}$ and $\chi^{\mathcal{S}}_{00}$, the composite channel $\chi$-matrix element $\chi^{\mathcal{S}\circ \mathcal{Q}}_{00}$ is bounded from above by
\begin{equation}
    \frac{\chi^{\mathcal{S}\circ \mathcal{Q}}_{00}}{d} \leq \cos^{2} \left( \arccos{\sqrt{\frac{\chi^{\mathcal{S}}_{00}}{d}}}-\arccos{\sqrt{\frac{\chi^{\mathcal{Q}}_{00}}{d}}} \right) \; .
\end{equation}

\section{Proof of Theorem~\ref{th:spec_multi}}
\label{apx:multi_cycle_proof}
We denote by $F_{e}\equiv F_{e}(\mathcal{R}_{\theta}\circ \mathcal{N}_{\theta})$, $\hat{F}_{e}\equiv \hat{F}_{e}(\mathcal{R}_{\hat{\theta}}\circ \mathcal{N}_{\theta})$, and $\Delta F_{e} \equiv F_{e}-\hat{F}_{e}$, and then do the simple manipulation
\begin{align}
    \arccos&{\sqrt{F_{e}(\mathcal{R}_{\hat{\theta}}\circ \mathcal{N}_{\theta})}}=\arccos{\sqrt{F_{e}-\Delta F_{e}}} \\ & =\arccos{\left\{ \sqrt{F_{e}}\sqrt{1-\frac{\Delta F_{e}}{F_{e}}} \right\}} \\ & =\arccos{\left\{ \sqrt{F_{e}}\left[1-\frac{\Delta F_{e}}{2 F_{e}}\right]  + \sqrt{F_{e}} R_{1} \left(\frac{\Delta F_{e}}{F_{e}}\right) \right\}} \; ,
\end{align}
where
\begin{equation}
    R_{1}(x) \equiv \frac{1}{3} (1-x_{0})^{-3/2}x^{2} \; ,
\end{equation}
is the Lagrange remainder term of the Taylor expansion of $\sqrt{1-x}$, and $x_{0} \in [0, x]$ is a constant. This yields
\begin{equation}
    \cos{\hat{\delta}} = \sqrt{F_{e}}-\frac{\Delta F_{e}}{2\sqrt{F_{e}}}+\sqrt{F_{e}} R_{1} \left(\frac{\Delta F_{e}}{F_{e}}\right) \; . \label{eqn:cosdelta_deriv_1}
\end{equation}
By denoting $\delta \equiv \arccos{\sqrt{F_{e}}}$, we find the Taylor series expansion of $\hat{\delta}$ with respect to $x\equiv \cos{\delta}-\cos{\hat{\delta}}$, as follows
\begin{equation}
    \hat{\delta}=\arccos{(\cos{\delta}-x)}=\delta+\frac{x}{\sin{\delta}}+R_{2}(x) \; ,
\end{equation}
where 
\begin{equation}
    R_{2}(x)=-\frac{\cos{\delta}-x_{0}}{2(1-(\cos{\delta}-x_{0})^{2})^{3/2}}x^{2} \; ,
\end{equation}
is the Lagrange form of the remainder term, and $x_{0} \in [0, x]$ is a constant. Substituting for $x$ using Eq.~\eqref{eqn:cosdelta_deriv_1} yields
\begin{align}
    \hat{\delta}-\delta &=\frac{x}{\sin{\delta}}+R_{2}(x) \\
    &=\frac{\frac{\Delta F_{e}}{2\sqrt{F_{e}}}-\sqrt{F_{e}} R_{1} \left(\frac{\Delta F_{e}}{F_{e}}\right)}{\sin{\delta}}+R_{2}(x) \\
    &=\frac{\Delta F_{e}}{2\sqrt{F_{e}(1-F_{e})}}+r\left( \frac{\Delta F_{e}}{F_{e}} \right)\; ,
\end{align}
where 
\begin{equation}
    r(x)=\sqrt{\frac{F_{e}}{1-F_{e}}}R_{1}(x)+R_{2}\left( \frac{\sqrt{F_{e}}}{2}x-\sqrt{F_{e}}R_{1}(x) \right) \; ,
\end{equation}
is the accumulative remainder term.

Next, we rewrite Theorem~\ref{th:rec_ent_fid} as a recurrence inequality, where the contribution of the spectator system in each time step is clearly separated.
\begin{align}
    F_{e}^{1\rightarrow n} &\leq \cos^{2}\left( \arccos{\sqrt{F_{e}^{1\rightarrow (n-1)}}}-\arccos{\sqrt{\hat{F}^{n}_{e}}} \right) \\ 
    & \equiv \cos^{2}\left( \delta^{1\rightarrow (n-1)}-\hat{\delta}^{n} \right) \\
    & =\cos^{2}\left( \delta^{1\rightarrow (n-1)}-\delta^{n}-(\hat{\delta}^{n}-\delta^{n})  \right) \; .
\end{align}
By denoting $a=\delta^{1\rightarrow (n-1)}-\delta^{n}$ and $b=\hat{\delta}^{n}-\delta^{n}$, and using the trigonometric identities $\cos^{2}(a-b)=\frac{1}{2}+\frac{1}{2}\cos{2(a-b)}$ and $\cos{(a-b)}=\cos{(a)} \cos{(b)}+\sin{(a)}\sin{(b)}$, we arrive at
\begin{align}
    & \cos^{2}(a-b) \\ 
    & = \frac{1}{2}+\frac{1}{2}\left(\cos{2a}\cos{2b}+\sin{2a}\sin{2b} \right) \\
    & = \cos^{2}(a)-\frac{1}{2}\cos{2a}(1-\cos{2b})+\frac{1}{2}\sin{2a}\sin{2b} \\
    & = \cos^{2}(a)+b\sin(2a)-\frac{1}{2}(\cos{2a})R_{3}(2b) \nonumber \\ 
    & \hspace{0.2cm}+\frac{1}{2}(\sin{2a})R_{4}(2b) \\
    & \equiv \cos^{2}(a)+b\sin(2a) + r^{\prime}(b)\; ,
\end{align}
where we have used the Taylor expansions of the sine and cosine functions to the first order of $b$, and denoted by $R_{3}(x)$ and $R_{4}(x)$ with their remainders in the Lagrange form, respectively.

Consequently, we have for Theorem~\ref{th:rec_ent_fid} the separation 
\begin{align}
    F_{e}&^{1\rightarrow n} \leq \cos^{2}\left( \delta^{1\rightarrow (n-1)}-\delta^{n} \right) \nonumber \\ 
    &+(\hat{\delta}^{n}-\delta^{n})\sin\left(2( \delta^{1\rightarrow (n-1)}-\delta^{n}) \right) +r^{\prime}(\hat{\delta}^{n}-\delta^{n}) \; ,
\end{align}
where the first term is the interference between the previous $n-1$ cycles and the $n$-th cycle error angles with perfect knowledge at the $n$-th step. On the other hand, the second term shows the contribution of the lack of knowledge into the recurrence inequality at the $n$-th timestep, for a fixed error angle $\delta^{1\rightarrow (n-1)}$, as
\begin{equation}
    \frac{\Delta F_{e}^{n}}{2\sqrt{F^{n}_{e}(1-F^{n}_{e})}}\sin\left( 2\delta^{1\rightarrow (n-1)}-2\arccos{\sqrt{F^{n}_{e}}} \right) \; .
\end{equation}
Note that the sign of this contribution depends on the difference between the error angles of the previous $n-1$ cycles and the $n$-th cycle. Taking the expectation with respect to the probability distribution $p_{X}(x_{n}|\theta_{n})$ and using Theorem~\ref{th:spec_QFI}, we arrive at the following:

\end{document}